\documentclass[11pt,english]{article}
\usepackage{amsmath,amssymb,amsthm}
\usepackage{multicol}
\usepackage[margin=1in]{geometry}
\usepackage{graphicx,color}
\usepackage{enumitem}
\usepackage{fullpage}
\usepackage[noblocks]{authblk}
\usepackage{tcolorbox}
\usepackage{babel}
\usepackage{wrapfig}
\usepackage{MnSymbol}
\usepackage{mdframed}
\usepackage{mathtools}
\usepackage{floatrow}
\usepackage{multirow}
\usepackage{palatino}

\usepackage[ruled,linesnumbered,vlined]{algorithm2e}
\usepackage{tablefootnote}
\usepackage{thm-restate}
\usepackage{bbding}
\usepackage{pifont}
\usepackage{nicefrac}

\usepackage{dsfont}

\newcounter{magicrownumbers}
\newcommand\rownumber{\stepcounter{magicrownumbers}\arabic{magicrownumbers}}

\definecolor{Darkblue}{rgb}{0,0,0.4}
\definecolor{Brown}{cmyk}{0,0.61,1.,0.60}
\definecolor{Purple}{cmyk}{0.45,0.86,0,0}
\definecolor{Darkgreen}{rgb}{0.133,0.543,0.133}

\usepackage[colorlinks,linkcolor=Darkblue,filecolor=blue,citecolor=blue,urlcolor=Darkblue,pagebackref]{hyperref}
\usepackage[nameinlink]{cleveref}

\usepackage[colorinlistoftodos,prependcaption,textsize=tiny]{todonotes}

\newcommand{\ddim}{\mbox{\rm ddim}}
\newcommand{\sddim}{\footnotesize \mbox{\rm ddim}}
\newcommand{\Texp}{\mathsf{Texp}}
\newcommand{\est}{{\rm est}}

\newcommand{\namedref}[2]{\hyperref[#2]{#1~\ref*{#2}}}
\newcommand{\propref}[1]{\hyperref[#1]{property~(\ref*{#1})}}

\newtheorem*{theorem*}{Theorem}
\newtheorem{theorem}{Theorem}
\newtheorem{lemma}{Lemma}
\newtheorem{definition}{Definition}
\newtheorem{claim}{Claim}
\newtheorem{observation}{Observation}
\newtheorem{corollary}{Corollary}
\newtheorem{proposition}{Proposition}

\newtheorem*{question*}{Question}

\newtheorem{remark}{Remark}

\newtheorem*{conjecture*}{Conjecture}

\newcommand{\eps}{\varepsilon}
\renewcommand{\epsilon}{\varepsilon}
\newcommand{\diam}{\mathrm{diam}}

\newcommand{\lca}{\mathrm{lca}}
\newcommand{\MST}{\mathrm{MST}}
\newcommand{\opt}{\mathrm{OPT}}
\newcommand{\inw}{\mathrm{INW}}
\newcommand{\HPI}{\mathrm{HPI}}

\newcommand{\cost}{\mathrm{Cost}}
\newcommand{\alg}{\mathrm{ALG}}
\newcommand{\vir}{\xi}
\newcommand{\Vir}{\Xi}
\newcommand{\N}{\mathbb{N}}
\newcommand{\R}{\mathbb{R}}
\newcommand{\Z}{\mathbb{Z}}

\newcommand{\old}[1]{{}}

\newcommand{\mst}{\mathrm{MST}}

\newcommand{\poly}{\mathrm{poly}}

\newcommand{\parent}{\mathrm{parent}}

\newcommand{\cC}{\mathcal{C}}
\newcommand{\cD}{\mathcal{D}}

\newcommand{\cF}{\mathcal{F}}

\newcommand{\cI}{\mathcal{I}}

\newcommand{\cM}{\mathcal{M}}

\newcommand{\cP}{\mathcal{P}}
\newcommand{\cQ}{\mathcal{Q}}

\newcommand{\E}{\mathbb{E}}
\newcommand{\etal}{{et al.\ \xspace}}

\hypersetup{
    colorlinks=true,
    linkcolor=teal,
    filecolor=violet,
    citecolor=violet,
    urlcolor=cyan,
    pdftitle={Overleaf Example},
    pdfpagemode=FullScreen,
}

\newcommand{\nonblind}[1]{}

\title{Online Duet between Metric Embeddings \\ and Minimum-Weight Perfect Matchings}

\date{}

\author{Sujoy Bhore\thanks{Department of Computer Science \& Engineering, Indian Institute of Technology Bombay, Mumbai, India.\\ Email: sujoy@cse.iitb.ac.in}
\quad
Arnold Filtser\thanks{Bar Ilan University, Ramat Gan, Israel. Email: arnold273@gmail.com. This research was supported by the Israel Science Foundation (grant No. 1042/22).}
\quad
Csaba D. T\'oth\thanks{Department of Mathematics, California State University Northridge, Los Angeles, CA; and Department of Computer Science, Tufts University, Medford, MA, USA. Email: csaba.toth@csun.edu. Research on this paper was supported, in part, by the NSF award DMS-2154347. }
}

\begin{document}
\maketitle

\begin{abstract}
Low-distortional metric embeddings are a crucial component in the modern algorithmic toolkit. In an online metric embedding, points arrive sequentially and the goal is to embed them into a simple space irrevocably, while minimizing the distortion.
Our first result is a deterministic online embedding of a general metric into Euclidean space with distortion $O(\log n)\cdot\min\{\sqrt{\log\Phi},\sqrt{n}\}$ (or, $O(d)\cdot\min\{\sqrt{\log\Phi},\sqrt{n}\}$ if the metric has doubling dimension $d$), solving affirmatively a conjecture by Newman and Rabinovich (2020), and quadratically improving the dependence on the aspect ratio $\Phi$ from Indyk et al.\ (2010).
Our second result is a stochastic embedding of a metric space into trees with expected distortion $O(d\cdot \log\Phi)$, generalizing previous results (Indyk et al.\ (2010), Bartal et al.\ (2020)).

Next, we study the problem of \emph{online minimum-weight perfect matching} (\emph{MWPM}).
Here a sequence of $2n$ points $s_1,\ldots s_{2n}$ in a metric space arrive in pairs, and one has to maintain a perfect matching on the first $2i$ points $S_i=\{s_1,\ldots s_{2i}\}$. We allow recourse (as otherwise the order of arrival determines the matching). The goal is to return a perfect matching that approximates the \emph{minimum-weight} perfect matching on $S_i$, while minimizing the recourse.
Online matchings are among the most studied online problems, however, there is no previous work on online MWPM. One potential reason for this is that online MWPM is drastically non-monotone, which makes online optimization highly challenging.
Our third result is a randomized algorithm with competitive ratio $O(d\cdot \log \Phi)$ and recourse $O(\log \Phi)$ against an oblivious adversary, this result is obtained via our new stochastic online embedding.
Our fourth result is a deterministic algorithm that works against an adaptive adversary, using $O(\log^2 n)$ recourse, and maintains a matching of total weight at most $O(\log n)$ times the weight of the MST, i.e., a matching of lightness $O(\log n)$.
We complement our upper bounds with a strategy for an oblivious adversary that, with recourse $r$, establishes a lower bound of $\Omega(\frac{\log n}{r \log r})$ for both competitive ratio as well as lightness.
\end{abstract}

\newpage

\setcounter{tocdepth}{2}
	\tableofcontents

\newpage	
\setcounter{page}{1}

\section{Introduction}\label{sec:intro}

The traditional model of algorithms design solves problems where the entire input is given in advance.
In contrast, online algorithms work under conditions of uncertainty, gradually receiving an input sequence $\sigma = \sigma_1,\sigma_2,\ldots, \sigma_n$ (where $\sigma_i$ is presented at step $i$).
The algorithm has to serve them in the order of occurrence, where the decisions are irrevocable, and without prior knowledge of subsequent terms of the input.
The objective is to optimize the total cost paid on the entire sequence $\sigma$. The performance of an online algorithm $\alg$ is measured using competitive analysis, where $\alg$ is compared to an optimal offline algorithm that knows the entire sequence in advance and can provide the solution with optimum cost~\cite[Ch.~1]{BY98}.
The two most central adversarial models in online algorithms are adaptive and oblivious.
In the \emph{adaptive adversary} model, the sequence of arriving points is determined ``on the fly'', and may depend on the previous decisions
made by the algorithm. This is a restrictive model, and in particular, randomization is not helpful in this model.
An algorithm is \emph{$k$-competitive against an adaptive adversary} if, for every sequence $\sigma$ of requests, the cost of the algorithm is at most $k$ times the optimal offline solution.\footnote{This paper is focused on minimization problems. In a maximization problem, an algorithm is $k$-competitive if for every sequence $\sigma$, the cost of the algorithm is at least a $k$ fraction of the optimal offline solution.\label{foot:maxProblem}}
An \emph{oblivious adversary} assumes that the input sequence is determined in advance (however, unknown to the algorithm). Here randomization can be useful.
A randomized online algorithm is \emph{$k$-competitive against an oblivious adversary} if, for every sequence $\sigma$ of requests, the expected cost of the algorithm is at most $k$ times the optimal offline solution.$^{\ref{foot:maxProblem}}$

Online problems arise in various areas of computer science, such as scheduling, network optimization, data structures, resource management in operating systems, etc.; see~\cite{buchbinder2009design, KarpVV90, albers2003online, boyar2017online}. Some preeminent examples of online problems are
$k$-server~\cite{BubeckCR23}, job scheduling~\cite{lattanzi2020online}, routing~\cite{aspnes1997line}, load balancing~\cite{azar2005line}, among many others.

One of the most fundamental and well-studied problems in the online algorithms world is online matching. Starting with the seminal paper by Karp, Vazirani, and Vazirani~\cite{KarpVV90}, a large body of work on ``online matchings'' is devoted to the \emph{online bipartite matching (server-client model)} problem, where one side of the bipartite graph (\emph{servers}) is fixed and the vertices of the other side (\emph{clients}) are revealed one at a time. In the classical variant~\cite{KarpVV90}, the objective is to maintain a maximum matching (not necessarily a perfect matching) in an unweighted graph. Since then, numerous variants of this problem have been studied, see e.g. \cite{AS22,buchbinder2009design, devanur2009adwords, gamlath2019online, goel2008online, kalyanasundaram1998line, mehta2007adwords, shmoys1995scheduling}. In a metric variant of the problem, the servers and clients are points in a metric space, clients arrive one-by-one, and the objective is to maintain a minimum-cost matching of clients to servers~\cite{BansalBGN14,GKS20,GuptaL12,KhullerMV94,KalyanasundaramP93,megow2020online,MeyersonNP06,Raghvendra18}.

\smallskip\noindent\textbf{Online Minimum-Weight Perfect Matching.}
In this paper, we study online \emph{minimum-weight} perfect matchings (MWPM) in metric spaces.
Points $s_1,\ldots s_{2n}$ arrive sequentially from a metric space $(X,d_X)$ (unknown in advance). For each new point $s_i$, we are given the distances to all previous points: $\{d_X(s_i,s_j)\}_{j=1}^{i-1}$.
\footnote{Equivalently, there is an underlying complete weighted graph $G=(V,E,w)$ with weight respecting triangle inequality. For each new vertex, we receive the edges to all previously arrived vertices.}
Denote by $S_i=\{s_1,\ldots, s_{2i}\}$ the set of the first $2i$ points. The goal is to maintain a perfect matching $M_i$ on $S_i$ such that the difference between $M_i$ and $M_{i+1}$ is bounded by a constant or a polylogarithmic function of $n$, and the weight of $M_i$ is as small as possible.

A standard \emph{online} algorithm can add edges to the matching, but the decisions are irrevocable, and therefore no edge is ever deleted. In this setting, the matching is completely determined by the order of points: $M=\{\{s_{2i-1},s_{2i}\}: i=1,\ldots, n\}$, and the weight of $M$ may be arbitrarily far from the optimum; see an example in~\Cref{fig:examples}~(a).  For this reason, we allow \emph{recourse} $r$: the online algorithm has to maintain a perfect matching on $S_i$, and in each step, it can delete up to $r$ edges. Our primary focus is the trade-off between recourse and the weight of the matching $M_i$.

\begin{figure}[t]
	\centering
	\includegraphics[width=1\textwidth]{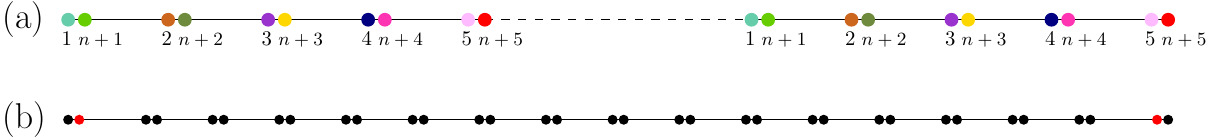}
	\caption{\small $(a)$ Example where the weight of an online matching is arbitrarily far form optimum assuming irrevocable decisions (i.e., no recourse). The metric is the real line. We first receive the pairs $\{i,W+i\}_{i=1}^{n}$, and then the pairs $\{i+\eps,W+i+\eps\}_{i=1}^{n}$, for sufficiently small $\eps$ and large $W$. The weight of the online perfect matching (specified in the illustration) is $2n\cdot W$, while the cost of the optimal perfect matching is $2\eps\cdot n$.\newline
	$(b)$ Example of the drastic non-monotonicity of minimum-weight perfect matching. The metric is the real line, where each point $\{1,2\dots,2n\}$ appears twice, while the points $\{0,2n+1\}$ appear once. Then the weight of a perfect matching is $2n+1$. After introducing the pair $\{0,2n+1\}$ (red in the figure), the weight of the perfect matching drops to $0$.}
	\label{fig:examples}
\end{figure}

Surprisingly, even though matchings are one of the most meticulously studied  online problems,
essentially no previous results were known for the online minimum-weight perfect matching problem with recourse.\footnote{The related problem of online minimum-weight perfect matching with delays \cite{AzarCK17,DU23,EmekKW16} has been studied previously; see \Cref{sec:related-work}.}
Note that the server-client model based online algorithms are significantly different from ours, and in general are not helpful for our problem.
One natural difficulty is that in contrast to other classical optimization problems, e.g., MST or TSP, the minimum weight of perfect matching is drastically non-monotone:\footnote{Due to the triangle inequality, the cost of an optimal TSP tour can only increase as new points arrive. The cost of an MST could decrease after additional points arrive, but this could happen up to at most a factor of $2$ (due to the fact that the MST is a $2$ approximation of the minimum Steiner tree).} it can decrease from a large weight to $0$ after introducing a single new pair! (See \Cref{fig:examples} (b) for an illustration.)
This non-monotonicity is the major bottleneck for maintaining a good approximation with limited recourse. We further discuss related classical online optimization problems in \Cref{sec:related-work}.

For online MST, for example, Gu \etal  \cite{GuG016} achieve a competitive ratio of $2^{O(k)}$ with a single recourse for every $k$ new points (that is fractional recourse, on average).
However, for MWPM we show (\Cref{pp:one}) that there is no competitive online algorithm if we are allowed to use a single recourse per vertex pair, which already holds for a sequence of $8$ points on a real line, even if the sequence is known in advance.
That is to say, online MWPM is a much more challenging problem than online MST.

\smallskip\noindent\textbf{Online Metric Embeddings.}
Low-distortion metric embeddings are a crucial component in the modern
algorithmic toolkit with applications in approximation algorithms~\cite{LLR95}, distributed algorithms~\cite{KKMPT12}, online algorithms~\cite{BBMN15}, and many more. A metric embedding is a map  $f: X \to Y$ between the points of two metric spaces $(X,d_X)$ and $(Y,d_Y)$. The \emph{contraction} and
\emph{expansion} of the map $f$ are the
largest $\alpha,\beta>0$,
respectively, such that for every pair $x,y\in X$,
\[
\alpha\cdot d_{X}(x,y)\leq d_{Y}(f(x),f(y))\leq \beta^{-1}\cdot d_{X}(x,y)~.
\]
The \emph{distortion} of the map is then 
$\alpha\cdot \beta$.
The embedding is \emph{non-contractive} (\emph{non-expansive}) if $\rho=1$ ($t=1$). Such an embedding is also called \emph{dominating}.
Arguably, in the TCS community, the two most celebrated metric embeddings are the following: (1) every $n$-point metric space embeds into Euclidean space $\ell_2$ with distortion $O(\log n)$ \cite{Bourgain85}, and (2)  every $n$-point metric space stochastically embeds into a distribution over dominating tree metrics (in fact ultrametrics\footnote{An ultrametric is a metric space with the strong triangle inequality: $\forall x,y,z,~ d_X(x,y)\le\max\{d_X(x,z),d_X(z,y)\}$. In particular, ultrametric can be represented as the shortest path metric of a tree graph. See \Cref{def:HST}.\label{foot:Ultrametric}}) with expected distortion $O(\log n)$  \cite{FRT04} (see also \cite{Bartal96,Bartal98,Bartal04}).
Specifically, there is a distribution $\cD$ over pairs $(f,U)$, where $U$ is an ultrametric, and $f:X\rightarrow U$ is a dominating embedding, such that for all $x,y\in X$, we have
$\mathbb{E}_{(f,U)\sim{\cal D}}\left[d_{U}(f(x),f(y)\right]\le O(\log n)\cdot d_{X}(x,y)$.

While \cite{Bourgain85,FRT04} enjoined tremendous success and have numerous applications, they require to know the
metric space $(X,d_X)$ in advance, and hence cannot be used in an online
fashion where the points are revealed one by one. In this paper, we first focus on online metric embeddings:
\begin{definition}[Online Embedding]\label{def:OE}
   An \emph{online embedding} of a sequence of points $x_1,\dots,x_k$ from a metric space $(X,d_X)$ into a metric space $(Y,d_Y)$ is a sequence of embeddings $f_1,\dots,f_k$ such that for every $i$, $f_i$ is a map from $\{x_1,\dots,x_i\}$ to $Y$, and $f_{i+1}$ extends $f_i$ (i.e., $f_i(x_j)=f_{i+1}(x_j)$ for $j\le i$).
    The embedding has \emph{expansion} $\alpha=\max_{i,j\le k}\frac{d_Y(f_k(x_i),f_k(x_j))}{d_X(x_i,x_j)}$  and \emph{contraction} $\beta=\max_{i,j\le k}\frac{d_{X}(x_{i},x_{j})}{d_{Y}(f_{k}(x_{i}),f_{k}(x_{j}))}$.
    The \emph{distortion} of the online embedding is $\alpha\cdot\beta$.
    If $\beta\le1$, we say that the embedding is \emph{dominating}.
	
    A \emph{stochastic online embedding} is a distribution $\cD$ over online embeddings.
   It has expected expansion  $\alpha=\max_{i,j\le k}\frac{\E[d_Y(f_k(x_i),f_k(x_j))]}{d_X(x_i,x_j)}$ and expected contraction $\beta=\max_{i,j\le k}\frac{d_X(x_i,x_j)}{\E[d_Y(f_k(x_i),f_k(x_j))]}$.
	A stochastic embedding is \emph{dominating} (or \emph{non-contractive}) if all the online embeddings in the support of the distribution are dominating (note that it is possible that $\beta\le 1$ and yet the embedding is not dominating).
	The expected distortion of a dominating stochastic embedding is $\alpha$.
	Similarly, a stochastic embedding is \emph{non-expansive} if, for every pair $x,y\in X$ and embedding $f$ from the distribution, $d_M(f(x_i),f(x_j))\le d_X(x_i,x_j)$.
	The expected distortion of a non-expansive stochastic embedding is $\beta$.
	In an online embedding algorithm, the embedding $f_i$ depends only on $\{x_1,\dots,x_i\}$.
\end{definition}

For stochastic online embeddings, the sequence of points should be fixed in advance (but unknown to the algorithm). That is, a deterministic online embedding can be used against an adaptive adversary, while stochastic online embedding can be used only against an oblivious adversary.

Indyk, Magen, Sidiropoulos, and Zouzias \cite{IMSZ10} (see also \cite{ERW10}) observed that Bartal's original embedding~\cite{Bartal96} can be used in an online fashion to produce a dominating stochastic embedding into trees (ultrametrics) with expected distortion $O(\log n\cdot\log\Phi)$,
where $\Phi=\frac{\max_{x,y\in X}d_X(x,y)}{\min_{x,y\in X}d_X(x,y)}$ is the \emph{aspect ratio} (a.k.a.\ \emph{spread}). Their original embedding had the caveat that the number of metric points $n$ and the aspect ratio $\Phi$ must be known in advance. Later, Bartal, Fandina, and Umboh~\cite{BFU20} removed these restrictions. They also provided an $\Omega\left(\frac{\log n\cdot\log\Phi}{\log\log n}\right)$ lower bound, showing this distortion to be tight up to second order terms.
For the case where the input metric space $(X,d_X)$ has doubling dimension\footnote{A metric $(X, d)$ has doubling dimension $\ddim$ if every ball of radius $2r$ can be covered by $2^{{\tiny {\rm ddim}}}$ balls of radius $r$.\label{foot:doubling}} $\ddim$ (known in advance), Indyk \etal \cite{IMSZ10} constructed a dominating stochastic online embedding into ultrametrics with expected distortion $2^{O(\sddim)}\cdot\log\Phi$.

In an attempt to construct an online version of Bourgain's embedding \cite{Bourgain85}, Indyk \etal \cite{IMSZ10} constructed online dominating stochastic embedding of an arbitrary $n$-point metric space into the Euclidean space $\ell_2$ with expected distortion $O(\log n\cdot\sqrt{\log\Phi})$ (again, $n$ and $\Phi$ must be known in advance), or with expected distortion $2^{O(\sddim)}\cdot\log\Phi$ for the case where the metric space has doubling dimension $\ddim$ (known in advance).
Newman and Rabinovich \cite{NR20} showed that every deterministic embedding into $\ell_2$ must have distortion $\Omega(\min\{\sqrt{n},\sqrt{\log\Phi}\})$. They conjectured\footnote{The conjecture appears in the full \href{https://arxiv.org/abs/2303.15945}{arXiv version}
	 (in the conference version, it is stated as an open problem).} that a similar upper bound holds:
\begin{conjecture*}[\cite{NR20}]
	Every sequence of $n$ points in a metric space received in an online fashion can be deterministically embedded into Euclidean space $\ell_2$ with distortion $\poly(n)$.
\end{conjecture*}

\subsection{Our Results}
The results in this paper are twofold: We study both online minimum-weight perfect matchings and online metric embeddings. The connections between them are illustrated in \Cref{fig:concepts}.
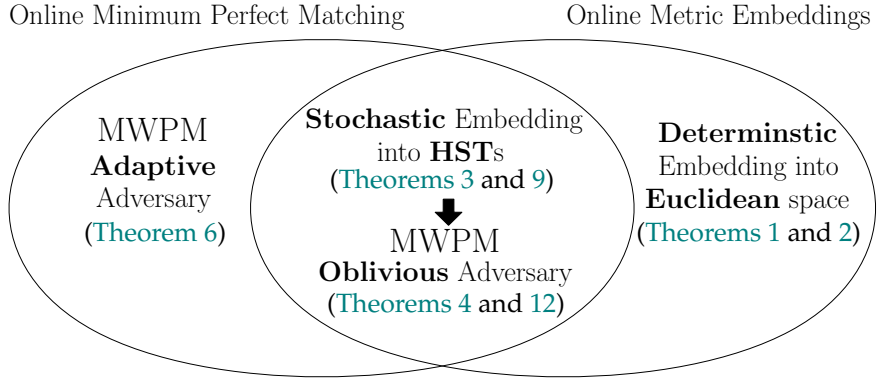
\begin{figure}[t]
	\centering
	\begin{picture}(325,170)
		\put(0,0){\includegraphics[width=0.7\textwidth]{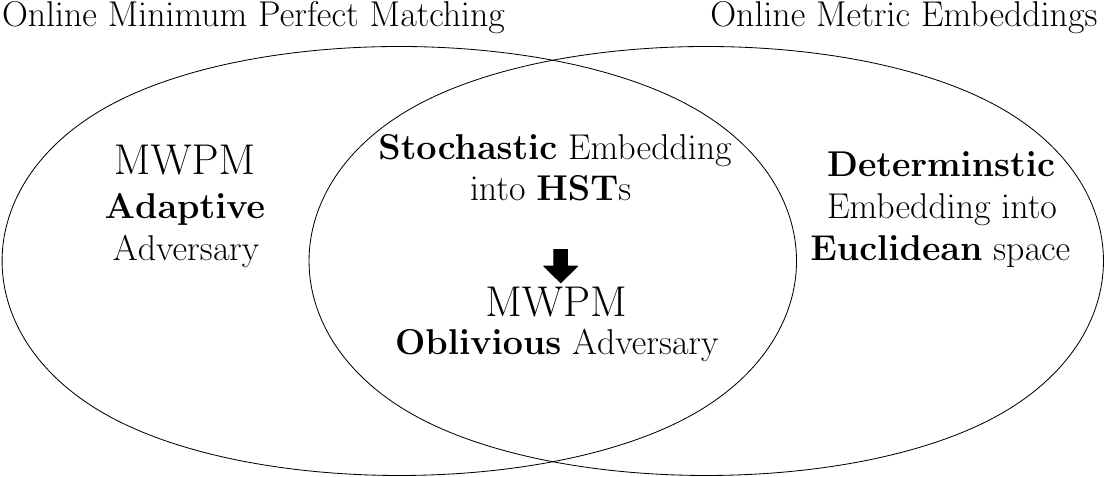}}
		\put(29,52){\small(\Cref{thm:Tree})}
		\put(122,72){\small(\Cref{thm:DoublingOnlineEmbedding,thm:EuclidineOnlineEmbeddingIntoHST})}
		\put(121,25){\small(\Cref{thm:ObliviousGeneralMetric,thm:ObliviousEuclideanMetric})}
		\put(238,52){\small(\Cref{thm:DoublingOnlineEmbeddingtoEuclidean,thm:LBEuclideanDoubling})}
	\end{picture}
	\caption{\footnotesize{A ``duet'' between online metric embeddings and minimum-weight perfect matchings: a Venn diagram of the relationship between the various results in this paper.}}
	\label{fig:concepts}
\end{figure}

\subsubsection{Metric Embeddings}

\begin{table}[]	
	\begin{tabular}{|l|l|c|l|l|l|l|}
		\hline
					&Input space          		 & Host Space  & Distortion                                 & Reference                                  & {\small Deter?}  & Pri.Kno.    \\ \hline
		\rownumber. & General              		 & \multirow{5}{*}{{\LARGE$\ell_2$}}    & $O(\log n\cdot \sqrt{\log\Phi})$      		& \cite{IMSZ10}                           	 &         &$n,\Phi$        \\ \cline{1-2} \cline{4-7}
		\rownumber. & 		General              &     & $\Omega(\min\{\sqrt{n},\sqrt{\log\Phi}\})$ & \cite{NR20}                                & yes     &                \\ \cline{1-2} \cline{4-7}
		\rownumber. & 		Doubling             &    & $2^{O(\sddim)}\cdot\log\Phi$               & \cite{IMSZ10}                              & 		   &$\ddim,\Phi$     \\ \cline{1-2} \cline{4-7}
		\rownumber. & 		Doubling             &     & $O(\ddim\cdot\sqrt{\log\Phi})$             & \Cref{thm:DoublingOnlineEmbeddingtoEuclidean} & yes  &none             \\ \cline{1-2} \cline{4-7}
		\rownumber. & 		Constant $\ddim$     &    & $\Omega(\sqrt{\log\Phi})$                  & \Cref{thm:LBEuclideanDoubling}             & 		   &				\\ \hline
		\rownumber. & 		General              & \multirow{6}{*}{\begin{tabular}[c]{@{}c@{}}{\Large HST}\\ \\ (ultrametric)\end{tabular}} & $O(\log n\cdot \log\Phi)$                  & \cite{IMSZ10,BFU20}                        &   	   &none             \\ \cline{1-2} \cline{4-7}
		\rownumber. & 		General              &  & $\tilde{\Omega}(\log n\cdot \log\Phi)$     & \cite{BFU20}                               &         &                 \\ \cline{1-2} \cline{4-7}
		\rownumber. & 		Doubling             &  & $2^{O(\sddim)}\cdot \log\Phi$           	& \cite{IMSZ10}                              &         &$\ddim,\Phi$     \\ \cline{1-2} \cline{4-7}
		\rownumber. & 		Doubling             &  & $O(\ddim\cdot \log\Phi)$                   & \Cref{thm:DoublingOnlineEmbedding}         &  	   &none            \\ \cline{1-2} \cline{4-7}
		\rownumber. & 		$(\R^d,\|\cdot\|_2)$ &  & $O(\sqrt{d}\cdot \log\Phi)$                & \Cref{thm:EuclidineOnlineEmbeddingIntoHST} &   	   &none             \\ \cline{1-2} \cline{4-7}
		\rownumber. & 		$\R$                 &  & $\Omega(\min\{n,\log\Phi\})$               & \cite{IMSZ10}                              &         &                  \\ \hline		
		\rownumber. & 		General              & \multirow{2}{*}{{\large Tree}}        & $2^{n-1}$                                  & \cite{NR20}                                & yes     &none             \\ \cline{1-2} \cline{4-7}
		\rownumber. & 		General              &         & $\Omega(2^{\frac n2})$                     & \cite{NR20}                                & yes     &                   \\ \hline
	\end{tabular}
	\caption{\small{Summary of new and previous result on online metric embeddings. Doubling stands for metric space with doubling dimension $\ddim$.   The column ``Deter?'' indicates whether the embedding is deterministic (in particular works against adaptive adversary), or the guarantee is only in expectation (in particular works only against oblivious adversary). The column ``Pri.Kno''  indicates what prior knowledge is required by the embedding (applicable only for the upper bounds).
	}}
	\label{tab:OnlineEmbedding}
\end{table}

Our results on online metric embeddings are summarized in \Cref{tab:OnlineEmbedding}.
Our first result is a deterministic online embedding with distortion $O(\ddim)\cdot\min\{\sqrt{\log\Phi},\sqrt{n}\}$ into Euclidean space $\ell_2$ where the input metric has doubling dimension $\ddim$. No prior knowledge (of any of $n$, $\Phi$, and $\ddim$) is required. As every $n$-point metric space has doubling dimension $O(\log n)$, this result simultaneously: (1) proves the conjecture by Newman and Rabinovich~\cite{NR20} (with an upper bound of $O(\sqrt{n}\log n)$ ); (2) matches the lower bound up to second order terms; (3) exponentially improves the dependence on $\ddim$ compared to \cite{IMSZ10}; (4) quadratically improves the dependence on $\Phi$; (5) gives a deterministic distortion guarantee instead of expected distortion; and (6) removes the requirement to know $\ddim$ and $\Phi$ in advance.
In fact, the quadratic improvement in the dependence on $\Phi$ answers an open question in \cite{IMSZ10} (see Remark~2 in \cite{IMSZ10}).

\begin{restatable}[]{theorem}{OnlineEmbedingIntoEuclidean}
	\label{thm:DoublingOnlineEmbeddingtoEuclidean}
	For a sequence of metric points $x_1,\ldots ,x_n$ arriving in an online fashion, there is a deterministic online embedding into Euclidean space $\ell_2$ with distortion $O(\ddim)\cdot\min\{\sqrt{\log\Phi},\sqrt{n}\}$.
	Here $\Phi$ is the aspect ratio, and $\ddim$ is the doubling dimension of $\{x_1,\ldots , x_n\}$. No prior knowledge is required.
\end{restatable}

Our second result is a lower bound showing that for constant doubling dimension, \Cref{thm:DoublingOnlineEmbeddingtoEuclidean} is tight. Our lower bound holds even for non-expansive stochastic online embeddings with expected distortion guarantees. This generalizes~\cite{NR20}, where it was shown that there is a family of $n$-point metric spaces with aspect ratio $\Phi=2^{\Omega(n)}$ such that every deterministic embedding into $\ell_2$ requires distortion $\Omega(\sqrt{\log \Phi})=\Omega(\sqrt{n})$ (however their family does not have bounded doubling dimension).%

\begin{restatable}[]{theorem}{OnlinegHSTLB}
	\label{thm:LBEuclideanDoubling}
 \begin{enumerate}[label=(\alph*)]
    \item For every 
    deterministic online embedding algorithm into $\ell_2$, there is a family $\mathcal{M}=\{M_n\}_{n\ge 1}$ of metric spaces such that $M_n$ is an $O(n)$-point metric space with aspect ratio $\Phi_n=4^n$, and uniformly constant doubling dimension which constitutes a shortest path metric of a series-parallel (in particular, planar) graph, such that the deterministic online embedding has distortion  $\Omega(\sqrt{\log \Phi_n})=\Omega(\sqrt{n})$.
    \item The lower bound in part~(a) also holds for non-expansive stochastic embeddings. Specifically, for every $n\in\N$, there is a family of $O(n)$-point metric spaces with uniformly constant doubling dimension and aspect ratio $\Phi_n=4^n$ such that every non-expansive stochastic embedding into $\ell_2$ must have expected distortion (that is, expected contraction) at least $\Omega(\sqrt{\log \Phi_n})=\Omega(\sqrt{n})$ w.r.t.\ this family.
 \end{enumerate}
\end{restatable}

Our third result is a dominating stochastic embedding into ultrametrics with expected distortion $O(\ddim\cdot\log \Phi)$. This is a generalization of \cite{BFU20} (as every metric space has doubling dimension $O(\log n)$), and an exponential improvement in the dependence on $\ddim$ compared with \cite{IMSZ10}.

\begin{restatable}[]{theorem}{OnlineEmbedingHST}
	\label{thm:DoublingOnlineEmbedding}
	Given a sequence of metric points $x_1,x_2,\dots$ arriving in an online fashion, there is a dominating stochastic metric embedding into ultrametrics (2-HSTs) with expected distortion $O(\ddim\cdot\log\Phi)$, where $\ddim$ and $\Phi$ are the doubling dimension and the aspect ratio of the metric space. No prior knowledge is required.
\end{restatable}
\begin{remark}\label{rem:PriorDistortionEuclidean}
	{\rm	In fact, for a pair of points $\{x_j,x_{k}\}$ where $j<k$, the expected distortion guarantee provided by \Cref{thm:DoublingOnlineEmbedding} is $O(\ddim_j)\cdot\log\Phi_j$, where $\ddim_j$ and $\Phi_j$ are the doubling dimension and aspect ratio of the metric space induced by the prefix $\{x_1,\dots,x_j\}$.
		This is also known as prioritized distortion.
		See \cite{EFN18,BFN19-scaling,FGK20,EN22} for further details on prioritized distortion.}
\end{remark}

If the points arrive from  Euclidean $d$-dimensional space, we obtain expected distortion $O(\sqrt{d}\cdot\log\Phi)$, which is a quadratic improvement in the dependence on the dimension (see \Cref{thm:EuclidineOnlineEmbeddingIntoHST}).

Bartal \etal \cite{BFU20} used their dominating stochastic online embedding to design competitive online algorithms for certain network design problems. Surprisingly, they showed that in many cases the dependence on the aspect ratio can be avoided.
One can improve some parameters by pluging in our \Cref{thm:DoublingOnlineEmbedding,thm:EuclidineOnlineEmbeddingIntoHST} into their framework.
One example is the Subadditive Constrained Forest problem \cite{GW95}, where  we can improve the competitive ratio from $O(\log^2 k)$ to $O(\ddim\cdot\log k)$, (or $O(\sqrt{d}\cdot\log k)$ for points in Euclidean $d$-space). See \Cref{sec:Network-Design-Problems} for further discussion.

\subsubsection{Minimum Weight Perfect Matchings}

Our results on online minimum-weight perfect matchings
are summarized in \Cref{table:1}.

\begin{table}[ht]
\begin{tabular}{|l|l|l|l|l|l|}
\hline
Adversary  & Metric & Recourse & Approx.\ ratio & Approx.\ type & Reference \\ \hline

Adaptive  & General  &   $O(\log^2 n)$     &     $O(\log n)$   &  lightness   &  \Cref{thm:Tree} \\ \hline

Oblivious   & General    &   $O(\min\{\log^3 n,\log\Phi\})$   &  $O(\ddim\cdot\log\Phi)$              & comp.\ ratio     & \Cref{thm:ObliviousGeneralMetric}    \\ \hline

Oblivious   & $(\R^d,\|.\|_2)$    & $O(\min\{\log^3 n,\log\Phi\})$   &  $O(\sqrt{d}\cdot \log\Phi)$              & comp.\ ratio     & \Cref{thm:ObliviousEuclideanMetric}    \\ \hline

Oblivious   & 2-HST  & $O(\min\{\log^3 n,\log\Phi\})$   &  $O(1) $              & comp.\ ratio     & \Cref{lem:2HSTAppprox}    \\ \hline

Oblivious  & General     &   $r$     &     $\Omega(\frac{\log n}{r \log r})$   &  \begin{tabular}[c]{@{}l@{}}comp. ratio\\ \& lightness\end{tabular} &  \Cref{thm:LowerBound} \\ \hline

Oblivious   & General    &   $1$     &     $\infty$  &  comp.\ ratio &  \Cref{pp:one} \\ \hline

\end{tabular}
\caption{\small Summary of our online algorithms and lower bounds for online minimum weight perfect matching with recourse. The input is $n$ points in a metric space with aspect ratio $\Phi$.\label{table:1}}
\end{table}

In \Cref{sec:comp}, we design a randomized algorithm against an oblivious adversary that maintains a perfect matching with competitive ratio $O(\ddim\cdot \log \Phi)$ and  recourse $O(\log \Phi)$.

\begin{restatable}[]{theorem}{ObliviousGeneralMetricUB}\label{thm:ObliviousGeneralMetric}
 There is a randomized algorithm that, for any sequence of metric points $x_1,\dots,x_{2n}$ revealed by an oblivious adversary in an online fashion with aspect ratio $\Phi$ 
 and doubling dimension $\ddim$ (both unknown in advance), maintains a perfect matching of expected competitive ratio $O(\ddim\cdot\log\Phi)$ with recourse
 $O(\log\Phi)$. Alternatively, the recourse can be bounded by $O(\log^3 n)$.
\end{restatable}

Moreover, we show that the competitive ratio can be further improved to $O(\sqrt{d}\cdot \log \Phi)$ if the input points are from Euclidean $d$-space (see \Cref{thm:ObliviousEuclideanMetric}).

\old{
\begin{restatable}[]{theorem}{ObliviousEuclideanMetricUB}
\label{thm:ObliviousEuclideanMetric}
	There is a randomized algorithm such that, given a sequence of points $x_1,\dots,x_{2n}$ in $(\R^d, \|.\|_2)$ revealed by an oblivious adversary in an online fashion with aspect ratio $\Phi$ (unknown in advance), maintains a perfect matching of expected competitive ratio $O(\sqrt{d}\cdot\log\Phi)$ with recourse $O(\log\Phi)$.
	Alternatively, the recourse can be bounded by $O(\log^3 n)$.
\end{restatable}
}

Note that, every $n$-point metric space has doubling dimension $O(\log n)$. For example, for the shortest path metric of unweighted graphs, \Cref{thm:ObliviousGeneralMetric} provides a competitive ratio of $O(\log^2n)$.

\Cref{thm:ObliviousGeneralMetric,thm:ObliviousEuclideanMetric} are proven by a reduction to \emph{hierarchically well-separated tree} (\emph{HST}, a.k.a. ultrametric, see \Cref{def:HST}) via \Cref{thm:DoublingOnlineEmbedding,thm:EuclidineOnlineEmbeddingIntoHST}, respectively.
For an HST of height $h$,
we can maintain a minimum-weight perfect matching (i.e., an optimum matching) using recourse $O(h)$ (see \Cref{lem:inward,lem:ultrametricOptimal}). Using heavy-path decomposition, we can also maintain a $O(1)$-approximate minimum-weight  matching with recourse  $O(\log^3 n)$ (see \Cref{lem:2HSTAppprox}).

Next, we establish a lower bound using points on the real line (with linear aspect ratio) such that the recourse times the competitive ratio must be $\tilde{\Omega}(\log n)$.
Note that for metric spaces with polynomial aspect ratio, our \Cref{thm:ObliviousGeneralMetric,thm:ObliviousEuclideanMetric} are tight up to a quadratic factor. \footnote{As we try to optimize simultaneously both the competitive ratio and the recourse, it is natural to define a new parameter called \emph{performance}, which equals competitive ratio times recourse. Thus for metric space with polynomial aspect ratio and constant $\ddim$, our \Cref{thm:ObliviousGeneralMetric} has performance $O(\log^2n)$, while by \Cref{thm:LowerBound}, the performance is at least $\min_{r}\left\{ r,\frac{\log n}{\log r}\right\} =\Omega(\frac{\log n}{\log\log n})$. Thus, ignoring second order terms, \Cref{thm:ObliviousGeneralMetric} is tight up to a quadratic factor.}

\begin{restatable}[]{theorem}{MainNewLowerBoundTheorem}
\label{thm:LowerBound}
For every $r\geq 2$, every online algorithm for minimum-weight perfect matching problem with recourse $r$, even for $n$ points in the real line, has competitive ratio $\Omega\left( \frac{\log n}{r\cdot \log r}\right)$ against an oblivious adversary. Furthermore $r$ can depend on $n$.
\end{restatable}

Finally, we design a deterministic algorithm against an adaptive adversary that maintains a perfect matching of weight $O(\log n)\cdot \cost(\MST)$ with recourse $O(\log^2 n)$ in any metric space (\Cref{thm:Tree}).
The \emph{lightness} of weighted graph $G$ on a point set is the ratio $\frac{\cost(G)}{\cost(\MST)}$ of the weight of $G$ to the weight of an MST, and it is a popular measure in network optimization. The lightness of a perfect matching may be arbitrarily close to zero, and it is always at most one\footnote{The weight of an MST is at least as large as the weight of the minimum perfect matching. Indeed, following the approach of Christofides algorithm, double each edge of the MST to obtain an Euler cycles $C$ of weight $2\,\cost(\MST)$. Let $x_1,\dots,x_{2n}$ be the order of the points in the order of their first occurrence along the tour. The two matchings $M_1=(x_1,x_2),\dots,(x_{2n-1},x_{2n})$ and $M_2=(x_2,x_3),\dots,(x_{2n-2},x_{2n-1}),(x_{2n},x_1)$ combined have weight at most $\cost(C)=2\,\cost(\MST)$. In particular, the weight of the minimum perfect matching is at most $\cost(\MST)$.}. However, for $2n$ uniformly random points in a unit cube $[0,1]^d$ for constant $d\in \mathbb{N}$, for example, the expected minimum weight of a perfect matching is proportional to the maximum weight of an $\MST$~\cite{SupowitRP83,SteeleS89}.
Interestingly, in our lower bound (\Cref{thm:LowerBound})
the weight of the perfect matching is $\Theta(\diam)=\Theta(\MST)$. Thus, in particular, it implies that the product of the recourse and the lightness of any oblivious algorithm is $\tilde{\Omega}(\log n)$. Hence, our algorithm against an adaptive adversary is comparable to the best possible oblivious algorithm w.r.t.\ the lightness parameter.

\begin{restatable}[]{theorem}{MatchingGeneralMetrics}
\label{thm:Tree}
For a sequence of points in a metric space $(X,d)$, we can maintain a perfect matching of weight $O(\log |S_i|)\cdot\cost(\MST(S_i))$ using recourse $O(\log^2 |S_i|)$ where $S_i$ is the set of the first $2i$ points.
\end{restatable}

\subsection{Technical Ideas}\label{subsec:techIdeas}

\smallskip\noindent\textbf{Online Padded Decompositions and Online Embedding into HST.}
A \emph{$(\Delta,\beta)$-padded decomposition} of a metric space $(X,d_X)$ is a random partition of $X$ into clusters of diameter at most $\Delta$ such that for a ball $B=B_X(v,r)$ of radius $r$ centered at $v$, the probability that the points of $B$ are split between different clusters is at most $\beta\cdot\frac{r}{\Delta}$. A metric space is \emph{$\beta$-decomposable} if it admits a $(\Delta,\beta)$-padded decomposition for every $\Delta>0$.
Building on previous work~\cite{Bartal96,Rao99,KLMN04}, the main ingredient in all our metric embeddings (in particular the deterministic \Cref{thm:DoublingOnlineEmbeddingtoEuclidean}) are padded decompositions.
Every $n$-point metric space is $O(\log n)$-decomposable \cite{Bartal96}, while every metric space with doubling dimension $\ddim$ is $O(\ddim)$-decomposable \cite{GKL03,Fil19Approx}.
Roughly speaking, Bartal's padded decomposition \cite{Bartal96} works using a ball growing technique: Take an arbitrary order over the metric points $x_1,\dots,x_n$, sample radii $R_1,\dots,R_n$ from exponential distribution with parameter $\Theta(\log n)$, and successively construct clusters  $C_i=B_X(x_i,R_i\cdot\Delta)\setminus\bigcup_{j<i}C_j$. In words, there are $n$ clusters centered in the points of $X$. Each point $y$ joins the first cluster $C_i$ such that $d_X(y,x_i)\le R_i\cdot\Delta$.
With high probability, it holds that $\max{R_i}\le \frac12$ and thus the diameter of all the resulting clusters is bounded by $\Delta$.
Further, consider a ball $B=B_X(v,r)$, and suppose that $u\in B$ is the first point to join a cluster $C_i$. Then $R_i\ge \frac{d_X(x_i,u)}{\Delta}$.
By the triangle inequality,
$R_i\ge \frac{d_X(x_i,u)}{\Delta}+2r$ will imply that all the points in $B$ will join $C_i$.
By the memoryless property of the exponential distribution,
the probability that  $R_i<\frac{d_X(x_i,u)}{\Delta}+2r$ is at most $O(\log n)\cdot\frac{2r}{\Delta}$.

An HST (hierarchically separated tree, see \Cref{def:HST}) is in essence just a hierarchical partition.
Bartal's embedding into HST then works by creating separating decompositions for all possible distance scales $\Delta_i=2^i$, $i\in\Z$. That is, each $i$-level cluster is partitioned into $(i-1)$-level clusters using the padded decomposition described above. Let $i$ be the first scale such that $u$ and $v$ were separated, the distance between them in the HST will be $2^{i+1}$.  Thus the expected distance between $u$ and $v$ is bounded by\footnote{Bartal \cite{Bartal96} obtains expected distortion $O(\log^2 n)$ by contracting all pairs at distance at most $\frac{\Delta_i}{\poly(n)}$ before preforming the decomposition at scale $i$. The effect of this contraction on pairs at distance $\tilde{\Theta}(\Delta_i)$ is negligible, while the contraction ensures that $u,v$ have nonzero probability of being separated in only $O(\log n)$ different scales.}
\begin{align}
\sum_{i\ge\log d_{X}(u,v)}2^{i+1}\cdot\Pr[u \mbox{ {\rm and } }v\mbox{ {\rm are separated at level }}i]
&\le\sum_{i\ge\log d_{X}(u,v)}2^{i+1}\cdot O(\log n)\cdot\frac{d_{X}(u,v)}{2^{i}}\nonumber\\
&\le O(\log\Phi\cdot\log n)~.\label{eq:HSTexplain}
\end{align}
Indyk \etal \cite{IMSZ10} observed that, given metric points $x_1,\dots,x_n$ in an online fashion, we can still preform Bartal's padded decomposition, since the order of the points were arbitrary. However, \cite{IMSZ10} required prior knowledge of $n$ and the aspect ratio $\Phi$ (in order to determine the parameter of the exponential distribution, and the relevant scales).
Later, Bartal \etal \cite{BFU20} observed that sampling $R_j$ using exponential distribution with parameter $O(\log j)$ will still return a padded decomposition, and thus removed the requirement to know $n$ in advance.
To remove the requirement to know $\Phi$ in advance, \cite{BFU20} simply forced $R_1$ to be $\Omega(1)$, and thus ensuring that all the partitions above a certain threshold are trivial.

Gupta, Krauthgamer, and Lee \cite{GKL03} showed that every metric space with doubling dimension $\ddim$ is $O(\ddim)$-decomposable. Their decomposition follows the approach from \cite{CKR04}, which samples a global permutation to decide where to cluster each point. Later, Filtser \cite{Fil19Approx} used the random shifts clustering algorithm of Miller, Peng, and Xu \cite{MPX13} to obtain a similar decomposition with strong diameter guarantee.\footnote{Given an edge-weighted graph $G=(V,E,w)$ and a cluster $C\subseteq V$, the \emph{(weak) diameter} of $C$ is $\max_{u,v\in C}d_G(u,v)$ the maximum pairwise distance in $C$ w.r.t.\ the shortest path metric of $G$. The \emph{strong diameter} of $C$ is $\max_{u,v\in C}d_{G[C]}(u,v)$ the maximum pairwise distance in $C$ w.r.t.\ the shortest path metric of the induced subgraph $G[C]$.} However, both these decompositions are crucially global and centralized, and it is impossible to execute them in an online fashion.
Indyk \etal \cite{IMSZ10} studied online embeddings of doubling metrics into ultrametrics. However, lacking good padded decompositions for doubling  spaces, they ended up using a similar partition based approach, which lead to an expected distortion $2^{O(\sddim)}\cdot\log\Phi$.

We show that one can construct a padded decomposition with padding parameter $O(\ddim)$ using the ball growing approach of Bartal \cite{Bartal96}: Sample the radii using an exponential distribution with parameter $O(\ddim)$.
This is crucial, as such a decomposition can be executed in an online fashion.
Furthermore, one does not need to know the doubling dimension in advance. It is enough to use the doubling dimension of the metric space induced on the points seen so far.
Interestingly, even if the doubling dimension eventually will turn out to be $O(\log n)$, the decomposition will have the optimal $O(\log n)$ parameter.
Using these decomopsitions, we construct an HST in an online fashion and obtain \Cref{thm:DoublingOnlineEmbedding} by replacing the $O(\log n)$ factor in inequality \eqref{eq:HSTexplain} by $O(\ddim)$.

\smallskip\noindent\textbf{Online Deterministic Embedding into Euclidean Space.}
Bourgain's \cite{Bourgain85} optimal embedding into Euclidean space with distortion $O(\log n)$ is a Fr\'{e}chet type embedding. Specifically, it samples subsets uniformly with different densities, and sets each coordinate to be equal to the distance to a certain sampled subset.
This is a global, centralized approach that cannot be executed in an online fashion.
In contrast, Rao's \cite{Rao99} classic embedding of the shortest path metric of planar graphs into $\ell_2$ with distortion $O(\sqrt{\log n})$ is based on padded decompositions.
Indyk \etal \cite{IMSZ10} followed
the padded decomposition based approach of Rao \cite{Rao99}.
Roughly speaking, Rao's approach is to create padded decompositions $\cP$ for all possible distance scales, where we have a distinct coordinate for each partition $\cP$.
For every cluster $C\in\cP$, assign a random coefficient $\alpha_C\in\{\pm1\}$. Finally, for every $x\in C$, assign value $\alpha_C\cdot h_C(x)$, where $h_C(x)$ is the distance from $x$ to the ``boundary'' of $C$. One can show that the distance in every coordinate is never expanding, while for two points $x,y$, in the scale $d_X(x,y)$, with constant probability, $x,y$ will be separated, and $x$ will be at least $\Omega(\frac{d_X(x,y)}{\log n})$ away from the boundary of its cluster---thus we will get some contribution to the distance.
By sampling many such decompositions, one can get concentration, and thus an embedding with distortion $O(\log n\cdot\log\Phi)$ against an oblivious adversary.
Indyk \etal did not suggest any way to cope with an adaptive adversary.
Newman and Rabinovich \cite{NR20} provided an $\Omega(\sqrt{n})$ lower bound for such an embedding, and conjectured that $\poly(n)$ distortion should always suffice. However, they did not suggest any way to achieve it.

Next Indyk \etal \cite{IMSZ10} moved to embedding doubling spaces. Lacking  good online padded decompositions, they observed that an isometric embedding of ultrametric into $\ell_2$ can be maintained in an online fashion, and thus getting expected distortion $2^{O(\sddim)}\cdot\log\Phi$ against an oblivious adversary.
Plugging in our new padded decomposition into Rao's approach, one can get (worst case w.h.p.) distortion $O(\ddim\cdot\sqrt{\log\Phi})$ against an oblivious adversary. But how can one construct deterministic embedding to cope with an adaptive adversary?

Our solution is to create a ``layer of abstraction''. The previous ideas provide an embedding with expected distortion $O(\ddim\cdot\sqrt{\log\Phi})$.
Specifically, we get expansion $O(\sqrt{\log\Phi})$ in the worst case, and contraction $O(\ddim)$ in expectation.
The only randomness is over the choice of radii in the ball growing that creates the padded decompositions (and some additional boolean parameters). Given a new point $x_i$, the expected squared distance $\E_f[\|f(x_i)-f(x_j)\|_2^2]$ can be computed exactly, as it only depends on the points that have arrived so far, with no randomness involved. In a sense, instead of mapping a metric point $x_i$ into a vector in $\ell_2$, we map it into a well-defined function $f_i:(r_1,r_2,\dots,r_i)\rightarrow\ell_2$. These functions are in the function space $L_2$, and the distance
\[
    \|f_{i}-f_{j}\|_{2}=\left(\int_{r_{1},\dots,r_{i}}\left\Vert f_{i}(r_{1},\dots,r_{i})-f_{j}(r_{1},\dots,r_{j})\right\Vert _{2}^{2}\right)^{\frac{1}{2}}
\]
equals the expected distance by the random metric embedding.
However, we want to return vectors and not complicated functions. In fact, the only required information is the $L_2$ distance between these functions, which define an Euclidean distance matrix. Given such a matrix, one can find a set of vectors implementing it (which is unique up to rotation and translation). Furthermore, these vectors can be efficiently and deterministically computed in an online fashion!

\smallskip\noindent\textbf{Online Minimum-Weight Perfect Matchings in Metric Spaces.}
Given a sequence of metric points $s_1,\ldots , s_{2i}$ in an online fashion from an unknown metric space, we use the online embedding algorithm to embed them into an ultrametric or the real line (with some distortion), and maintain a matching with recourse on the embedded points. If we can maintain a good approximation for the online MWPM in an ultametric (or in $\mathbb{R}$), then we can maintain the same approximation ratio, with the distortion of the embedding as an overhead, for the online MWPM problem.

\smallskip\noindent\textbf{Optimal Matchings on Trees: Inward Matchings.}
An ultrametric is represented by a hierarchically well-separated tree $T$, a rooted tree with exponentially decaying edge weights. As points arrive in an online fashion, the online embedding algorithm may successively add new leaves to $T$; and the points are embedded in the leaves of $T$. We show that a simple greedy matching is optimal (i.e., has minimum-weight) in an ultrametric, and can easily be updated with recourse proportional to the height of $T$ (\Cref{lem:inward,lem:ultrametricOptimal}). Specifically, an \emph{inward matching}, introduced here, maintains the invariant that the points in each subtree induce a near-perfect matching. When a pair of new points arrive, we can restore this property by traversing the shortest paths between the corresponding nodes in $T$. Consequently, an inward matching can be maintained with recourse $O(h)$, where $h$ is the height of $T$. We have $h\leq O(\log \Phi_U)$, where $\Phi_U$ is the aspect ratio of the ultrametric $U$, which in turn is bounded by $O(\log \Phi_X)$, the aspect ratio of the metric space $X$ induced by the input points seen so far.

\smallskip\noindent\textbf{Heavy-Path Decomposition on HSTs.}
Bartal \etal \cite{BFU20} showed that the distortion of any online metric embedding algorithm into trees depends on the aspect ratio $\Phi$, in particular, the factor $\log \Phi$ in our distortion bounds in Theorems~\ref{thm:DoublingOnlineEmbedding}--\ref{thm:EuclidineOnlineEmbeddingIntoHST} is unavoidable. It is unclear whether any dependence on $\Phi$ is necessary for the bounds for the online MWPM problem. We can eliminate the dependence on $\Phi$ for the recourse, while maintaining the same approximation guarantee (which, however, still depends on the distortion, hence on $\Phi$).

Instead of an optimal matching on the HST $T$, we maintain an 2-approximate minimum-weight perfect matching. We use the classical \emph{heavy-path decomposition} of the tree $T$, due to Sleator and Tarjan~\cite{SleatorT83}, which is a partition of the vertices into subsets that each induce a path (\emph{heavy path}); the key property is that every path in $T$ intersects only $O(\log n)$ heavy paths, regardless of the height of $h$. The heavy path decomposition can be maintained dynamically with $O(\log n)$ split-merge operations over the paths.

We relax the definition of inward matchings such that at most one edge of the marching can pass between any two adjacent heavy paths, but we impose only mild conditions within each heavy path. A  charging scheme shows that the relaxed inward matching is a 2-approximation of the MWPM (\Cref{lem:2HSTAppprox}). On each heavy path, we maintain a matching designed for points on a real line, with $O(\log^2 n)$ depth\footnote{The \emph{depth} of a matching on $n$ points in $\mathbb{R}$ (or a path) is the maximum number of pairwise overlapping edges.}, which supports split-merge operations in $O(\log^2 n)$ changes in the matching (i.e., recourse).
Overall, we can maintain a 2-approximate minimum-weight matching on $T$ with worst-case recourse $O(\log^3 n)$.

\smallskip\noindent\textbf{Minimum-Weight Matching on a Real Line: Reduction to Depth.}
We reduce the online MWPM problem to a purely combinatorial setting. For a set of edges $E$ on a finite set $S\subset \mathbb{R}$, we say that $E$ is \emph{laminar} if there are no two edges $a_1b_1$ and $a_2b_2$ such that $a_1<a_2<b_1<b_2$ (i.e., no two \emph{interleaving} or \emph{crossing} edges). Containment defines a partial order: We say that $a_1b_1\preceq a_2b_2$ (resp., $a_1b_1\prec a_2b_2$) if the interval $a_2b_2$ contains (resp., properly contains) the interval $a_1b_1$. The Hasse diagram of a laminar set $E$ of edges is a forest of rooted trees $F(E)$ on $E$, where a directed edge $(a_1b_1,a_2b_2)$ in $F(E)$ means that $a_2b_2$ is the shortest interval that strictly contains $a_1b_1$. The \emph{depth} of $E$ is the depth of the forest $F(E)$; equivalently, the depth of $E$ is the maximum number of pairwise overlapping edges in $E$.
Based on this, we show that for a dynamic point set $S$ on $\mathbb{R}$, one can maintain a laminar near-perfect matching of depth $O(\log n)$ such that it modifies (adds or deletes) at most $O(\log^2 n)$ edges in each step.

Importantly, the laminar property and the depth of the matching depend only on the order of the points in $S$, and the real coordinates do not matter. While it is not difficult to maintain a laminar near-perfect matching. However, controlling the depth is challenging. We introduce the notion of \emph{virtual edges}, which is the key technical tool for maintaining logarithmic depth.
We maintain a set of invariants that ensure that, for a nested sequence of edges yields a nested sequence of virtual edges with the additional
property that they have exponentially increasing lengths. We argue that if a near-perfect matching with virtual edges satisfies the invariants then the depth of such matching is logarithmic.

We reduce the case of general metric space to a line metric using a result by Gu \etal \cite{GuG016}:
Given a sequence of metric points, one can maintain a spanning tree of weight $O(\cost(\MST))$ in an online fashion (insertion only), with constant recourse per point.
We maintain an Euler tour\footnote{A DFS traversal of a tree $T$ (starting from an arbitrary root) defines an \emph{Euler tour} $\mathcal{E}(T)$ that traverses every edge of $T$ precisely twice (once in each direction).} $\mathcal{E}(T)$ for the tree $T$ produced by their algorithm.
The Euler tour $\mathcal{E}(T)$ induces a Hamilton path $\mathcal{P}(T)$ (according to the order of first appearance). Intuitively, we treat the metric points as if they were points on a line ordered according to $\mathcal{P}(T)$.
We show that each edge deletion and each edge insertion in the forest $F$ incurs $O(1)$ edge insertions or deletions in the tour $\mathcal{E}(T)$ and path $\mathcal{P}(T)$, a very limited change! Thus we can use our data structure for the line (\Cref{sec:line}) to maintain a near-perfect matching for the points w.r.t.\  $\mathcal{P}(T)$.
As the total weight of this path is $O(\cost(\MST))$, we obtain a perfect matching of weight $O(\log n)\cdot\cost(\MST)$ using recourse $O(\log^2 n)$.

\smallskip\noindent\textbf{Lower Bounds for Competitive Ratio and Recourse on Minimum-Weight Matchings.}
For integers $r\geq 2$ and $n\geq 10\, r$, we show (\Cref{thm:LB}) that an adaptive adversary can construct a sequence $S_n$ of integer points on the real line such that any deterministic online perfect matching algorithm with recourse $r$ per arrival maintains a perfect matching of weight $\Omega\left( \frac{\diam(S_n)\,\log n}{r\, \log r}\right)$. The adversary presents points in $k=\Theta\left( \frac{\log n}{r\, \log r}\right)$ rounds. In round~0, the points are consecutive integers. In subsequent rounds, the number of points decreases exponentially, but the spacing between them increases. If the algorithm matches new points among themselves in every round, the weight of the resulting matching would be $\Omega(k\cdot \opt)$. The weight could be improved with recourse, however, the number of new points rapidly decreases, and they do not generate enough recourse to make amends: We show that the weight increases by $\Omega(\opt)$ in every round. There is one twist in the adversarial strategy, which makes it adaptive: If the matching $M_{i-1}$ at the beginning of round~$i$ could possibly ``absorb'' the points of round~$i$ (in the sense that the weight would not increase by $\Omega(\opt)$), then we show that $M_{i-1}$ already contains many long edges and $\cost(M_{i-1})\geq \Omega(i\cdot \opt)$: In this case, the adversary can simply skip the next round.

In fact, this lower bound construction extends to oblivious adversaries by skipping some of the rounds randomly. Moreover, since, for a set of points in the real line, the minimum weight of a perfect matching is trivially bounded by the diameter of the point set, we conclude that the lightness and competitive ratio of any online algorithm with constant recourse is $\Omega(\log n)$; and an $O(1)$-competitive algorithm would require recourse at least  $r=\Omega(\log n/\log \log n)$.

\subsection{Related Work}\label{sec:related-work}

\noindent\sloppy\textbf{Online Minimum-Weight Perfect Matching with Delays.}
Similarly to our model, Emek et al.~\cite{EmekKW16} considered online minimum-weight perfect matchings in a metric space. However, they allow \emph{delays} instead of recourse: The decisions of the online matching algorithm are irrevocable, but may be delayed, incurring a time penalty of $t_i$ if a point $s_i$ remains unmatched for $t_i$ units of time. The objective is to minimize the \emph{sum} of the weight and all time penalties. Emek et al.~\cite{EmekKW16} show that a randomized algorithm (against an oblivious adversary) can achieve a competitive ratio $O(\log^2 n+\log \Phi)$ in this model, where
$\Phi$ is the aspect ratio of the metric space (which can be unbounded as a function of $n$). Later, Azar et al.~\cite{AzarCK17}, improved the competitive ratio to $O(\log n)$. Ashlagi et al.~\cite{AshlagiACCGKMWW17} studied the bipartite version of this problem; and Mari et al.~\cite{MariPRS23} considered the stochastic version of this problem where the input requests follow Poisson arrival process.
Recently, Deryckere and Umboh~\cite{DU23} initiated the study of online problems with set delay, where the delay cost at any given time is an arbitrary function of the set of pending requests.
However, time penalties cannot be directly compared to the recourse model. An advantage of the recourse model is that it allows us to maintain perfect matching explicitly at all times, as opposed to the delay model where some points might remain unmatched at every step.

Online algorithms with recourse have been studied extensively over the years; see~\cite{IM91,MegowSVW16, GuG016, GGKKS22, BGW21, BHR19, GKS20}. The question is, given the power of hindsight, how much one can improve the solution of an online algorithm~\cite{GuG016}.

\smallskip\noindent\textbf{Online MST with Recourse.} In the online minimum spanning tree (MST) problem, points in a metric space arrive one by one, and we need to connect each new point to a previous point to maintain a spanning tree. Without recourse, Imase and Waxman~\cite{IM91} showed that a natural greedy
algorithm is $O(\log n)$-competitive, and this bound is the best possible (see also~\cite{AlonAzar93, DumitrescuT07}). They also showed how to maintain a $2$-competitive tree with recourse $O(n^{3/2})$ over the first $n$ arrivals for every $n$. Therefore, the amortized budget, i.e., the average number of swaps per arrival, is $O(\sqrt{n})$. Later, Megow et al.~\cite{MegowSVW16} substantially improved this result. They gave an algorithm with a constant amortized budget bound. In a breakthrough, Gu et al.~\cite{GuG016} showed that one can maintain a spanning tree of weight $O(\cost(\MST))$ with recourse $O(1)$ per point.

When new points arrive in the online model, the weight of the MST may decrease. However, it cannot decrease by a factor more than $\frac12$. Indeed, the decrease is bounded by the \emph{Steiner ratio}~\cite{GilbertPollak}, which is the infimum of the ratio between the weight of a Steiner tree and the MST for a finite point set, and is at least $\frac12$ in any metric space. Similarly, in the online traveling salesman  problem (TSP), where the length of the optimal TSP tour increases monotonically as new points arrive, one can maintain an $O(1)$-competitive solution with constant recourse (see \Cref{sec:related-work}).

\smallskip\noindent\textbf{Online TSP.} In the online traveling salesman problem (TSP), points of a metric space arrive one by one, and we need to maintain a traveling salesman tour (or path) including the new point. Rosenkrantz et al.~\cite{RosenkrantzSL77} showed that a natural greedy algorithm with \emph{one} recourse per point insertion (replacing one edge by two new edges) is $O(\log n)$ competitive, and there is a lower bound of $\Omega(\log n/\log \log n)$ even in Euclidean plane~\cite{Azar94, BafnaKP94}. However, as an Euler tour around an MST 2-approximates the weight of a TSP tour, the online MST algorithm by Gu et al.~\cite{GuG016} immediately yields a $O(1)$-competitive algorithm with recourse $O(1)$.

Kalyanasundaram and Pruhs~\cite{kalyanasundaram1994constructing} studied a variant of online TSP, where new cities are revealed locally during the traversal of a tour (i.e., an arrival at a city reveals any adjacent cities that must also be visited). Jaillet and Lu~\cite{jaillet2011online} studied the online TSP with service flexibility, where they introduced a sound theoretical model to incorporate “yes-no” decisions on which requests to serve, together with an online strategy to visit the accepted requests.

\smallskip\noindent\textbf{Greedy Matchings.}
Online algorithms with or without recourse often make greedy choices.
For the online MWPM, the following greedy approach with constant recourse seems intuitive: Suppose points $p_1$ and $p_2$ arrive when our current matching on $S_i$ is $M_i$. Then we find a closest neighbor for $p_1$ and $p_1$, resp., say $a_1$ and $a_2$ in $S_{i+1}$, delete any current edges $a_1b_1,a_2b_2\in M_i$, and add all edges of a minimum-weight matching on $\{a_1,a_2,b_1,b_2,p_1,p_2\}$ to the matching.

An online greedy approach would, at best, ``approximate'' an offline greedy solution. The \emph{offline greedy} algorithm successively adds an edge $ab$ between the closest pair of vertices and removes both $a$ and $b$ from further consideration. Reingold and Tarjan~\cite{ReingoldT81} showed, however, that the greedy algorithm on $2n$ points in a metric space achieves an $O(n^{\log \frac32})$-approximation, where $\log\frac32\approx 0.58496$, and this bound is the best possible already on the real line. Frieze, Mc{D}iarmid, and Reed~\cite{FriezeMR90} later showed that for integers $S_n=\{1,2,\ldots , 2n\}\subset \mathbb{R}$, the offline greedy algorithm returns an $O(\log n)$-approximation, and this bound is tight (if ties are broken arbitrarily when multiple point pairs attain the minimum distance).

\smallskip\noindent\textbf{Metric Embeddings.} There is a vast literature on metric embeddings that we will not attempt to cover here. We refer to the extended book chapter \cite{MatEmbedding13}, and some of the recent papers for an overview \cite{AFGN22,FGK20,Fil21,FL21}. See also the recent \href{https://hackmd.io/@3S70qBUwTR6_CErLY2dm4A/SJfp46KGi}{FOCS22 workshop}.
In the context of online embeddings, embeddings into low dimensional $\ell_1$, $\ell_2$, and $\ell_\infty$ normed spaces were studied \cite{IMSZ10,NR20}. In particular, every tree metric admits an isometric (with distortion $1$) online embedding into $\ell_1$ \cite{NR20}.
A significant part of the metric embeddings literature is concerned with the embedding of topologically restricted metric spaces, such as planar graphs, minor free graphs, and graphs with bounded treewidth/pathwidth \cite{Rao99,KLMN04,Fil20,FKS19,CFKL20,FL22}. However, at present, these embeddings do not have online counterparts. The reason is perhaps the lack of a good online version of a padded decompositions for such spaces~\cite{KPR93,FT03,AGGNT19,Fil19Approx}. Designing online padded decompositions for such spaces is a fascinating open problem.

\section{Preliminaries}
\paragraph{Ultrametrics.}
An ultrametric $\left(X,d\right)$ is a metric space satisfying a strong form of the triangle inequality, that is, for all $x,y,z\in X$,
$d(x,z)\le\max\left\{ d(x,y),d(y,z)\right\}$. A related notion is a $k$-hierarchically well-separated tree ($k$-HST).

\begin{definition}[$k$-HST]\label{def:HST}
	A metric $(X,d_X)$ is a \emph{$k$-hierarchically well-separated tree} (\emph{$k$-HST}) if there exists a bijection $\varphi$ from $X$ to leaves of a rooted tree $T$ in which:
	\begin{enumerate}
		\item each node $v\in T$ is associated with a label $\Gamma_{v}$ such that $\Gamma_{v} = 0$ if $v$ is a leaf, and $\Gamma(v)\geq k\Gamma(u)$ if $v$ is an internal node and $u$ is any child of $v$;
		\item $d_X(x,y) = \Gamma(\lca(\varphi(x),\varphi(y)))$ where $\lca(u,v)$ is the least common ancestor of any two given nodes $u,v$ in $T$.
	\end{enumerate}
\end{definition}

It is well known that any ultrametric is a $1$-HST, and any $k$-HST is an ultrametric (see \cite{BartalLMN04}).
Note that a $1$-HST induces a laminar partition of the metric space.
For an internal node $v\in T$, denote by $X_v$ the set of leaves in the subtree rooted at $v$.
We will use the terms ultrametric and $1$-HST interchangeably throughout the paper.

\paragraph{Doubling Dimension.}  The doubling dimension of a metric space is a measure of its local ``growth rate''.
A metric space $(X,d)$ has doubling constant $\lambda$ if for every $x\in X$ and radius
$r>0$, the ball $B(x,2r)$ can be covered by $\lambda$ balls of radius $r$. The doubling dimension is defined as $\ddim=\log_2\lambda$. A $d$-dimensional $\ell_p$ space has $\ddim=\Theta(d)$, and every $n$ point metric has $\ddim=O(\log n)$.
Even though it is NP-hard to compute the doubling dimension of a metric space \cite{GK13}, in polynomial time, one can compute an $O(1)$-approximation \cite[Theorem~9.1]{HM06}.
The following lemma gives the standard packing property of doubling metrics (see, e.g., \cite[Proposition~1.1]{GKL03}.).
\begin{lemma}[Packing Property] \label{lem:doubling_packing}
	Let $(X,d)$ be a metric space  with doubling dimension $\ddim$.
	If $S \subseteq X$ is a subset of points with minimum interpoint distance $r$ that is contained in a ball of radius $R$, then
	$|S| \le 2^{\ddim\cdot\left\lceil \log\frac{2R}{r}\right\rceil}$ .
\end{lemma}

\section{Online Padded Decompositions}

\subsection{Online Net Construction}

Hierarchical nets in metric spaces were adapted to dynamic point sets in the early 2000s~\cite{GaoGN06,KrauthgamerL04}. In particular, in the semi-dynamic (insert-only) setting, these data structures only insert (but never delete) points into each net, so they can be used in online algorithms (with irrevocable decisions). For the sake of completeness, we show how to construct hierarchical nets online.

Given a metric space $(X,d_X)$, a \emph{$\Delta$-net} $N\subseteq X$ is a set of points at pairwise distance at least $\Delta$, such that for every point $x\in X$ there is a net point within distance $\Delta$. Given a sequence of points in an online fashion, one can easily construct a $\Delta$-net using a greedy algorithm (add $x_j$ to $N$, if $N$ does not contain any net point at distance at most $\Delta$ from $x$).
For our purpose, we will need nets for all possible distance scales.
Storing greedy nets for an $n$-point metric space for all possible distance scales might require large space (as nets in consecutive scales might differ greatly).
Instead, we will create a nested sequence of nets for all the distance scales simultaneously, in an online fashion.

\begin{lemma}\label{lem:nets}
	There is an algorithm that,
	for a sequence of points in a metric space $(X,d_X)$ arriving in an online fashion, maintains a nested sequence of nets $\dots\supseteq N^{(j)}_{-2}\supseteq N^{(j)}_{-1}\supseteq N^{(j)}_{0}\supseteq N^{(j)}_{1}\supseteq N^{(j)}_{2}\supseteq\dots$,
	where $N^{(j)}_i$ denote the $i$'th net after the algorithm has seen $x_1,\dots,x_j$.
	Furthermore, it holds that:	
	\begin{enumerate}\itemsep 0pt
		\item The algorithm is online, that is, once a point joins $N_{i}$, it remains there forever.
		\item For every $i,j$, we have $N^{(j)}_i\subseteq N^{(j+1)}_i$.
		\item For every $x,y\in N^{(j)}_{i}$, we have $d_{X}(x,y)>2^{i}$.
		\item For every $x_i$ and every $i$, there is $y\in N^{(j)}_{i}$ such that $d_{X}(x,y)<2^{i+1}$.
	\end{enumerate}
\end{lemma}

\begin{proof}
	The construction is incremental: We need to show how to maintain all required properties after the arrival of each point. Consider the sequence $x_1,x_2,\ldots$ of metric points. The first point, $x_{1}$, belongs to every net. Denote by $N_{i}^{(j)}$ the $i$'th-net after we
	have seen $x_{1},\dots,x_{j}$.
	Next $x_{j+1}$ arrives. Let $\tilde{i}$
	be the maximum number such that $d_{X}(x_{j+1},N_{i}^{(j)})>2^{i}$
	for every $i\le\tilde{i}$.
	In other words, $d_X(x_{j+1},N_{\tilde{i}+1}^{j})\le 2^{\tilde{i}+1}$, and for every $i\le\tilde{i}$ there is no net point in $N_{i}^{(j)}$ at distance $2^i$ from $x_{j+1}$.
	We set
	\[
	N_{i}^{(j+1)}=\begin{cases}
		N_{i}^{(j)} & i>\tilde{i} ,\\
		N_{i}^{(j)}\cup\{x_{j+1}\} & i\le\tilde{i} .
	\end{cases}
	\]
	That is, we added $x_{j+1}$ to all the nets $N_{\tilde{i}}^{(j+1)},N_{\tilde{i}-1}^{(j+1)},N_{\tilde{i}-2}^{(j+1)},\dots$ .
	Clearly, the nets are nested at every point in time (and could be
	stored with $O(n)$ machine words). In addition, by definition, for
	every $x,y\in N_{i}^{(j+1)}$, $d_{X}(x,y)>2^{i}$. We prove the final property of the nets by induction on $j$ and $i$. Consider $x_{j+1}$,
	for $i\le\tilde{i}$ it holds that $d_{X}(x_{j+1},N_{i}^{(j+1)})=0<2^{i+1}$.
	For $i=\tilde{i}+1$, it holds that $d_{X}(x_{j+1},N_{i}^{(j+1)})=d_{X}(x_{j+1},N_{i}^{(j)})\le2^{i}<2^{i+1}$.
	Assume that the hypothesis holds for $i$, and consider $i+1$. There is
	a point $x_{j'}\in N_{i}^{(j+1)}$ such that $d_{X}(x_{j+1},x_{j'})<2^{i+1}$.
	If $x_{j'}\in N_{i+1}^{(j+1)}$, then clearly $d_{X}(x_{j+1},N_{i+1}^{(j+1)})\le d_{X}(x_{j+1},x_{j'})<2^{i+1}<2^{i+2}$.
	Else, $N_{i}^{(j')}$is the maximal net $x_{j'}$ belongs to, and hence
	there is a point $x_{j''}\in N_{i+1}^{(j')}$ for which $d_{X}(x_{j'},x_{j''})<2^{i+1}$.
	We conclude
	\begin{align*}
		d_{X}(x_{j+1},N_{i+1}^{(j+1)}) & \le d_{X}(x_{j+1},x_{j''})\le d_{X}(x_{j+1},x_{j'})+d_{X}(x_{j'},x_{j''})<2^{i+1}+2^{i+1}=2^{i+2}~,
	\end{align*}
	as required.
\end{proof}

\subsection{Online Low Diameter Decompositions}
This subsection is devoted to constructing low diameter decomposition for doubling spaces:
\begin{theorem}\label{thm:LDDonline}
	Consider a sequence of metric points $x_1,x_2,\dots$ arriving in an online fashion, given by an oblivious adversary. Given a parameter $\Delta>0$, there is an online algorithm sampling a partition $\cP$ of $X=\{x_1,x_2,\dots\}$ in an online fashion (i.e. once a points joins a cluster it remains there forever) such that:
	\begin{enumerate}
		\item Every cluster has diameter at most $\Delta$.
		\item For every $x_j$ and $R>0$, the probability that the points of the ball $B_X(x_j,R)$ belong to different clusters is bounded by $O(\ddim_j)\cdot\frac{R}{\Delta}$, where $\ddim_j$ is the doubling dimension of the metric space induced by $x_1,\dots,x_j$ (not required to be known in advance).
	\end{enumerate}
	In addition, the ball $B(x_1,\frac{\Delta}{4})$ is guaranteed to be fully contained in a single cluster.
\end{theorem}
Property~2 above is called the \emph{padding property}.
\paragraph{Truncated Exponential Distributions.} Similarly to previous clustering algorithms (e.g. \cite{Bartal96,ABN11, AGGNT19, Fil19Approx,BFU20}), we will use a truncated exponential distribution. That is, exponential distribution conditioned on the event that the outcome lies in a certain interval.
The \emph{$[\theta_1,\theta_2]$-truncated exponential distribution} with
parameter $\lambda$ is denoted by $\Texp_{[\theta_1, \theta_2]}(\lambda)$, and
the density function is:
$f(y)= \frac{ \lambda\, e^{-\lambda\cdot y} }{e^{-\lambda \cdot \theta_1} - e^{-\lambda
		\cdot \theta_2}}$, for $y \in [\theta_1, \theta_2]$.

Points arrive one by one $x_1,\dots,x_j,\dots$ in an online fashion, and for each new point the distances to previous points are revealed. We will maintain a net $N$ with minimum pairwise distance at least $\frac\Delta8$ such that every point has a net point at distance at most $\frac\Delta4$.
The set $N$ will increase monotonically (once a point joins $N$, it will remain there forever). Such a set can be constructed greedily, see~\Cref{lem:nets}.

We will maintain an estimate $\est^j_{\sddim}$ of the doubling dimension of the metric space induced by $\{x_1,\dots,x_j\}$. There is a polynomial-time algorithm, providing a constant approximation of the doubling dimension \cite{HM06}.
Here $\est^j_{\sddim}$ will be a positive integer for every $j$. \footnote{By definition, every metric space with a single point has doubling dimension $0$, while every metric space with two different points has doubling dimension $1$. We ignore the issue of doubling dimension $0$, in any case, our algorithm will provide distortion (worst case) $O(1)$ w.r.t.\ $x_1$.}
In addition, we will ensure that the estimates are monotonically non-decreasing. That is, $\est^j_{\sddim}$ will be the maximum estimate returned by the \cite{HM06} algorithm on any prefix.
Note that it is possible that the doubling dimension of a sub-metric $\{x_1,\dots,x_k\}$ is larger than the doubling dimension of the metric itself $\{x_1,\dots,x_n\}$. However, the the doubling dimensions can decrease by at most a constant factor (for details, see e.g.\ the comment after Definition~3.2 in \cite{GK13}). We can (implicitly) blow up the \cite{HM06} estimate by this constant factor, and conclude the following:
\begin{claim}\label{clm:ddimEst}
	For every input prefix $\{x_1,\dots,x_k\}$ defining a sub-metric of doubling dimension $\ddim_k$, it holds that $\ddim_k\le\est^k_{\sddim}\le O(\ddim_k)$.
\end{claim}

For every newly arriving point $x_j$ joining $N$, let $\lambda_j=4\cdot\est^j_{\sddim}$. We will sample a radius parameter $r_j\sim  \Texp_{[1,2]}(\lambda_j)$.
Note that all the sampled radii for all metric points are in $[1,2]$, while the parameter $\lambda_j$ of the truncated exponential distribution is monotonically non-decreasing.
We define a partition $\cP^j$ as follows: $P^j_{1}=B_X(x_1,\frac{\Delta}{4}\cdot r_1)$ is all the previously revealed points at distance at most $\frac{\Delta}{4}\cdot r_1$ from $x_1$. Generally,
\begin{equation}
	P^j_{q}=\begin{cases}
		\emptyset & \mbox{\rm if } x_{q}\notin N,\\
		B_{X}(x_{q},\frac{\Delta}{4}\cdot r_{q})\setminus\bigcup_{q'<q}P^j_{q'} &  \mbox{\rm if } x_{q}\in N.
	\end{cases}\label{eq:PartitionDef}
\end{equation}
In words, each net point $x_q\in N$ creates a cluster which is the ball of radius $\frac{\Delta}{4}\cdot r_{q}$ around $x_q$, minus all the previously created clusters.
Note that the cluster center $x_q$ does not necessarily belong to the cluster $P^j_{q}$ it defines.
Furthermore, the clusters defined w.r.t.\ the entire metric space $(X,d_X)$, even though so far we have seen only $j$ points. In particular, if a point joins a cluster, it will stay there forever.
The diameter of each cluster is at most $2\cdot \frac\Delta4\cdot\max_q r_q\leq 2\cdot \frac\Delta4\cdot2=\Delta$.
The ball of radius $\frac\Delta4$ around $x_1$ is fully contained in the cluster $P_1$.
Finally, every (revealed) point $x_j$ belong to some cluster in $\cP^j$. Indeed, this is as $N$ contains a point $x_q$ at distance $\frac{\Delta}{4}$ from $x_j$, thus $x_j$ will join $P^j_{q}$ if it has not joined any previous cluster.

Next, we prove the padding property. Consider a point $x_{k}$, and parameter $R>0$ such that $R\le \frac{\Delta}{2\cdot\lambda_k}$,
and let $B=B_X(x_k,R)$.
Note that for larger parameter $R$ there is nothing to prove, as clearly the probability that the points in $B$ belong to different clusters is bounded by $O(\ddim_k)\cdot\frac{R}{\Delta}\ge 1$.

Denote by $\cQ_j$ the event that $\frac{\Delta}{4}\cdot r_j\ge d_X(x_j,x_k)-R$. Note that if $\cQ_j$ has not occurred, then no point of $B$ will join $P_j$.
Denote by $\cF_{j}=\cQ_{j}\cap\bigcap_{j'<j}\overline{\cQ_{j'}}$
the event that $j$ is the first index such that $\cQ_{j}$ occurred.
Denote by $\mathcal{C}_j$ the event that $\mathcal{F}_j$ occurred and $\frac{\Delta}{4}\cdot r_{j}<d_{X}(x_{j},x_{k})+R$. That is, for every $j'<j$, $\cQ_{j'}$ has not occurred, and $\frac{\Delta}{4}\cdot r_{j}\in\left[d_{X}(x_{j},x_{k})-R,d_{X}(x_{j},x_{k})+R\right)$.
First, we show that if none of the events $\{\cC_j\}_j$ occurred, then $B$ is contained in a single cluster.
\begin{claim}\label{clm:CuttingEventsUnion}
	Denote by $\Xi$ the event that the points of the ball $B$ belong to different clusters. Then $\Xi\subseteq\bigcup_{j\ge1}\cC_j$.
\end{claim}
\begin{proof}
	Assume that none of the events $\{\cC_j\}_j$ has occurred. We want to show that all the points in $B$ belong to a single cluster.
	Let $j$ be the index such that $\cF_j$ occurred. For every $j'<j$, it holds that $\frac{\Delta}{4}\cdot r_{j'}<d_{X}(x_{j'},x_{k})-R$. In particular, for every $y\in B$, we have
	$d_{X}(x_{j'},y)\ge d_{X}(x_{j'},x_{k})-d_{X}(x_{k},y)\ge d_{X}(x_{j'},x_{k})-R>\frac{\Delta}{4}\cdot r_{j'}$. It follows that $B\cap P_{j'}=\emptyset$.
	As $\cF_j$ occurred but $\cC_j$ did not, then it holds that $\frac{\Delta}{4}\cdot r_{j}\ge d_{X}(x_{j},x_{k})+R$. In particular, for every $y\in B$, $d_{X}(x_{j},y)\le d_{X}(x_{j},x_{k})+d_{X}(x_{k},y)\le d_{X}(x_{j},x_{k})+R\le\frac{\Delta}{4}\cdot r_{j}$. It follows that
	$y\in P_{j}=B_{X}(x_{j},\frac{\Delta}{4}\cdot r_{j})\setminus\bigcup_{j'<j}P_{j'}$, and thus $B\subseteq P_j$. Hence $\Xi$ has not occurred, as required.	
\end{proof}

Let $N'=\{x_i\in N : i\le k\mbox{ and } d_X(x_i,x_k)\le \frac{\Delta}{2}+R\}$ be all the net points arriving before $x_k$, at distance at most $\frac{\Delta}{2}+R$ from $x_k$.
Since $N$ is a net, there is a point $x_q\in N$  such that $q\leq k$ and $d_X(x_i,x_k)\leq \frac{\Delta}{4}$; then we have $x_q\in N'$ and in particular $N'$ is nonempty.
Note that $\cF_j$ will necessarily occur for some $x_j\in N'$.
Indeed, if $x_{k}\notin \bigcup_{j<q}P_{j}$, then $x_{k}$ will join $P_{q}$.
Furthermore, for every net point $x_q\in \{x_1,\dots,x_k\}\setminus N'$ it holds that $\frac{\Delta}{4}\cdot r_{q}\le\frac{\Delta}{2}<d_{X}(x_{q},x_{k})-R$, and thus $\cF_q$ will not occur.
It follows that for every $x_j\notin N'$, $\mathcal{F}_j=\mathcal{C}_j=\emptyset$.

\begin{claim}
	Let $\gamma=\frac{8R}{\Delta}$. Then for every $j$, we have
	$\Pr\left[\mathcal{C}_{j}\right]\le\left(1-e^{-\lambda_j\cdot\gamma}\right)\left(\Pr\left[\mathcal{F}_{j}\right]+\frac{1}{e^{\lambda_j}-1}\right)$.
\end{claim}
\begin{proof}
	If one of the events $\{\cQ_{j'}\}_{j'<j}$ has occurred, then $\mathcal{C}_j=\emptyset$ and we are done. We thus can assume that none of them occurred.
	Similarly, we can assume that $x_j\in N'$ (as otherwise $\Pr[\mathcal{C}_j]=0$).
	Let $\rho=\frac{4}{\Delta}\cdot\left(d_{X}(x_{j},x_{k})-R\right)$.
	Note that $\rho$ is the minimal value such that if $r_j\ge\rho$, then $\cF_j$ will occur.
	As $x_j\in N'$, necessarily $\rho\le\frac{4}{\Delta}\cdot\left(\frac{\Delta}{2}+R-R\right)=2$.
	Set $\tilde{\rho}=\max\{\rho,1\}$.
	It holds that
	\[
	\Pr\left[\mathcal{F}_{j}\right]=\Pr\left[r_{j}\ge\rho\right]=\Pr\left[r_{j}\ge\tilde{\rho}\right]=\int_{\tilde{\rho}}^{2}\frac{\lambda_{j}\cdot e^{-\lambda_{j}\cdot y}}{e^{-\lambda_{j}}-e^{-2\lambda_{j}}}dy=\frac{e^{-\lambda_{j}\cdot\tilde{\rho}}-e^{-2\lambda_{j}}}{e^{-2\lambda_{j}}-e^{-2\lambda_{j}}}~.
	\]
	Event $\cC_j$ occurs if and only if $r_{j}\ge\rho$, and $r_{j}<\frac{4}{\Delta}\cdot\left(d_{X}(x_{j},x_{k})+R\right)=\rho+\frac{4}{\Delta}\cdot2R=\rho+\gamma$.
	It follows that
	\begin{align*}
		\Pr\left[\mathcal{C}_{j}\right]=\Pr\left[\rho\le r_{j}<\rho+\gamma\right] & \le\Pr\left[\tilde{\rho}\le r_{j}<\tilde{\rho}+\gamma\right]\\
		& =\int_{\tilde{\rho}}^{\min\left\{ 2,\tilde{\rho}+\gamma\right\} }\frac{\lambda_{j}\cdot e^{-\lambda_{j}\cdot y}}{e^{-\lambda_{j}}-e^{-2\cdot\lambda_{j}}}dy\\
		& \le\frac{e^{-\lambda_{j}\cdot\tilde{\rho}}-e^{-\lambda_{j}\cdot(\tilde{\rho}+\gamma)}}{e^{-\lambda_{j}}-e^{-2\cdot\lambda_{j}}}\\
		& =\left(1-e^{-\lambda_{j}\cdot\gamma}\right)\cdot\frac{e^{-\lambda_{j}\cdot\tilde{\rho}}}{e^{-\lambda_{j}}-e^{-2\cdot\lambda_{j}}}\\
		& =\left(1-e^{-\lambda_{j}\cdot\gamma}\right)\cdot\left(\Pr\left[\mathcal{F}_{j}\right]+\frac{e^{-2\cdot\lambda_{j}}}{e^{-\lambda_{j}}-e^{-2\cdot\lambda_{j}}}\right)\\
		& =\left(1-e^{-\lambda_{j}\cdot\gamma}\right)\cdot\left(\Pr\left[\mathcal{F}_{j}\right]+\frac{1}{e^{\lambda_{j}}-1}\right)~.
		\qedhere\end{align*}
\end{proof}

The probability that at least one of the events $\{\cC_j\}_{j}$ occurred is thus
\begin{align}
	\Pr\left[\bigcup_{j}\mathcal{C}_{j}\right]=\sum_{x_{j}\in N_{v}}\Pr\left[\mathcal{C}_{j}\right] & \le\sum_{x_{j}\in N'}\left(1-e^{-\lambda_{j}\cdot\gamma}\right)\cdot\left(\Pr\left[\mathcal{F}_{j}\right]+\frac{1}{e^{\lambda_{j}}-1}\right)\nonumber\\
	& \le\left(1-e^{-\lambda_{k}\cdot\gamma}\right)\cdot\left(1+\sum_{x_{j}\in N'}\frac{1}{e^{\lambda_{j}}-1}\right)~,\label{eq:CutScale-i}
\end{align}
where the second inequality holds as $\lambda_j$ is monotonically non-decreasing in $j$, and the events $\{{\cal{F}}_j\}_{x_j\in N'}$ are mutually disjoint.
Denote by ${\cal N}_{s}=\left\{x_j\in N' : \est^j_{\sddim}=s\right\}$ the subset of net points in
$N'$ that used $s$ as their doubling dimension estimate.
By the packing property (\Cref{lem:doubling_packing}), as all the points in ${\cal N}_{s}$ have pairwise distances at least $\frac{\Delta}{8}$, contained in a ball of radius $\frac{\Delta}{4}\cdot2+R<\Delta$ centered at $x_q$, and lie in a space with doubling dimension at most $\ddim_j\le\est^j_{\sddim}=s$, it holds that
\[
|{\cal N}_{s}|
\le 2^{s\cdot\left\lceil \log\frac{2\Delta}{\nicefrac{\Delta}{8}}\right\rceil }
=2^{4s}=e^{s\cdot\ln(16)}.
\]
Denote by $\tilde{\lambda}_s=4s$ the $\Texp$ parameter used by all the points in ${\cal N}_{s}$.
We have
\[
\sum_{x_{j}\in N'}\frac{1}{e^{\lambda_{j}}-1}=\sum_{s\ge1}\sum_{x_{j}\in{\cal N}_{s}}\frac{1}{e^{\tilde{\lambda}_{s}}-1}=\sum_{s\ge1}\frac{|{\cal N}_{s}|}{e^{\tilde{\lambda}_{s}}-1}\le\sum_{s\ge1}\frac{e^{s\cdot\ln(16)}}{e^{\tilde{\lambda}_{s}}-1}~.
\]
Note that
\[
\frac{e^{s\cdot\ln(16)}}{e^{\tilde{\lambda}_{s}}-1}=\frac{e^{-s}\cdot e^{s\cdot\ln(16e)}}{e^{\tilde{\lambda}_{s}}-1}\overset{(*)}{<}\frac{e^{-s}\cdot(e^{\tilde{\lambda}_{s}}-1)}{e^{\tilde{\lambda}_{s}}-1}=e^{-s}~,
\]
where the inequality $^{(*)}$ holds as $\tilde{\lambda}_{s}=4s$.
It follows that,
\[
\sum_{x_{j}\in N'}\frac{1}{e^{\lambda_{j}}-1}
\le\sum_{s\ge1}e^{-s}
=\frac{1}{e-1}< e^{-\frac12}
\le e^{-\lambda_{k}\cdot\gamma}.
\]
Continuing from \Cref{eq:CutScale-i}, and using \Cref{clm:ddimEst}, we have
\begin{align}
	\Pr\left[\bigcup_{j}\mathcal{C}_{j}\right]
	&=\sum_{x_{j}\in N_{v}}\Pr\left[\mathcal{C}_{j}\right]
	\le\left(1-e^{-\lambda_{k}\cdot\gamma}\right)\cdot\left(1+e^{-\lambda_{k}\cdot\gamma}\right)\nonumber\\
	& =1-e^{-2\gamma\cdot\lambda_{k}}\le2\gamma\cdot\lambda_{k}=\lambda_{k}\cdot\frac{8R}{\Delta}=O(\ddim_k)\cdot \frac{R}{\Delta}~.\label{eq:UBcutProb}
\end{align}
\Cref{thm:LDDonline} now follows by \Cref{clm:CuttingEventsUnion}.

\subsection{Online Low Diameter Decompositions for Euclidean Space}
This subsection is devoted to maintaining low-diameter decomposition for Euclidean spaces:
\begin{theorem}\label{thm:LDDonlineEuclidean}
	Consider a sequence of points $x_1,x_2\dots\in\R^d$ arriving in an online fashion, given by an oblivious adversary. Given a parameter $\Delta>0$, there is an online algorithm sampling a partition $\cP$ of $\R^d$ in an online fashion (i.e., once a points joins a cluster it remains there forever) such that:
	\begin{enumerate}
		\item Every cluster has diameter at most $\Delta$ w.r.t.\ $\ell_2$.
		\item For every pair of points $\{x_j,x_k\}$, the probability that $x_j$ and $x_k$ belong to different clusters is bounded by $O(\sqrt{d})\cdot\frac{\|x_j-x_k\|_2}{\Delta}$.
	\end{enumerate}
	In addition, the ball $B(x_1,\frac{\Delta}{4})$ is guaranteed to be fully contained in a single cluster.
\end{theorem}
\begin{proof}
	Charikar \etal \cite{CCGGP98} constructed a partition of $\R^d$ into clusters of diameter $\Delta$ such that for every pair of points $u,v\in \R^d$, the probability that $u$ and $v$ belong to different clusters is bounded by $O(\sqrt{d})\cdot\frac{\|x_j-x_k\|_2}{\Delta}$.
	This decomposition works by simply picking points $y_1,y_2,\dots,$ u.a.r. and setting the clusters to be $C_i=B_2(y_i,\frac\Delta2)\setminus\bigcup_{j<i}B_2(y_j,\frac\Delta2)$.
	The crux of this decomposition is the fact that the ratio between the volumes of the union and the intersection of two balls is bounded by
	\[
	\frac{{\rm Vol}\left(B_{2}(y_{i},\frac{\Delta}{2})\cap B_{2}(y_{j},\frac{\Delta}{2})\right)}{{\rm Vol}\left(B_{2}(y_{i},\frac{\Delta}{2})\cup B_{2}(y_{j},\frac{\Delta}{2})\right)}\ge1-2\sqrt{d}\cdot\frac{\|y_{i}-y_{j}\|_{2}}{\Delta}
	~.
	\]
	The points $y_i$ and $y_j$ will be clustered together if and only if the first center to be chosen from $B_{2}(y_{i},\frac{\Delta}{2})\cup B_{2}(y_{j},\frac{\Delta}{2})$ belongs to $B_{2}(y_{i},\frac{\Delta}{2})\cap B_{2}(y_{j},\frac{\Delta}{2})$. This probability equals to the ratio between the volumes as above.
	
	This decomposition was later used to create locality sensitive hashing \cite{AI08}, locality-sensitive ordering \cite{Fil23}, and consistent hashing \cite{CFJKVY22} (a.k.a.\ sparse partitions).
	In particular, it is known \cite{CFJKVY22} that a partition of the entire space $\R^d$ (without any bounding box) can be sampled using only $\poly(d)$ space (that is, a function that given a point returns its cluster center).
	
	Our contribution here is the last point guaranteeing that the ball $B_2(x_1,\frac{\Delta}{4})$ will belong to a single cluster. This is obtained by a slight modification of \cite{CCGGP98}. We treat the \cite{CCGGP98} decomposition as a black box.
	The partition is created as follows: sample a radius $r\in[\frac\Delta4,\frac\Delta2]$ uniformly at random.
	All the points in $B_2(x_1,r)$ belong to a single cluster. The remaining points are partitioned into clusters according to \cite{CCGGP98}. Formally, let $\cP$ be the partition created by \cite{CCGGP98}. Our partition is $B_2(x_1,r)\cup\left\{P\setminus B_2(x_1,r) :  P\in\cP \right\}$.
	
	Clearly, every cluster has diameter at most $\Delta$, and the ball $B(x_1,\frac{\Delta}{4})$ is guaranteed to be fully contained in a single cluster.
	It remains to bound the probability that two points belong to different clusters. Consider two points $x_j,x_k$, and suppose w.l.o.g\ that $\|x_1-x_j\|_2\le \|x_1-x_k\|_2$.
	The points $x_j,x_k$ will belong to different clusters only if either only one of them belong to $B(x_1,r)$, of if they are separated by \cite{CCGGP98}.
	By the union bound and the triangle inequality, it follows that:
	\begin{align*}
		\Pr\left[x_{j},x_{k}\text{ are separated}\right] & \le\Pr\left[r\in\left[\|x_{1}-x_{j}\|_{2},\|x_{1}-x_{k}\|_{2}\right)\right]+\Pr\left[x_{j},x_{k}\text{ are separated by \cite{CCGGP98}}\right]\\
		& \le\frac{\|x_{1}-x_{k}\|_{2}-\|x_{1}-x_{j}\|_{2}}{\nicefrac{\Delta}{4}}+2\sqrt{d}\cdot\frac{\|x_{j}-x_{k}\|_{2}}{\Delta}=O(\sqrt{d})\cdot\frac{\|x_{j}-x_{k}\|_{2}}{\Delta}~.
		\qedhere\end{align*}
\end{proof}

\section{Online dominating Stochastic Embedding into Ultrametrics:\\
	Proofs of \Cref{thm:DoublingOnlineEmbedding} and \Cref{thm:EuclidineOnlineEmbeddingIntoHST}}\label{subsec:OnlineToHST}
We restate \Cref{thm:DoublingOnlineEmbedding} for convenience.
\OnlineEmbedingHST*
\begin{proof}[Proof of \Cref{thm:DoublingOnlineEmbedding}]
	For every $i\in\Z$, set $\Delta_i=2^i$.
	Using \Cref{lem:nets}, we maintain a nested sequence of nets $\dots\subseteq N_{-2}\subseteq N_{-1}\subseteq N_{0}\subseteq N_{1}\subseteq N_{2}\subseteq\dots$ in an online fashion, where $N_i$ has minimum pairwise distance $\frac{\Delta_i}{8}=2^{i-3}$, and such that every point has a net point at distance at most $\frac{\Delta_i}{4}=2^{i-2}$.
	Following  \Cref{thm:LDDonline}, we will maintain an estimate of the doubling dimension, and use it to sample parameters  $r_j\sim  \Texp_{[1,2]}(\lambda_j)$ for every newly arriving point $x_j$.
	For every $i\in N$, following \Cref{thm:LDDonline} and using the sampled radii and nets, we will obtain a low diameter decomposition $\cP_i$.
	Note that we use a single radii $r_j$ for all distance scales.
   Furthermore, all the nets we use are nested.
   Hence storing all the net-information and the radii will take us only $O(n)$ words. As a result, we obtain partitions for all possible scales $\{\cP_{i}\}_{i\in \Z}$. It is important that we assume infinitely many partitions, as we do not know the aspect ratio in advance. However, as all the points in $B_X(x_1,\frac\Delta4)$ are contained in a single cluster (for every scale $\Delta$), starting from some scale and up, we will have only a single cluster containing all points. Similarly, due to the diameter bound, starting from some scale and down, our partition will be into singletons.

	An ultrametric is defined by a hierarchical (laminar) partition. However, even though the partitions $\{\cP_{i}\}_{i\in \Z}$ are highly correlated, they are not necessarily nested.
	Finally, we will force these partitions to be nested in the natural way to obtain laminar partitions $\{\tilde{\cP}_{i}\}_{i\in \Z}$.
	Formally, $x,y\in X$ will belong to the same cluster of $\tilde{\cP}_{i}$ if and only if they belong to the same cluster in both $\cP_{i}$ and $\tilde{\cP}_{i+1}$. In particular, $\cP_{i}$ is a refinement of $\tilde{\cP}_{i}$, and thus the diameter of each cluster in $\tilde{\cP}_{i}$ is bounded by $\Delta_i$ as well.
	Clearly $\{\tilde{\cP}_{i}\}_{i\in \Z}$ are laminar, and naturally define an ultrametric $U$, where every cluster $C\in \tilde{\cP}_{i}$ is associated with an internal node with label $\Delta_i$.
	This finishes the construction of the ultrametric. Note that storing the nested nets, sampled radii, and the resulting ultrametric all take $O(n)$ space.
	
	Once a point $x_j$ joins a cluster in  $\tilde{\cP}_i$, it will remain there forever. It follows that no pairwise distance is ever changed, and hence our embedding is online.
	Further, consider a pair of points $x_{j},x_{j'}$ for which we defined $d_U(x_j,x_{j'})=2^{i}$. It holds that both $x_{j}$ and $x_{j'}$ belong to the same $i$-level cluster in $\tilde{\cP}_{i}$, and hence also in $\cP_{i}$.
	As each such cluster has diameter $2^i$, $d_X(x_{j},x_{j'})\le 2^i=d_U(x_{j},x_{j'})$, the resulting embedding is dominating. It remains to show the expected distortion (that is expected expansion) property.
	
	\sloppy Consider a pair of points $x_{k},x_{k'}$ where $k<k'$.
	Denote by $\Psi_i$ the event that $x_{k},x_{k'}$ are separated in $\cP_i$.
	Note that the event that  $x_{k},x_{k'}$ are separated in $\tilde{\cP}_i$ is $\bigcup_{i'>i}\Psi_{i'}$.
	Then $d_U(x_{k},x_{k'})=\Delta_{i+1}=2^{i+1}$, where $i$ is the maximal index such that $\Psi_i$ holds (this is, $\lca(x_{k},x_{k'})$ will be associated with an $i+1$ level cluster).
	Let $i_{\max}=\left\lceil \log\left(4\cdot\max\left\{ d_{X}(x_{1},x_{k}),d_{X}(x_{1},x_{k'})\right\} \right)\right\rceil$. By \Cref{thm:LDDonline}, for every $i\ge i_{\max}$ all the points at distance $\frac{\Delta_i}{4}\ge \frac{\Delta_{i_{\max}}}{4}\ge\max\{d_{X}(x_{1},x_{k}),d_{X}(x_{1},x_{k'})\}$ from $x_1$ are contained in a single cluster. In particular $x_{k},x_{k'}$ will be contained in the same cluster, and thus $\Psi_i$ cannot occur.
	
	Denote by $c$ a constant such that \Cref{thm:LDDonline} bounds the probability that the points of the ball $B_X(x_k,R)$ belong to a different cluster in $\cP_i$ by $c\cdot\ddim_k\cdot\frac{R}{2^i}$.
	If the ball of radius $d_X(x_k,x_{k'})$ around $x_k$ is contained in a single $\cP_i$ cluster, then it must hold that $x_k,x_{k'}$ belong to a single cluster, and thus $\Psi_i$ did not occur.
	by  \Cref{thm:LDDonline} it follows that $\Pr[\Psi_i]\le c\cdot\ddim_k\cdot\frac{d_X(x_k,x_{k'})}{2^i}$.
	Denote $i_{\min}=\left\lceil \log\left(c\cdot\ddim_{k}\cdot d_{X}(x_{k},x_{k'})\right)\right\rceil$.
	It holds that
	\begin{align*}
		\mathbb{E}\left[d_{U}(x_{k},x_{k'})\right] & =\sum_{i}\Pr\left[\Psi_{i}\text{ and }\overline{\cup_{i'\ge i}\Psi_{i'}}\right]\cdot2^{i+1}\\
		& \le\sum_{i<i_{\min}}\Pr\left[\Psi_{i}\text{ and }\overline{\cup_{i'\ge i}\Psi_{i'}}\right]\cdot2^{i+1}+\sum_{i=i_{\min}}^{i_{\max}-1}\Pr\left[\Psi_{i}\right]\cdot2^{i+1}\\
		& \le2^{i_{\min}+2}+\sum_{i=i_{\min}}^{i_{\max}-1}c\cdot\ddim_{k}\cdot\frac{d_{X}(x_{k},x_{k'})}{2^{i}}\cdot2^{i+1}\\
		& =4\cdot2^{i_{\min}}+2c\cdot\ddim_{k}\cdot(i_{\max}-i_{\min})\cdot d_{X}(x_{k},x_{k'})\\
		& =O\left(\ddim_{k}\cdot\log\Phi\right)\cdot d_{X}(x_{k},x_{k'})~,
	\end{align*}
	where the last equality holds as $i_{\max}-i_{\min}\le1+\log\left(\frac{4\cdot\max\{d_{X}(x_{1},x_{k}),d_{X}(x_{1},x_{k'})\}}{c\cdot\sddim_{k}\cdot d_{X}(x_{k},x_{k'})}\right)\le O(\log\Phi)$,
	where $\Phi$ is the aspect ratio.
\end{proof}

Consider the case where all the input points are from the Euclidean space $\R^d$.
Following the exact same lines as in the proof of \Cref{thm:DoublingOnlineEmbedding}, where we replace \Cref{thm:LDDonline} by \Cref{thm:LDDonlineEuclidean}, we obtain a quadratic improvement in the dependence on the dimension. For space considerations note that it is enough to store the partition from \Cref{thm:LDDonlineEuclidean} only for a single scale $\Delta=1$. This is as the only argument in the proof combining different scales is the union bound. Hence we can use the same partition (up to scaling) for all distance scales. We conclude that the ultrametric can be computed, and maintained using only $\poly(n,d)$ space.

\begin{restatable}[]{theorem}{OnlineEmbedingEuclideanToHST}
	\label{thm:EuclidineOnlineEmbeddingIntoHST}
	Given a sequence of points $x_1,x_2,\dots$ in Euclidean $d$-space $(\R^d,\|\cdot\|_2)$ arriving in an online fashion, there is a dominating stochastic metric embedding into ultrametrics  (2-HSTs) with expected distortion $O(\sqrt{d}\cdot\log\Phi)$, where $\Phi$ is the aspect ratio (unknown in advance).	
\end{restatable}

\begin{remark}\label{rem:HST}{\rm
    \Cref{thm:DoublingOnlineEmbedding} and \Cref{thm:EuclidineOnlineEmbeddingIntoHST} provide dominating stochastic embeddings into ultrametrics for a sequence of metric points arriving in an online fashion. A finite ultrametric is equivalent to an HST, and the HST can be covered from the metric. Since applications (e.g., online matchings) use the HST representation, we describe how the topology of the HST changes as new points arrive. The proof of \Cref{thm:DoublingOnlineEmbedding} describes a randomized online algorithm that maintains a laminar family of partitions $\{\tilde{\cP}_{i}\}_{i\in \Z}$ of the metric points $\{x_1,\ldots , x_i\}$. The partitions yields an infinite tree $\mathcal{T}$. We obtain an HST $T_i$ from $\mathcal{T}$ by suppressing all nodes of degree two, except for the top-level node of descending chains, which become leaves.

    When a new point $x_{i+1}$ arrives, a new descending chain is attached to a node in $\mathcal{T}$. This means that for all $i<j$, the HST $T_i$ is the tree induced by the leaves corresponding to $\{x_1,\ldots , x_i\}$ in $T_j$. Conversely, for all $i\geq 1$, one can construct $T_{i+1}$ from $T_i$ by one of the following operations: (1) Create a new root and attach $T_i$ and a new leaf $\ell_{i+1}$ corresponding to $x_{i+1}$ to the root; (2) attach a new leaf $\ell_{i+1}$ corresponding to $x_{i+1}$ to an existing node of $T_i$; (3) subdivide an edge of $T_i$ and attach a new leaf $\ell_{i+1}$ to the subdivision node; (4) attach two new leaves to a leaf $\ell_i$ of $T_i$, one corresponds to $x_{i+1}$ and the other to the point $x_j$ previously embedded at $\ell_i$.
}\end{remark}

\section{Online Embedding into Euclidean Space: Proof of \Cref{thm:DoublingOnlineEmbeddingtoEuclidean}}
This section is devoted to proving \Cref{thm:DoublingOnlineEmbeddingtoEuclidean}.
We refer to \Cref{subsec:techIdeas} for intuition regarding the proof.
We restate the theorem for convenience.
\OnlineEmbedingIntoEuclidean*

Let $X=\{x_1,x_2,\dots\}$ be the metric points according to the order of arrival.
We will begin by describing a random embedding $f:X\rightarrow\ell_2$ where each index $i\in\Z$ will have an associated coordinate. Here $f_i(x_j)$ is a random variable responsible for the scale $2^i$. Specifically,
$f(x_1)$ will equal $0$ in all the coordinates, and in general, $f(x_j)$ will also equal $0$ for all but $O(j)$ coordinates.
Our goal is to eventually obtain a deterministic embedding.
This will be achieved by letting $D_{j,k}=\E\left[\left\|f(x_j)-f(x_k)\right\|_2^2\right]$. This represents real Euclidean distances (squared), from which a deterministic embedding can be reconstructed.

Using \Cref{lem:nets}, we maintain a nested sequence of nets $\{N_i\}_{i\in\Z}$ in an online fashion.
For every $i\in\Z$, following \Cref{thm:LDDonline}, we create partitions $\{\cP_i\}_{i\in\Z}$. These partitions are not necessarily laminar.
Specifically, we maintain an estimate $\est_{\sddim}^j$ of the doubling dimension  of the prefix $\{x_1,\dots,x_j\}$, and for every arriving point $x_j$ we sample a radius $r_{j}\in\Texp_{[1,2]}(\lambda_j)$. The $i$'th partition is then defined to be $$\cP_{i}=\left\{ P_{q}^{i}=\left(B_{X}\left(x_{q},\frac{\Delta_i}{4}\cdot r_{q}\right)\setminus\bigcup_{q'<q}P_{q'}^{i}\right)\right\}_{x_{q}\in N_{i}},$$ where $\Delta_i=2^i$.
In addition,  we will sample independent boolean parameters $\alpha_{j}\in \{0,1\}$, getting each of the two values with probability $\frac12$.
These parameters will be used to ``zero'' out some clusters, which will ensure that each vector $f(x_j)$ is nonzero only in linearly many coordinates.
Specifically, for every cluster $P^i_q$, we define parameters $\alpha_{q,i}$ as follows. Let $\alpha_{1,i}=0$ for every $i$ (deterministically).
In general, for $q>1$, let $\tilde{i}_q$ be the maximum index such that $x_q\in N_{\tilde{i}_q}$ (recall that $x_q\in N_{i}$ for $i\le\tilde{i}_q$). We then set
\[
\alpha_{q,i}=\begin{cases}
	\alpha_{q} & \mbox{\rm for } i\ge\tilde{i}_q-2,\\
	0 &  \mbox{\rm for } i\le\tilde{i}_q-3.
\end{cases}
\]
Intuitively,
for $i>\tilde{i}_q$, $\alpha_{q,i}$ will be irrelevant (as $x_q\notin N_i$ and thus $P_{q}^{i}=\emptyset$), for the scales $\{\tilde{i}_q-2,\tilde{i}_q-1,\tilde{i}_q\}$ it will be either $0$ or $1$, randomly, and for all the scales bellow $\tilde{i}_q-3$ it will be ``zeroed'' out.

\begin{figure}[t]
	\centering
	\includegraphics[width=0.6\textwidth]{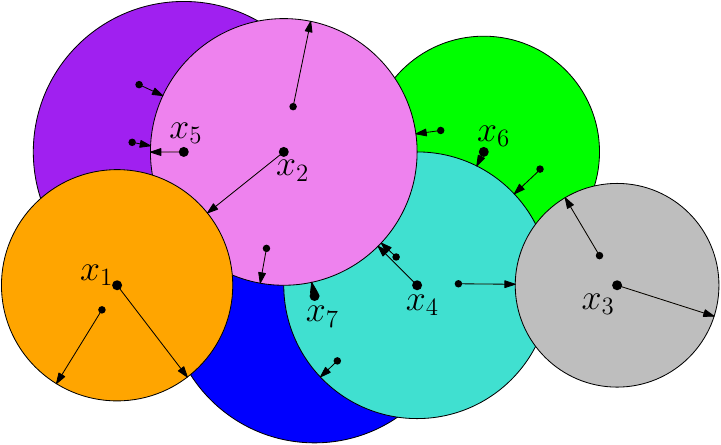}
	\caption{\small{ Illustration of the partition of a single scale $\Delta=2^i$ and the paddedness of each point. The source metric here $(X,d_X)$ is induced by the Euclidean plane. The net points $N_i$, and cluster centers are $\{x_1,\dots,x_7\}$, where their respective clusters $P_1^i,\dots,P_7^i$ are colored orange, violet, grey, turquoise, purple, green, and blue. The rest of the metric points are represented by smaller unnamed dots. The \emph{paddedness} is the distance to the boundary of the cluster each point belongs to. In the figure, the paddedness of each point is equal to the length of the dart originating from it.}}
	\label{fig:paddedness}
\end{figure}

The coordinate $f_{i}(x_j)$ will be created using the partition $\cP_{i}$, the boolean variables $\{\alpha_{j,i}\}_{x_j\in N_i}$, and the \emph{paddedness} parameter which we define next.
Consider a point $x_j$, and the partition $\cP_{i}$, where $x_j\in P^{i}_q$ joins the cluster centered at $x_q\in N_i$. The \emph{paddedness} of $x_j$ represents the smallest possible distance from $x_j$ to a point outside of the cluster $P^{i}_q$, including points that might arrive in the future. Let $N_{i}^{q}=N_{i}\cap\{x_{1},\dots,x_{q}\}$ be all the net points arriving before $x_q$ (the center of the cluster to which $x_j$ joined). Then the paddedness is set to be
\[
\partial_{i}(x_{j})\coloneqq\min_{x_{k}\in N_{i}^{q}}\left|r_{k}\cdot\frac{\Delta_{i}}{4}-d_{X}(x_{k},x_{j})\right|~;
\]
see \Cref{fig:paddedness} for an illustration.
Note that as $x_j$ joined the cluster centered at $x_q\in N_i$,  it holds that $d_{X}(x_{q},x_{j})\le r_{q}\cdot\frac{\Delta_{i}}{4}$. As $r_q\le2$, it follows that
$\partial_{i}(x_{j})\le\left|r_{q}\cdot\frac{\Delta_{i}}{4}-d_{X}(x_{q},x_{j})\right|\le r_{q}\cdot\frac{\Delta_{i}}{4}\le\frac{\Delta_{i}}{2}$.
In addition, note that the paddedness is defined in the same way for all points in $P^{i}_q$, regardless of their arriving time (index).  In a sense, this ensures that the paddedness is ``continuous'' in each cluster. We will actually prove something stronger. In particular, \Cref{clm:BoundCoordinate} below shows that the paddedness is a Lipschitz function (w.r.t.\ all the points).

The $i$'th coordinate $f_{i}(x_{j})$ is defined to be:
\[
f_{i}(x_{j})\coloneqq\alpha_{q,i}\cdot\partial_{i}(x_{j})
\]
That is, $x_j$ is either sent to its paddedness parameter, or set to $0$, where this decision is consistent in each cluster. Note also that all the points in the cluster centered at $x_1$ are always sent to $0$ (as $\alpha_{1}=0$). This ensures that $f_{i}(x_{j})=0$ for every index $i$ such that $\frac{\Delta_i}{4}\ge d_X(x_1,x_j)$ (equivalently, for all $i\ge \log d_X(x_1,x_j)+2$ ).
In \Cref{clm:paddedness} and \Cref{clm:BoundCoordinate} below, we prove that $f_i:X\rightarrow\R_{\ge0}$ is Lipschitz (regardless of the random choices).

\begin{claim}\label{clm:paddedness}
	For every $i\in\Z$ and $x_{j}\in P_q^{i}\in\cP_i$, it holds that $\partial_{i}(x_{j})\le\min_{x_{j'}\notin P_q^{i}}d_X(x_j,x_{j'})$, where the minimum is taken over all the points $x_{j'}\in X\setminus P^i_q$ (including points $x_{j'}$ arriving after $x_j$).
\end{claim}
\begin{proof}
	Consider $x_{j'}\in P_{q'}^{i}$ for $q'\ne q$. We distinguish between two cases:
	\begin{itemize}
		\item If $q'<q$, as $x_{j}\notin P_{q'}^{i}$ and $x_{j'}\in P_{q'}^{i}$ it
		holds that $d_{X}(x_{q'},x_{j'})\le r_{q'}\cdot\frac{\Delta_{i}}{4}<d_{X}(x_{q'},x_{j})$, where the second inequality holds as $x_j$ ``considered'' joining $P_{q'}^i$.
		Thus 
		\[
		\partial_{i}(x_{j})\le d_{X}(x_{q'},x_{j})-r_{q'}\cdot\frac{\Delta_{i}}{4}\le d_{X}(x_{q'},x_{j})-d_{X}(x_{q'},x_{j'})\le d_{X}(x_{j},x_{j'})~.
		\]
		\item Else $q'>q$, as $x_{j}\in P_{q}^{i}$ and $x_{j'}\notin P_{q}^{i}$ it holds
		that $d_{X}(x_{q},x_{j})\le r_{q}\cdot\frac{\Delta_{i}}{4}<d_{X}(x_{q},x_{j'})$, where the second inequality holds as $x_{j'}$ ``considered'' joining $P_{q}^i$.
		Thus 
		\[
		\partial_{i}(x_{j})\le r_{q}\cdot\frac{\Delta_{i}}{4}-d_{X}(x_{q},x_{j})\le d_{X}(x_{q},x_{j'})-d_{X}(x_{q},x_{j})\le d_{X}(x_{j},x_{j'})~.
		\qedhere
        \]
	\end{itemize}
\end{proof}

\begin{claim}\label{clm:BoundCoordinate}
	For every $i\in\Z$, and $x_{j},x_{j'}$, it holds that $|f_{i}(x_{j})-f_{i}(x_{j'})|\le d_X(x_{j},x_{j'})$.
\end{claim}
\begin{proof}
	Let $q,q'\in N_i$ such that $x_{j}\in P_{q}^{i}$ and $x_{j'}\in P_{q'}^{i}$. Suppose first that $q\ne q'$, that is, $x_j$ and $x_{j'}$ belong to different clusters in $\cP_{i}$. Then by \Cref{clm:paddedness}, we have
	\[
	\left|f_{i}(x_{j})-f_{i}(x_{j'})\right|=\left|\alpha_{q,i}\cdot\partial_{i}(x_{j})-\alpha_{q',i}\cdot\partial_{i}(x_{j'})\right|\le\max\left\{ \partial_{i}(x_{j}),\partial_{i}(x_{j'})\right\} \le d_{X}(x_{j},x_{j'})~.
	\]
	Else, $q=q'$, and thus $x_j$ and $x_{j'}$ belong to the same cluster $P_{q}^{i}\in\cP_{i}$.
	In particular $N_{i}^{q}=N_{i}^{q'}$.
	Suppose w.l.o.g.\ that $\partial_{i}(x_{j})\ge\partial_{i}(x_{j'})$.
	It holds that
	\begin{align*}
		\left|f_{i}(x_{j})-f_{i}(x_{j'})\right| & =\left|\alpha_{q,i}\cdot\partial_{i}(x_{j})-\alpha_{q,i}\cdot\partial_{i}(x_{j'})\right|\le\left|\partial_{i}(x_{j})-\partial_{i}(x_{j'})\right|=\partial_{i}(x_{j})-\partial_{i}(x_{j'})~.
	\end{align*}
	Let $x_{s}\in N_{i}^{q}$ be the net point realizing the minimum in
	the definition of paddedness of $x_{j'}$, that is, $\partial_{i}(x_{j'})=\left|r_{s}\cdot\frac{\Delta_{i}}{4}-d_{X}(x_{s},x_{j'})\right|$.
	Suppose first that $s\ne q$.
	Since $s\in N_i^q$, then $s\le q$. Hence both $x_j$ and $x_{j'}$ had the opportunity to join $P^i_s$, but joined $P_{q}^{i}$. It follows that
	$d_{X}(x_{s},x_{j}),d_{X}(x_{s},x_{j'})>r_{s}\cdot\frac{\Delta_{i}}{4}$. It holds
	that
	\begin{align*}
		\partial_{i}(x_{j})-\partial_{i}(x_{j'}) & =\min_{x_{k}\in N_{i}^{q}}\left|r_{k}\cdot\frac{\Delta_{i}}{4}-d_{X}(x_{k},x_{j})\right|-\left(d_{X}(x_{s},x_{j'})-r_{s}\cdot\frac{\Delta_{i}}{4}\right)\\
		& \le\left(d_{X}(x_{s},x_{j})-r_{s}\cdot\frac{\Delta_{i}}{4}\right)-\left(d_{X}(x_{s},x_{j'})-r_{s}\cdot\frac{\Delta_{i}}{4}\right)\le d_{X}(x_{j},x_{j'})~.
	\end{align*}
	Otherwise, if $s=q$, it holds that
	\begin{align*}
		\partial_{i}(x_{j})-\partial_{i}(x_{j'}) & =\min_{x_{k}\in N_{i}^{q}}\left|r_{k}\cdot\frac{\Delta_{i}}{4}-d_{X}(x_{k},x_{j})\right|-\left(r_{q}\cdot\frac{\Delta_{i}}{4}-d_{X}(x_{q},x_{j'})\right)\\
		& \le\left(r_{q}\cdot\frac{\Delta_{i}}{4}-d_{X}(x_{q},x_{j})\right)-\left(r_{q}\cdot\frac{\Delta_{i}}{4}-d_{X}(x_{q},x_{j'})\right)\le d_{X}(x_{j},x_{j'})~.
	\end{align*}
	In both cases, we have shown that $\left|f_{i}(x_{j})-f_{i}(x_{j'})\right|\le d_{X}(x_{j},x_{j'})$, as required.
\end{proof}

In the next claim we show that with high enough probability, the paddedness will be of a ``significant'' size (which will be crucial for the lower bound side of our proof).

\begin{claim}\label{clm:paddingProb}
	For every $i\in\Z$ and $x_{j}$, let $\ddim_j$ be the doubling dimension of the metric space induced by the prefix $\{x_1,\dots,x_j\}$.
	There exists a universal constant $c$ such that $\Pr[\partial_{i}(x_{j})\ge c\cdot\frac{\Delta_i}{\sddim_j}]\ge\frac78$.
\end{claim}
	By \Cref{thm:LDDonline}, the probability that the points in the ball $B_X(x_j,\Theta(\frac{\Delta_i}{\sddim}))$ belong to different clusters is at most $\frac18$ (for an appropriate constant inside the $\Theta$ in the definition of $R$).
	Intuitively, this should be enough to prove the claim.
	Unfortunately, for a formal proof, we will need to inspect the proof of \Cref{thm:LDDonline} more closely.
\begin{proof}
	Recall the proof of \Cref{thm:LDDonline}. Set $R=\Theta(\frac{\Delta_i}{\sddim})$. Using the terminology used in the proof of \Cref{thm:LDDonline},
	 $\cC_q$ is the event that for every $k<q$, we have $\frac{\Delta}{4}\cdot r_{k}<d_{X}(x_{j},x_{k})-R$ and $d_{X}(x_{j},x_{q})-R\le\frac{\Delta}{4}\cdot r_{q}<d_{X}(x_{j},x_{q})+R$.
	Observe that if none of the events $\{\cC_q\}_q$ occurred, then $\partial_{i}(x_{j})> R$.
	Indeed, in this case, suppose $x_j\in P^j_q$.
	For $k<q$, as $\cC_k$ has not occurred,  $d_{X}(x_{k},x_{j})-r_{k}\cdot\frac{\Delta_{i}}{4}>\left(r_{k}\cdot\frac{\Delta_{i}}{4}+R\right)-r_{k}\cdot\frac{\Delta_{i}}{4}=R$. From the other hand, as $\cC_q$ has not  occurred, $r_{q}\cdot\frac{\Delta_{i}}{4}-d_{X}(x_{q},x_{j})>r_{q}\cdot\frac{\Delta_{i}}{4}-\left(\frac{\Delta}{4}\cdot r_{q}-R\right)=R$.
	It follows that $\partial_{i}(x_{j})\coloneqq\min_{x_{k}\in N_{i}^{q}}\left|r_{k}\cdot\frac{\Delta_{i}}{4}-d_{X}(x_{k},x_{j})\right|> R$.
	
	By \Cref{eq:UBcutProb}, it holds that $\Pr\left[\bigcup_{j}\mathcal{C}_{j}\right]=O(\ddim_k)\cdot \frac{R}{\Delta_i}\le\frac18$, where the last equation holds for an appropriate constant in the definition of $R$.
\end{proof}

\Cref{clm:BoundCoordinate} and \Cref{clm:paddingProb} are already enough to prove that our embedding has expected distortion $O(\ddim\cdot\log\Phi)$. However, as $\Phi$ may be unbounded, we will also bound the distortion by a function of $n$.

\begin{claim}\label{clm:nonZeroCoordinates}	
	For each $j$, $f(x_j)$ is nonzero in at most $3j$ coordinates.
\end{claim}
\begin{proof}
	Fix $q\le j$, and let $\cI_q=\left\{ i\in\Z : x_{j}\in P_{q}^{i}\text{ and }\alpha_{q,i}=1\right\}$ be the set of scales $i$ such that $x_j$ belongs to the cluster $P_q^i$ centered at $x_q$ where $\alpha_{q,i}\ne0$. We first argue that $|\cI_q|\le3$.
	
	Following the definition of $\alpha_{q,i}$, recall that $\tilde{i}_q$ is the maximum scale such that $x_q\in N_{\tilde{i}_q}$.
	For every $i>\tilde{i}_q$, as $x_j$ joins a cluster centered in a net point, clearly $x_{j}\notin P_{q}^{i}$.
	For $i\le\tilde{i}_q-3$, by definition $\alpha_{q,i}=0$. Thus there are exactly three scales in which $x_j$ could potentially join $P_{q}^{i}$, where $\alpha_{q,i}$ could possibly equal $1$.
	
	The claim now follows, as the subset of coordinates $i$ where $f_i(x_j)\ne 0$ is a subset of $\bigcup_{q\le j}\cI_q$.
\end{proof}

The next lemma bounds the expected distortion.
\begin{lemma}\label{lem:EuclideanEmbeddingSingleSample}
	Consider two points $x_j,x_q$, and let $\ddim$ be the doubling dimension of the metric space $(X,d_X)$.
	It holds that
	\begin{itemize}
		\item $\|f(x_j)-f(x_{q})\|^2_2= O\left(\min\left\{\log\Phi,n\right\}\right)\cdot d^2_X(x_j,x_{q})$, regardless of random choices; and
		\item  $\E\left[\|f(x_j)-f(x_{q})\|^2_2\right]=\Omega(\frac{1}{\sddim^2})\cdot d^2_X(x_j,x_{q})$.
	\end{itemize}
\end{lemma}
\begin{proof}
	We begin by proving the upper bound w.r.t.\ $\log \Phi$.
	Let
 $$i_{\max}=\left\lceil \log\left(4\cdot\max\left\{ d_{X}(x_{1},x_{j}),d_{X}(x_{1},x_{q})\right\} \right)\right\rceil.$$
 By \Cref{thm:LDDonline}, for every $i\ge i_{\max}$ all the points at distance $\frac{\Delta_i}{4}\ge\max\{d_{X}(x_{1},x_{j}),d_{X}(x_{1},x_{q})\}$ from $x_1$ are contained in the cluster $P_1^i$. In particular $x_{j},x_{q}\in P_1^i$. As $\alpha_{1,i}=0$ by definition, it follows that $f_{i}(x_j)=f_{i}(x_{q})=0$. Thus the contribution of all these coordinates is $0$.
	
Let $i_{\min}=\left\lceil \log d_{X}(x_{j},x_{q})\right\rceil$. Now \Cref{clm:BoundCoordinate} yields
    \begin{align}
		\left\Vert f(x_{j})-f(x_{q})\right\Vert _{2}^{2} & =\sum_{i\le i_{\min}}|f_{i}(x_{j})-f_{i}(x_{q})|^{2}+\sum_{i=i_{\min}+1}^{i_{\max}}|f_{i}(x_{j})-f_{i}(x_{q})|^{2}\nonumber\\
		& \le\sum_{i\le i_{\min}}\left(\frac{\Delta_{i}}{2}\right)^{2}+\sum_{i=i_{\min}+1}^{i_{\max}}\left(d_{X}(x_{j},x_{q})\right)^{2}\nonumber\\
		& =O(1)\cdot2^{2i_{\min}}+\left(i_{\max}-i_{\min}\right)\cdot\left(d_{X}(x_{j},x_{q})\right)^{2}\nonumber\\
		& =O(\log\Phi)\cdot\left(d_{X}(x_{j},x_{q})\right)^{2}~,\label{eq:first}
	\end{align}
	where the last equality holds as $i_{\max}-i_{\min}\le\log\left(\frac{4\cdot\max\{d_{X}(x_{1},x_{j}),d_{X}(x_{1},x_{q})\}}{d_{X}(x_{j},x_{q})}\right)+1\le O(\log\Phi)$.
	It follows that $\left\Vert f(x_{j})-f(x_{q})\right\Vert_{2}^2\le O(\log\Phi)\cdot d^2_X(x_j,x_{q})$.
	
	Next we prove the second part of the upper bound. Let $\cI\subset \Z$ be the set of scales where either $f_i(x_j)$ or $f_i(x_q)$ is nonzero.  \Cref{clm:nonZeroCoordinates} yields $|\cI|\le(j+q)\cdot3<6n$. Then
   \Cref{clm:BoundCoordinate} implies that
	\begin{equation}\label{eq:second}
	\left\Vert f(x_{j})-f(x_{q})\right\Vert _{2}^{2}=\sum_{i\in\cI}\left|f_{i}(x_{j})-f_{i}(x_{q})\right|^{2}\le\sum_{i\in\cI}\left(d_{X}(x_{j},x_{q})\right)^{2}\le6 n\cdot\left(d_{X}(x_{j},x_{q})\right)^{2}~.
	\end{equation}
	The combination of \cref{eq:first,eq:second} yields $\|f(x_j)-f(x_{q})\|^2_2= O\left(\min\left\{\log\Phi,n\right\}\right)\cdot d^2_X(x_j,x_{q})$.

	Finally, we prove the lower bound. 	Let
 \[
 i^*=\left\lfloor \log\frac{d_{X}(x_{j},x_{q})}{2}\right\rfloor~.
 \]
 Note that $\Delta_{i^{*}}=2^{i^{*}}>2^{\log\frac{d_{X}(x_{j},x_{q})}{2}-1}$,
	or $d_{X}(x_{j},x_{q})<4\cdot\Delta_{i^{*}}$, and $\Delta_{i^{*}}\le2^{ \log\frac{d_{X}(x_{j},x_{q})}{2} }=\frac{d_{X}(x_{j},x_{q})}{2}$.
	As the partition $\cP_i$ is $\Delta_{i^{*}}$-bounded, $x_q$ and $x_j$ will belong to different clusters.
	Let $x_{s_j},x_{s_q}\in N_{i^*}$  be the net points such that $x_j\in P^{i^*}_{s_j}$ and $x_{q}\in P^{i^*}_{s_q}$.
	
	Let $\Psi_\alpha$ be the event that $\alpha_{s_j,i^*}\ne \alpha_{s_q,i^*}$.
	We argue that $\Pr[\Psi_\alpha]=\frac12$, regardless of the centers $x_{s_j}$ and $x_{s_q}$ (which are determined by the choice of the radii).
	Recall that $\alpha_{s_q,i^*}$ (resp., $\alpha_{s_j,i^*}$) equals $1$ with probability $\frac12$ iff $i^*>\tilde{i}_q-3$ (resp., $i^*>\tilde{i}_j-3$).
	If $i^*>\max\{\tilde{i}_q,\tilde{i}_j\}-3$, then clearly $\Pr[\Psi_\alpha]=\frac12$. If (w.l.o.g.)
	$\tilde{i}_q-3<	i^*\le \tilde{i}_j-3$, then $\alpha_{s_j,i^*}=0$ while $\alpha_{s_q,i^*}$ equals $1$ with probability $\frac12$, as required.
	The only problematic case is when $i^*\le\min\{\tilde{i}_q,\tilde{i}_j\}-3$. In that case, however, we have $x_{s_{q}},x_{s_{j}}\in N_{i^{*}+3}$, while the triangle inequality yields
	\begin{align*}
		d_{X}(x_{s_{j}},x_{s_{q}}) & \le d_{X}(x_{s_{j}},x_{j})+d_{X}(x_{j},x_{q})+d_{X}(x_{q},x_{s_{q}})\\
		& \le\frac{\Delta_{i^{*}}}{2}+4\Delta_{i^{*}}+\frac{\Delta_{i^{*}}}{2}=5\Delta_{i^{*}}<2^{i^{*}+3}~,
	\end{align*}
	which contradicts the assumption that $N_{i^{*}+3}$ is a $2^{i^{*}+3}$-net.
	
	Let $\Psi_j$ (resp., $\Psi_q$) be the event that $\partial_{i}(x_{j})\ge c\cdot\frac{\Delta_{i}}{\sddim}$ (resp., $\partial_{i}(x_{q})\ge c\cdot\frac{\Delta_{i}}{\sddim}$), where $c$ is the constant from \Cref{clm:paddingProb}. Then by \Cref{clm:paddingProb}, $\Pr[\Psi_j],\Pr[\Psi_q]\ge\frac78$. Let $\Psi$ be the event that all three events $\Psi_\alpha,\Psi_j,\Psi_q$ occur simultaneously. Then by the union bound $\Pr\left[\overline{\Psi}\right]\le\Pr\left[\overline{\Psi}_{\alpha}\right]+\Pr\left[\overline{\Psi}_{j}\right]+\Pr\left[\overline{\Psi}_{q}\right]\le\frac{1}{2}+\frac{1}{8}+\frac{1}{8}=\frac{3}{4}$,
	implying that $\Pr\left[\Psi\right]\ge\frac{1}{4}$.
	
	If the event $\Psi$ indeed occurred, then
	\begin{align*}
		\left\Vert f(x_{j})-f(x_{q})\right\Vert _{2} & \ge\left|f_{i^{*}}(x_{j})-f_{i^{*}}(x_{q})\right|=\left|\alpha_{s_{j}}\cdot\partial_{i^{*}}(x_{j})-\alpha_{s_{q}}\cdot\partial_{i^{*}}(x_{q})\right|\\
		& \ge\min\left\{ \partial_{i^{*}}(x_{j}),\partial_{i^{*}}(x_{q})\right\} \ge c\cdot\frac{\Delta_{i^{*}}}{\sddim}
        =\Omega\left(\frac{1}{\sddim}\right)\cdot d_{X}(x_{j},x_{q})~.
	\end{align*}
	We conclude that
	\[
	\E\left[\|f(x_{j})-f(x_{q})\|_{2}^{2}\right]
     \ge\E\left[\|f(x_{j})-f(x_{q})\|_{2}^{2}\Big\mid\Psi\right]\cdot\Pr\left[\Psi\right]
     \ge\Omega\left(\frac{1}{\sddim^{2}}\right)\cdot d_{X}^{2}(x_{j},x_{q})~,
	\]
	as required.
\end{proof}

For every $j,q\in[n]$, set $D_{j,q}=\E\left[\|f(x_j)-f(x_{q})\|^2_2\right]$.
\begin{claim}\label{clm:EuclideanMatrix}
	 $\{D_{j,q}\}_{j,q\in[n]}$ is an Euclidean distance matrix. That is, there are some points $y_1,\dots,y_n\in\ell_2$ of some arbitrary dimension (at most $n-1$) such that $D_{j,q}=\|y_j-y_q\|^2_2$.
\end{claim}
\begin{proof}
	For every point $x_j$, define a function $g_j:[1,2]^n\times\{0,1\}^n\rightarrow\ell_2$ as follows: given $\{r_j\}_{j\in [n]}\in [1,2]^n$ and $\{\alpha_j\}_{j\in [n]}\in \{0,1\}^n$, $g_j\left(\{r_j\}_{j\in [n]},\{\alpha_j\}_{j\in [n]}\right)$ equals to $f(x_j)$ as defined above. It holds that
	\begin{align*}
		D_{j,q} & =\E\left[\|f(x_{j})-f(x_{q})\|_{2}^{2}\right]\\
		& =\int_{\{r_{k}\}_{k},\{\alpha_{k}\}_{k}}\left\Vert g_{j}(\{r_{k}\}_{k},\{\alpha_{k}\}_{k})-g_{q}(\{r_{k}\}_{k},\{\alpha_{k}\}_{k})\right\Vert _{2}^{2}dr_{1}\dots dr_{n}d\alpha_{1}\dots d\alpha_{n}\\
		& =\|g_{j}-g_{q}\|_{2}^{2}~,
	\end{align*}
	where the measure used in the integration is according to the truncated exponential distribution used to sample the radii.
	Thus we obtain a Hilbert space over functions, and $\{D_{j,q}\}_{j,q\in[n]}$ represents the squared distances between functions in this space. As any $n$ points in every Hilbert space can be isometrically embedded into Euclidean $\R^{n-1}$ space, the claim follows.
\end{proof}

The values $D_{j,q}$ can be computed (e.g., using conditional expectation). From the Euclidean distance matrix $\{D_{j,q}\}_{j,q\in[n]}$, one can compute points $y_1,\dots,y_n\in\ell_2$ in \Cref{clm:EuclideanMatrix} using orthgonalization (e.g., compute an orthonormal basis using the Gram-Schmidt process, and express $x_1,\dots,x_n$ in that basis).
Furthermore, the embedding is extendable in an online fashion (as the Gram-Schmidt process successively computes orthonormal bases for the subspaces spanned by prefixes $y_1,\ldots, y_i$, for $i=1,\ldots , n$).
Our embedding is naturally defined: The point $x_j$ is mapped to $y_j\in \ell_2$.
After embedding $x_1,\ldots ,x_{n-1}$ to $y_1,\dots,y_{n-1}\in \ell_2$, given a new metric point $x_n$, the squared distances $\{D_{j,n}\}_{j\in[n]}$ are now defined, and a point $y_{n}\in\ell_2$ can be computed using orthogonalization. The distortion guarantee holds by \Cref{lem:EuclideanEmbeddingSingleSample}.

\section{Lower Bound for Distortion in Online Euclidean Embedding}
\label{sec:LB:distortion}

In a classic lower bound construction, Newman and Rabinovich~\cite{NR02} showed that there exists an $n$-vertex planar graph $G$ (for arbitrarily large $n$), such that every embedding of $G$ into $\ell_2$ will suffer from distortion $\Omega(\sqrt{\log n})$. The example they used is the diamond graph. However, the diamond graph has a large doubling dimension.
Later, Gupta, Krauthgamer, and Lee \cite{GKL03} used a similar Laakso graph \cite{Laakso02}, to obtain a metric family with uniformly constant doubling dimension, and showed that it requires distortion  $\Omega(\sqrt{\log n})$ in order to be embedded into $\ell_2$.
Recently, Newman and Rabinovich \cite{NR20}, used similar arguments to \cite{NR02} (and the same diamond graphs) to show that there is a family of metric spaces such that every online embedding into $\ell_2$ requires distortion $\Omega(\sqrt{\log\Phi})$.
In this section, we similarly adapt the construction from  \cite{GKL03} to exhibit a family of metric spaces, all with uniformly constant doubling dimension, such that every deterministic online embedding into $\ell_2$ requires distortion $\Omega(\sqrt{\log \Phi})$.

The lower bound from \cite{NR20} is against a deterministic online embedding algorithm. Our lower bound here is stronger, as it also holds against a non-expansive
random online embedding (that succeeds on all vertex pairs with positive constant probability).

\OnlinegHSTLB*
Note that our \Cref{thm:DoublingOnlineEmbeddingtoEuclidean} on this family $\cal{M}_k$ will guarantee expansion $O(\sqrt{\log \Phi})$ and contraction $\Omega(1)$.
As this is an embedding into $\ell_2$, given the family $\cM$ in advance, we can scale and obtain a non-contractive embedding with expansion $O(\sqrt{\log \Phi})$.
Thus \Cref{thm:LBEuclideanDoubling} implies that our \Cref{thm:DoublingOnlineEmbeddingtoEuclidean} is tight in its dependence on the aspect ratio $\Phi$, as well as its dependence on the cardinally $n$. Further, an advance knowledge of $n$ or $\Phi$, or even the fact that the metric is a shortest path metric of a series-parallel graph, would not help.
\begin{proof}[Proof of \Cref{thm:LBEuclideanDoubling}]
  The Laakso graph is a family of graphs $\{G_n\}_{n\ge0}$. The graph $G_0$ consists of a single edge of weight $1$; $G_1$ consists of a $4$-cycle with two additional leafs, all of edge weight $4^{-1}$, illustrated in \Cref{fig:Laakso}. In general, the level $k$ Laakso graph, $G_k$, is obtained from $G_{k-1}$ by replacing each edge in $G_{k-1}$ by a copy of $G_1$, where all edges in $G_{k}$ have weight $4^{-k}$. One can embed isometrically (w.r.t.\ shortest path distance) the vertices of $G_{k-1}$ into $G_k$. The Laakso graphs are planar (in fact, series-parallel) with constant doubling dimension \cite{LP01} (see also \cite{Laakso00,Laakso02}). Nonetheless, it is known \cite{GKL03} that every embedding of the $k$-Laakso graph $G_k$ into $\ell_2$ has distortion at least $\sqrt{k}$. As $G_k$ contains $n<6^k$ vertices, it is an example of a doubling metric for which every embedding into $\ell_2$ has distortion $\Omega(\sqrt{\log n})$ (and this bound is tight \cite{GKL03,KLMN04}).
	
	\begin{figure}[t]
		\centering
		\includegraphics[width=.84\textwidth]{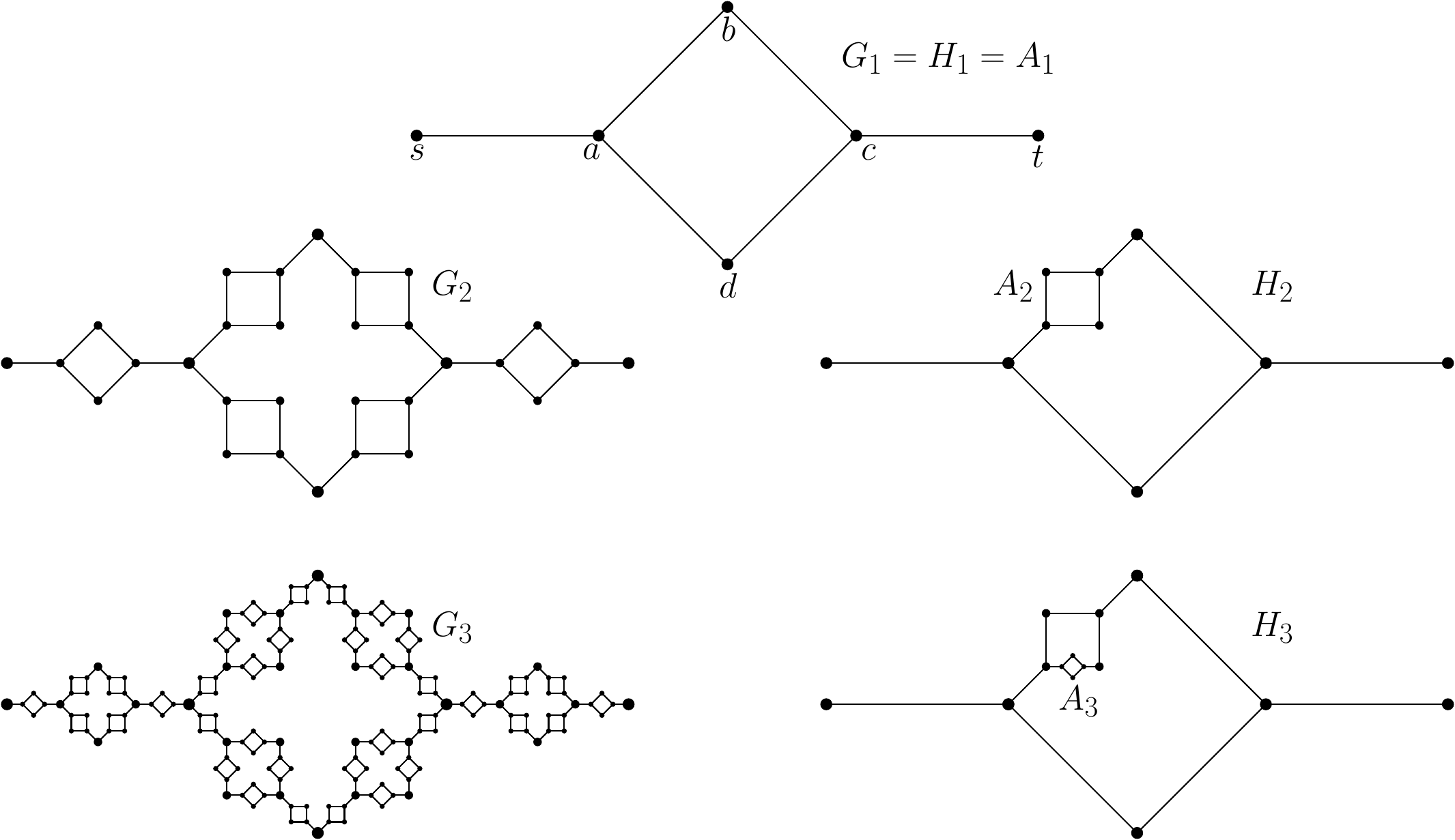}
		\caption{\footnotesize{The Laakso Graph. On the left represented the Laakso graphs $G_1, G_2, G_3$. On the right is our version, where only a single edge is replaced.
		}}
		\label{fig:Laakso}
	\end{figure}
	
	We construct a family of graphs, denoted $\{H_k\}_{k\ge0}$, constructed in a similar manner to the Laakso graphs, however in the recursive step, instead of replacing every edge with a copy of $G_1$, we replace only a single random edge with a copy of $G_1$. The choice of which edge to replace will be done adaptively w.r.t. the embedding produced by the online algorithm.
	It will hold that $H_k$ can be embedded isometrically into $G_k$.
	More formally, fix a deterministic online embedding algorithm, and let $f$ be the embedding function. We assume that $f$ is dominating, and will give a lower bound for its expansion.
	We first feed the online algorithm with $H_0=G_0$ and $H_1=G_1$, which are the same as the Laakso graphs. In the graph $G_k$, we will have a single copy of $G_1$ with weights $4^{-k}$, we will denote this copy by $A_k$, and its vertices by $s_k,a_k,b_k,c_k,d_k,t_k$ according the labeling in \Cref{fig:Laakso}.
	To obtain $G_{k+1}$, we will choose an edge $\{x_{k},y_{k}\}\in\left\{ \{s_{k},a_{k}\},\{a_{k},b_{k}\},\{b_{k},c_{k}\},\{c_{k},d_{k}\},\{d_{k},a_{k}\},\{c_{k},t_{k}\}\right\}$ maximizing $\|f(x_k)-f(y_k)\|_2$.
	We delete the edge $\{x_k,y_k\}$ (which has weight $4^{-k}$), and add a copy of $G_1$ with weights $4^{-(k+1)}$ (this copy is denoted $A_{k+1}$), by identifying $x_k$ with $s_{k+1}$, and $y_k$ with $t_{k+1}$.
	Note that the distances between vertices formerly in $H_k$ remain unchanged.
	Next we will feed the online algorithm with the points in $H_{k+1}\setminus H_{k}$. Then we will choose $\{x_{k+1},y_{k+1}\}\in A_k$ (maximizing the expansion).
  The graph family $\{H_k\}_{k\ge0}$
	is now well defined. See \Cref{fig:Laakso} for illustration.
	Note also that every graph $H_k$ is series-parallel (in particular planar), and could be embedded isometrically into $G_k$. It follows that the shortest path metric in every $H_k$ has a constant doubling dimension (as it is a sub-metric of a metric space with a constant doubling dimension). The number of vertices in $H_k$ is $4k+2$, while the aspect ratio is $\Phi=4^k$.

	For simplicity of notation, for a vertex $v\in H_k$, we denote by $\vec{v}=f(v)$ the embedding of $v$ into $\ell_2$.
	Denote by $\rho_k$ the expansion of the pair $x_k,y_k$, that is, $\rho_{k}=\frac{\|\vec{x}_{k}-\vec{y}_{k}\|}{d_{H_{k}}(x_{k},y_{k})}$.
	We argue by induction that for every $k\geq 0$, we have
	\begin{equation}\label{eq:InductionHypothesis}
    \rho_{k}\ge\sqrt{\frac{k+4}{4}}=\Omega(\sqrt{\log\Phi})=\Omega(\sqrt{|H_{k}|})~.
	\end{equation}
	Note that once the induction is complete, the deterministic component of the theorem will follow.
	
	The base case $k=0$ of the induction is immediate as for every non-contractive embedding $f$, for the first two points, $x_0$ and $y_0$ it holds that $\rho_{0}=\frac{\left\Vert \vec{x_{0}}-\vec{y_{0}}\right\Vert _{2}}{d_{H_{0}}(x_{0},y_{0})}\ge1=\sqrt{\frac{0+4}{4}}$.
	Suppose that \Cref{eq:InductionHypothesis} holds for the graph $H_{k-1}$, that is, $\rho_{k-1}=\frac{\|\vec{x}_{k-1}-\vec{y}_{k-1}\|_{2}}{d_{H_{k}}(x_{k-1},y_{k-1})}\ge\sqrt{\frac{k+3}{4}}$). We prove \Cref{eq:InductionHypothesis} for $H_k$.
	Recall that in $H_k$ we replace $\{x_{k-1},y_{k-1}\}$  with the copy $A_k$ with the vertices $s_k,a_k,b_k,c_k,d_k,t_k$. We will use the following
	inequality, which states that for every $\vec{s},\vec{t},\vec{a},\vec{b},\vec{c},\vec{d}\in\ell_{2}$,
	\begin{equation}
		\|\vec{s}-\vec{t}\|_{2}^{2}+\|\vec{b}-\vec{d}\|_{2}^{2}\le4\cdot\left(\|\vec{a}-\vec{s}\|_{2}^{2}+\|\vec{c}-\vec{t}\|_{2}^{2}\right)+2\cdot\left(\|\vec{a}-\vec{b}\|_{2}^{2}+\|\vec{b}-\vec{c}\|_{2}^{2}+\|\vec{c}-\vec{d}\|_{2}^{2}+\|\vec{d}-\vec{a}\|_{2}^{2}\right)~.\label{eq:parallelogramExtended}
	\end{equation}
	Inequality (\ref{eq:parallelogramExtended}) was previously used in \cite{GKL03} for a similar purpose. However, as we were unable to find a reference for the proof, we add a formal proof in \Cref{sec:ParallelogramProof}.
	Using inequality (\ref{eq:parallelogramExtended}),
	the fact that every edge in $A_k$ has weight $4^{-k}$,
	and the assumption that every edge in $A_k$ has expansion at most $\rho_k$, it follows that
	\begin{align}
		& \|\vec{s}_{k}-\vec{t}_{k}\|_{2}^{2}+\|\vec{b}_{k}-\vec{d}_{k}\|_{2}^{2}\nonumber \\
		& \quad\le4\cdot\left(\|\vec{a}_{k}-\vec{s}_{k}\|_{2}^{2}+\|\vec{c}_{k}-\vec{t}_{k}\|_{2}^{2}\right)+2\cdot\left(\|\vec{a}_{k}-\vec{b}_{k}\|_{2}^{2}+\|\vec{b}_{k}-\vec{c}_{k}\|_{2}^{2}+\|\vec{c}_{k}-\vec{d}_{k}\|_{2}^{2}+\|\vec{d}_{k}-\vec{a}_{k}\|_{2}^{2}\right)~.\nonumber \\
		& \quad\le16\cdot\left(\rho_{k}\cdot4^{-k}\right)^{2}\label{eq:poincare}
	\end{align}
	Using the induction hypothesis, and the assumption that the embedding $f$ is non-contractive,
	\begin{align}
		\|\vec{s}_{k}-\vec{t}_{k}\|_{2}^{2}+\|\vec{b}_{k}-\vec{d}_{k}\|_{2}^{2} & \ge\rho_{k-1}^{2}\cdot d_{H_{k}}(s_{k},t_{k})^{2}+d_{H_{k}}(b_{k},d_{k})^{2}\nonumber \\
		& \ge\frac{k+3}{4}\cdot4^{-2(k-1)}+(2\cdot4^{-k})^{2}=(k+4)\cdot4^{-2k+1}~.\label{eq:DistortionLB}
	\end{align}
	Combining inequalities \Cref{eq:poincare,eq:DistortionLB}, %
   it follows that
	\[
	\rho_{k}^{2}\ge\frac{4^{2k}}{16}\cdot(k+4)\cdot4^{-2k+1}=\frac{k+4}{4}~,
	\]
	and in particular, $\rho_k\ge\sqrt{\frac{k+4}{4}}$, as required.
    This proves part~(a) of \Cref{thm:LBEuclideanDoubling}

	Next, prove part~(b) of \Cref{thm:LBEuclideanDoubling}, extending the lower bound to online non-expanding stochastic embeddings. Let $\cM_n$ be the family of metrics that could possibly be produced by the process described above. This is an $(n-1)$-depth recursive process, with $6$ different options in each step (according to the choice of edge replaced by a copy of $G_1$). Thus $\cM_n$ contains $6^{n-1}$ metrics, all containing $O(n)$ points, with aspect ratio $4^{n}$, and uniformly constant doubling dimension (which constitute the shortest path metric of a series-parallel graph).
	Suppose that there is an online non-expanding stochastic embeddings with distortion (i.e.,  contraction) at most $\rho$.
	In other words, for every metric $(X,d_X)\in\cM_n$, let $x_{1},\dots,x_{m}$ be the metric points with respect to the order above in which they are reveled to the online algorithm. The online algorithm produces a random embedding $f$ that is
	\begin{enumerate}
		\item non-expanding: for every $x_{j},x_{q}$ and every
		embedding $f$ in the support of the distribution, $\|f(x_{j})-f(x_{q})\|_{2}\le d_{X}(x_{i},x_{j})$;
        and
		\item has bounded expected contraction: for every $x_{j},x_{q}$, $\mathbb{E}\left[\|f(x_{j})-f(x_{q})\|_{2}\right]\ge\frac{1}{\rho}\cdot d_{X}(x_{j},x_{q})$.
	\end{enumerate}
	For every $j,q\in[m]$, set $D_{j,q}=\mathbb{E}\left[\|f(x_{j})-f(x_{q})\|_{2}^{2}\right]$.
	On the one hand, as all embeddings in the support are non-contractive, it holds
	that
	\[
	D_{j,q}=\mathbb{E}\left[\|f(x_{j})-f(x_{q})\|_{2}^{2}\right]\le d_{X}^{2}(x_{j},x_{q})\,.
	\]
	On the other hand, from Jensen's inequality
	it follows that
	\[
	\sqrt{D_{j,q}}=\left(\mathbb{E}\left[\|f(x_{j})-f(x_{q})\|_{2}^{2}\right]\right)^{\frac{1}{2}}\ge\mathbb{E}\left[\|f(x_{j})-f(x_{q})\|_{2}\right]\ge\frac{1}{\rho}\cdot d_{X}(x_{i},x_{j})\,,
	\]
	or $D_{j,q}\ge\frac{1}{\rho^{2}}\cdot d_{X}^{2}(x_{i},x_{j})$. By
	\Cref{clm:EuclideanMatrix}, $\{D_{j,q}\}_{j,q\in[m]}$ is an
	Euclidean distance matrix. Note that the numbers $\{D_{j,q}\}_{j,q\in[m]}$
	are deterministically defined. Following the discussion after
	\Cref{clm:EuclideanMatrix}, there is a deterministic online algorithm
	producing vectors $y_{1},\dots,y_{n}\in\ell_{2}$ such that for every
	$j,q\in[m]$, $\|y_{j}-y_{q}\|_{2}^{2}=D_{j,q}$. In particular, $\frac{1}{\rho}\cdot d_{X}(x_{i},x_{j})\le\|y_{j}-y_{q}\|_{2}\le d_{X}(x_{i},x_{j})$.
	Thus we obtain a online deterministic embedding ($g(x_{j})=y_{j}$) with
	distortion at most $\rho$ for every metric in $\cM_n$.
	As the online embedding is deterministic and $\rho$ is fixed, a suitable scaling yields
    a deterministic dominating online embedding with distortion at most $\rho$.
	Now part~(a) of \Cref{thm:LBEuclideanDoubling} implies that $\rho\ge\sqrt{\frac{n+4}{4}}$,
    completing the proof of part~(b).
\end{proof}

\section{Light Perfect Matchings for Points in a Line}
\label{sec:line}

In this section, we consider a fully dynamic point set $S$ on the real line (i.e., both insertions and deletions are allowed). Our point set $S$ is not always even, and so we maintain a \emph{near-perfect matching}, which covers all but at most one point in $S$. We establish \Cref{thm:1D}, which we restate for convenience.

\begin{theorem}\label{thm:1D}
There is a data structure that maintains, for a dynamic finite point set $S\subset \mathbb{R}$, a near-perfect matching of weight $O(\log |S|)\cdot \diam(S)$ such that each point insertion or deletion incurs $O(\log^2 |S|)$ edge deletions and insertions in the matching.
\end{theorem}

We reduce the problem to a purely combinatorial setting. We need some definitions. Let $E$ be a set of edges on a finite set $S\subset \mathbb{R}$. We say that $E$ is \emph{laminar} if there are no two edges $a_1b_1$ and $a_2b_2$ such that $a_1<a_2<b_1<b_2$ (i.e., no two \emph{interleaving} or \emph{crossing} edges). Containment defines a partial order: We say that $a_1b_1\preceq a_2b_2$ (resp., $a_1b_1\prec a_2b_2$) if the interval $a_2b_2$ contains (resp., properly contains) the interval $a_1b_1$. The Hasse diagram of a laminar set $E$ of edges is a forest of rooted trees $F(E)$ on $E$, where a directed edge $(a_1b_1,a_2b_2)$ in $F(E)$ means that $a_2b_2$ is the shortest interval that contains $a_1b_1$. The \emph{depth} of $M$ is the depth of the forest $F(E)$. (Equivalently, the depth of $E$ is the maximum number of pairwise overlapping edges in $E$.)
With this terminology, we establish the following theorem.

\begin{theorem}\label{thm:laminar}
There is a data structure that maintains, for a dynamic point set $S$ on the Euclidean line $\mathbb{R}$, a laminar near-perfect matching of depth $O(\log |S|)$ such that it modifies (adds or deletes) at most $O(\log^2 |S|)$ edges in each step.
\end{theorem}

Clearly, every laminar matching of depth $d$ on a point set $S$ has weight at most $d\cdot \diam(S)$, and so \Cref{thm:laminar} immediately implies \Cref{thm:1D}. Importantly, the laminar property and the depth of the matching depend only on the order of the points in $S$, and the real coordinates do not matter. We would like to represent a set of $k$ points on the line by integers $[k]:=\{1,2,\ldots, k\}$. However, this representation does not easily support the insertion of new points.

To support insertions, we maintain a collection of finite sets $\mathcal{C}=\{A_1,\ldots, A_t\}$, where each set consists of consecutive integers; and a laminar near-perfect matching of depth $O(\log |A_j|)$ on each set $A_j$ (that is, we allow an unmatched point in each \emph{odd} set in the collection). Our data structure supports four operations: create, delete, merge, and split. Specifically, we define the operations
\begin{itemize}
    \item \texttt{create}: insert a 1-element set into $\mathcal{C}$;
    \item \texttt{delete}: remove a 1-element set from $\mathcal{C}$;
    \item \texttt{merge}$(A_j,A_{j'})$: concatenate the sets $A_j=\{1,\ldots , |A_j|\}$ and $A_{j'}=\{1,\ldots , |A_{j'}|\}$ 
    into the set $\{1,\ldots , |A_j|+|A_{j'}|\}$;
    \item \texttt{split}$(A_j,k)$: split the set $A_j=\{1,\ldots , |A_j|\}$ into
    two sets, $\{1,\ldots , k\}$ and $\{k+1,\ldots , |A_j|\}$;
\end{itemize}

It is easy to see how these operations support point insertions and deletions:
Suppose that we want to insert a new point into $A_j=\{1,\ldots , |A_j|\}$ between elements $k$ and $k+1$. Then the \texttt{split}$(A_j,i)$ operation splits $A_j$ between $k$ and $k+1$ into $[k]$ and $[|A_j|-k]$; the \texttt{create} operation creates a singleton $[1]$, and finally two \texttt{merge} operations concatenate $[k]$, $[1]$, and $[|A_j|-k]$ into a single set $[|A_j|+1]$.
It remains to maintain a near-perfect matching on each set in $\mathcal{C}$.

\begin{lemma}\label{lem:laminar}
There is a data structure that maintains a laminar near-perfect matching $M(A_j)$ of depth $O(\log |A_j|)$ for each set $A_j$ in the collection $\mathcal{C}=\{A_1,\ldots , A_t\}$ of intervals.
For each \texttt{merge}$(A_j,A_{j'})$ operation, it modifies $O(\log |A_j|+|A_{j'}|)$ edges; and for each \texttt{split}$(A_j,k)$ operation, it modifies $O(\log^2 |A_j|)$ edges.
\end{lemma}

\paragraph{Virtual edges.}
The remainder of this section is dedicated to the proof of \Cref{lem:laminar}. It is not difficult to maintain a laminar near-perfect matching. However, controlling the depth is challenging. Virtual edges are the key technical tool for maintaining logarithmic depth. For each edge $ab \in M(A_j)$, with $a<b$, we maintain a \emph{virtual edge} $\vir(ab)$ with $ab\preceq \vir(ab)$.
We define the \emph{length} of a virtual edge $\vir(ab)=cd$ as $\ell(ab):=|c-d|$.
Then the length of a virtual edge is an integer in $\{1,\ldots , |A_j|-1\}$.
Invariants (I1)--(I4) below will ensure that for a nested sequence of edges $a_1b_1\prec \ldots \prec a_kb_k$ yields a nested sequence virtual edges $\vir(a_1b_1)\preceq \ldots \preceq \vir(a_kb_k)$ with the additional property that they have exponentially increasing lengths. For new edges we usually set $\vir(ab):=ab$; and $ab\neq \vir(ab)$ indicates intuitively that the algorithm created edge $ab$ as a replacement for some previous edge $\vir(ab)$, and we keep a record of the deleted edge $\xi(ab)$ for accounting purposes.

Note that every edge and virtual edge is an interval in $\R$; we say that an edge $ab$, where $a<b$, \emph{contains} a vertex $c$ if $a\leq c\leq b$.

\paragraph{Data Structure.}
Consider a collection $\mathcal{C}=\{A_1,\ldots , A_t\}$ of sets of consecutive integers.
For each set $A_j$, our data structure maintains a laminar matching $M(A_j)$
(which is not necessarily perfect or near-perfect); and a virtual edge $\vir(ab)$
for each edge $ab\in M(A_j)$. The set of all virtual edges is $\Vir(A_j)=\{\vir(ab): ab\in M(A_j)\}$.
Furthermore, the matching $M(A_j)$ and its virtual edges satisfy the following invariants.
\begin{enumerate}\itemsep 2pt
    \item[(I1)] $M(A_j)\cup \Vir(A_j)$ is laminar.
    \item[(I2)] For every $ab\in M(A_j)$, there is no edge $e\in M(A_j)\cup \Vir(A_j)$ such that $ab\prec e\prec \vir(ab)$.
    \item[(I3)] None of the unmatched points in $A_j$ is contained in any edge in $M(A_j)\cup \Vir(A_j)$.
    \item[(I4)]\label{i4} If $a_1b_1, a_2b_2\in M(A_j)$ and $a_1b_1\prec a_2b_2$,
        then $\ell(a_1b_1)\leq \frac12 \, \ell(a_2b_2)$.
\end{enumerate}

\begin{lemma}\label{lem:inv}
If a matching $M(A_j)$ with virtual edges satisfies Invariants (I1)---(I4),
then the depth of $M(A_j)$ is $O(\log |A_j|)$.
\end{lemma}
\begin{proof}
By invariant (I1), $M(A_j)$ is laminar. If its depth is $k$, then $M(A_j)$ contains a nested sequence of edges $a_1b_1\prec \ldots \prec a_kb_k$.
By (I4), the lengths of the corresponding virtual edges increase exponentially (by factors of 2). The maximum (resp., minimum) length of a virtual edge is $|A_j|-1$ (resp., 1). Consequently, $k\leq \lceil \log |A_j|\rceil$, as claimed.
\end{proof}

For the implementation of our \texttt{merge} and \texttt{split} operations,
maintaining invariant (I4) requires some attention. Indeed, if we add a new edge $a_2b_2$ between two unmatched points, such that there are no other unmatched points between $a_2$ and $b_2$, and set $\vir(a_2b_2):=a_2b_2$, then invariants (I1)--(I3) are automatically maintained, but (I4) might be violated. We say that an edge $a_2b_2\in M(A_j)$ \emph{violates} invariant (I4) if there exists another edge $a_1b_1\in M(A_j)$ such that $a_1b_1\prec a_2b_2$ but $\ell(a_1b_1)> \frac12 \, \ell(a_2b_2)$.
The following algorithm greedily replaces edges that violate invariant (I4).

\paragraph{Algorithm \texttt{repair}$(M)$.} Input: $A_j=\{1,\ldots, |A_j|\}$, matching $M=M(A_j)$ with virtual edges satisfying (I1)--(I3). While $M$ does not satisfy (I4), do: Let $a_2b_2\in M$ be a maximal edge that violates (I4). Furthermore, let $a_1b_1\in M$ be an edge with the longest virtual edge such that $a_1b_1\prec a_2b_2$ but $\ell(a_1b_1)> \frac12 \, \ell(a_2b_2)$. We may assume w.l.o.g.\ that $a_1<b_1$ and $a_2<b_2$.
Modify $M$ as follows: delete both $a_1b_1$ and $a_2b_2$ from $M$, and add new edges $a_2a_1$ and $b_1b_2$. If $|a_1-a_2|\leq |b_1-b_2|$, then let $\vir(a_2a_1):=\vir(a_2b_2)$ and $\vir(b_1b_2):=b_1b_2$, see \Cref{fig:repair}; else let $\vir(a_2a_1):=a_2a_1$ and $\vir(b_1b_2):=\vir(a_2b_2)$.

	\begin{figure}[htbp]
		\centering
		\includegraphics[width=0.75\textwidth]{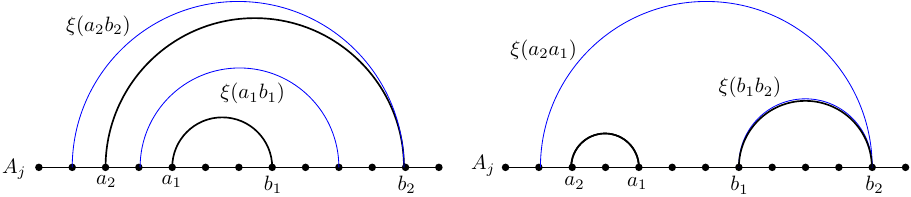}
		\caption{One iteration of \texttt{repair}$(M)$ replaces edges $a_1b_1\prec a_2b_2$ with $a_2a_1$ and $b_1b_2$.}
		\label{fig:repair}
	\end{figure}

\begin{lemma}\label{lem:repair}
Algorithm \texttt{repair}$(M)$ maintains invariants (I1)--(I3).
If $M$ contains $k$ edges that violate (I4) and they are pairwise non-overlapping, then the algorithm terminates after $O(k\log |A_j|)$ iterations; in particular it modifies  $O(k\log |A_j|)$ edges in $M$.
\end{lemma}
\begin{proof}
First we show that each iteration of Algorithm \texttt{repair}$(M)$ maintains invariants (I1)--(I3).
Consider one iteration of the algorithm where edges $a_1b_1$ and $a_2b_2$, with  $a_1b_1\prec a_2b_2$, are replaced by $a_2a_1$ and $b_1b_2$. Assume that $M$ satisfies (I1)--(I3) before the replacement. Then $M\cup \Vir$ is laminar at the beginning of the iteration. The two new edges, $a_2a_1$ and $b_2b_1$, do not cross. Suppose, for the sake of contradiction, that one of them crosses an existing edge $e\in M\cup \Vir$, $e\notin \{a_1b_1,a_2b_2,\vir(a_1b_1),\vir(a_2b_2)\}$. Since $M\cup \Vir$ is laminar, then $e$ crosses neither $a_1b_1$ nor $a_2b_2$.
It follows that $a_1b_1\prec e\prec a_2b_2$. By (I2), we may further assume that $\vir(a_1b_1)\prec e\prec a_2b_2$.

We distinguish between two cases: If $e\in M$, then it has a virtual edge $\vir(e)$ with $e\preceq \vir(e)$. Transitivity yields $\vir(a_1b_1)\prec \vir(e)$, and so the virtual edge $\vir(e)$ is longer than $\vir(a_1b_1)$, contradicting the choice of $a_1b_1$.
Otherwise, $e\in \Vir$, and $e$ is the virtual edge of some edge $e'\in M$ with $e'\prec \vir(e')=e$. In this case, the virtual edge $\vir(e')=e$ is longer than $\vir(a_1b_1)$, contradicting the choice of $a_1b_1$ again. We conclude that $M\cup \Vir$ is laminar at the end of the iteration, and invariant (I1) is maintained.

For invariant (I2), assume w.l.o.g.\ that $|a_1-a_2|\leq |b_1-b_2|$, and so the algorithm assigned virtual edges $\vir(a_2a_1)=\vir(a_2b_2)$ and $\vir(b_1b_2)=b_1b_2$. Then $b_1b_2$ clearly satisfies (I2). Since $M\cup \Vir$ remains laminar after one iteration, there is no edge $e\in M\cup \Vir$ such that $a_1b_1\prec e\prec a_2b_2$. By invariant (I2) for edge $a_2b_2$, there is no edge $e\in M\cup \Vir$ with $a_2b_2\prec e\prec \vir(a_2b_2)$. This implies that at the end of the iteration, there is no edge $e\in M\cup \Vir$ such that $a_2a_1\prec e\prec \vir(a_2b_2)=\vir(a_2 a_1)$, as required.

For invariant (I3), notice that both new edges, $a_2a_1$ and $b_1b_2$, are contained the $a_2b_2$, which in turn does not contain any unmatched point by (I3). The set of unmatched points remains the same, hence the new edges do not contain unmatched points, either.
We have shown that each iteration maintains invariants (I1)--(I3).

We now prove the second statement in \Cref{lem:repair}, about the number of iterations. Assume that $M$ contains $k$ edges that violate (I4) and they are pairwise non-overlapping. Consider again one iteration where edges $a_1b_1$ and $a_2b_2$ are replaced by $a_2a_1$ and $b_1b_2$; and w.l.o.g.\ $|a_1-a_2|\leq |b_1-b_2|$. We claim that $b_1b_2$ is the only edge that may become a violator in this iteration. Indeed,  $|a_1-a_2|\leq |b_1-b_2|$ implies that $|a_1-a_2|\leq \frac12\, |a_2-b_2|\leq \ell(a_2b_2)=\ell(a_2 a_1)$. For any edge $e\in M$ with $e\prec a_1a_2$, we have $e\preceq \vir(e)\preceq a_2a_1$ by invariant (I2), which implies $\ell(e)\leq |a_1-a_2|$. Overall, $\ell(e)\leq \frac12\, \ell(a_2a_1)$ for all edges $e\in M$ where $e\prec a_2a_1$, that is, $a_2a_1$ is not a violator.

It follows that in the course of Algorithm \texttt{repair}$(M)$, the number of violators can only decrease and the violators remain pairwise non-overlapping. Furthermore, in each iteration, either the number of violators decreases by one, or one violator $e$ is replaced by another violator $e'$ such that $\ell(e')\leq \frac12 \ell(e)$.
Since $\ell(e)$ is an integer in $\{1,\dots ,|A_j|\}$, a violator may recursively be replaced $O(\log |A_j|)$ times. Summation over $k$ initial violators yields an overall bound of $O(k\log |A_j|)$ on the number of iterations.
\end{proof}

Now, we can describe how Algorithm \texttt{repair}$(M)$ supports split and merge operations.

\paragraph{Merge operation.}
Assume that $M(A_j)$ and $M(A_{j'})$ are near-perfect matchings with virtual edges satisfying invariants (I1)--(I4). Let $A=\texttt{merge}(A_j,A_{j'})$.
If $A_j$ or $A_{j'}$ is even, then let $M(A)=M(A_j)\cup M(A_{j'})$ with the same virtual edges. Then $M(A)$ is a near-perfect matching; and it is easily verified that it satisfies invariants (I1)--(I4).

Assume now that both $A_j$ and $A_{j'}$ are odd, which means that they each contain one unmatched point, say $a\in A_j$ and $b\in A_{j'}$ with $a<b$ in $A=\{1,\ldots, |A_j|+|A_{j'}|\}$. Let $M(A)=M(A_j)\cup M(A_{j'})\cup \{ab\}$ with the same virtual edges for all edges in $M(A_j)\cup M(A_{j'})$ and with $\vir(ab):=ab$.
Note that $M(A)$ satisfies invariants (I1)--(I3). Indeed, by invariant (I3), $a$ and $b$ are not contained in any edge in $M(A_j)\cup \Vir(A_j)$ and $M(A_{j'})\cup \Vir(A_{j'})$, and so the set of edges remains laminar when we add $ab$. Since $\vir(ab)=ab$, the new edge satisfies (I2); and (I3) holds vacuously since $M(A)$ is a perfect matching.
Furthermore, the only possible violator to (I4) is the new edge $ab$.
By \Cref{lem:repair}, Algorithm \texttt{repair}$(M(A))$ returns
a perfect matching that satisfies (I1)--(I4), and modifies only $O(\log |A|)$ edges.

\paragraph{Split operation.}
Assume that $M(A_j)$ is a near-perfect matching with virtual edges satisfying invariants (I1)--(I4), and $1<k<|A_j|$ is an integer. For an operation \texttt{split}$(A_j,k)$, we delete all edges $ab\in M(A_j)$ such that the virtual edge $\vir(ab)$ is between $\{1,\ldots , k\}$ and $\{k+1,\ldots ,|A_j|\}$. Since the depth of $M(A_j)\cup \Vir(A_j)$ is $O(\log |A_j|)$, there are $O(\log |A_j|)$ such edges. Sort the unmatched vertices in $\{1,\ldots , k\}$ and $\{k+1,\ldots ,|A_j|\}$, respectively, in increasing order and greedily match consecutive pairs (leaving at most one vertex unmatched in each set). Denote by $M_1$ and $M_2$ the resulting matchings. For each new edge $ab\in M_1\cup M_2$, we define the virtual edge as $\vir(ab)=ab$; and let $\Vir_1$ and $\Vir_2$ denote the corresponding virtual edges.

After the first step, when we delete edges from $M(A_j)$, the unmatched vertices are not contained in any remaining edges in $M(A_j)\cup \Vir(A_j)$ due to invariants (I1) and (I3). Then the greedy algorithm on the unmatched points in $\{1,\ldots , k\}$ and $\{k+1,\ldots ,|A_j|\}$, resp., produces pairwise noncrossing edges that do not cross any surviving edges in $M(A_j)\cup \Vir(A_j)$, either. Consequently, both $M_1\cup \Vir_1$ and $M_2\cup \Vir_2$ are lateral. This establishes (I1). The new edges clearly satisfy (I2), and any remaining unmatched point satisfies (I3), as well. Finally, note that only the new edges may violate invariant (I4). There are $O(\log |A_j|)$ new edges and they are pairwise non-crossing by construction. By \Cref{lem:repair}, Algorithms \texttt{repair}$(M_1)$ and  \texttt{repair}$(M_2)$ return near-perfect matchings that satisfy (I1)--(I4), and modify only $O(\log |A_j|\cdot (\log k+ \log (|A_j|-k) )\leq O(\log^2 |A_j|)$ edges.

This completes the proof of \Cref{lem:laminar}. Since we can handle an insertion or deletion with a constant number of split and merge operations, our data structure can dynamically maintain a matching satisfying invariant (I1)--(I4) with recourse $O(\log^2 n)$. By \Cref{lem:inv}, this matching has depth $O(\log n)$, as required;
completing the proof of \Cref{thm:laminar}.

\section{Light Perfect Matchings for General Metrics}
\label{sec:light}

Recall that the online algorithm by Gu et al.~\cite{GuG016} maintains a spanning tree of weight\\ $O(\cost(\MST))$ with constant recourse per point for a sequence of points in a metric space. Combined with \Cref{thm:1D}, we obtain the proof of \Cref{thm:Tree}. We restate the theorem for convenience.

\MatchingGeneralMetrics*

\begin{proof}
Let $S_i=\{s_1,\ldots, s_{2i}\}$ be the set of the first $2i$ points in $(X,d)$. As noted above, Gu et al.~\cite{GuG016} maintains a spanning tree $T_i$ of weight $O(\cost(\MST(S_i)))$ for the set $S_i$, such that $T_i$ is obtained from $T_{i-1}$ by deleting and adding $O(1)$ edges. W.l.o.g.\ we may assume that the update involves $O(1)$ edge deletions, followed by $O(1)$ edge insertions. In intermediate stages of the update, we maintain a spanning forest of $O(1)$ trees.

A DFS traversal of a tree $T$  (starting from an arbitrary root $\varrho$) defines an \emph{Euler tour} $\mathcal{E}(T)$ that traverses every edge of $T$ precisely twice (once in each direction). By omitting repeated vertices in $\mathcal{E}(T)$, we also obtain a Hamilton path $\mathcal{P}(T)$ on the vertices of $T$, starting from the root $\varrho$. The triangle inequality implies $\cost(\mathcal{P}(T)) \leq \cost(\mathcal{E}(T)) \leq 2\, \cost(T)$.

Consider the dynamic forest $F$ produced by the algorithm by Gu et al.~\cite{GuG016}. We wish to maintain an Euler tour $\mathcal{E}(T)$ and a Hamilton path $\mathcal{P}(T)$ for each tree in $F$. We need to show that each edge deletion and each edge insertion in the forest $F$ incurs $O(1)$ edge insertions or deletions in the tours $\mathcal{E}(T)$ and paths $\mathcal{P}(T)$.

\smallskip\noindent\textbf{Edge deletion.}
Suppose an edge $uv$ is deleted from a tree $T$ rooted at $\varrho$, where $u$ is closer to the root than $v$. The deletion of $uv$ splits $T$ into two trees, say, $T_1$ and $T_2$, rooted at $\varrho_1=\varrho$ and $\varrho_2=v$, respectively. By deleting both occurrences of $uv$ from $\mathcal{E}(T)$, the tour breaks into two paths; we can recover $\mathcal{E}(T_1)$ and $\mathcal{E}(T_2)$ by identifying the two endpoints of each path. The deletion of both occurrences of $uv$ from $\mathcal{P}(T)$ breaks it into up to three paths, corresponding to the following subpaths along $\mathcal{E}(T)$: (1) from $\varrho$ to $u$, (2) from $v$ to $v$, and (3) from $v$ to $\varrho$. The union of the 1st and 3rd paths is $\mathcal{P}(T_1)$, and the 2nd path is $\mathcal{P}(T_2)$.

\smallskip\noindent\textbf{Edge insertion.}
Suppose an edge $uv$ is inserted between trees $T_1$ and $T_2$, with roots $\varrho_1$ and $\varrho_2$, respectively. We may choose the root of the resulting tree $T$ to be $\varrho=\varrho_1$. We can merge $\mathcal{E}(T_1)$ and $\mathcal{E}(T_2)$ into $\mathcal{E}(T)$ by adding edges $uv$ and $vu$. To construct the Hamilton path $\mathcal{P}(T)$, we break $\mathcal{P}(T_1)$ (resp., $\mathcal{P}(T_2)$) into two paths at $u$ (resp., $v$); and  concatenate the resulting four paths into $\mathcal{P}(T)$.

The data structure in \Cref{sec:line} maintains a near-perfect matching of depths $O(\log n)$ on the paths $\mathcal{P}(T)$; using $O(\log^2 n)$ changes in the matching for each edge deletion and insertion in the forest.
Since $T_i$ is a spanning tree on an even vertex set $S_i$, then the data structure produces a perfect matching of weight $O(\log n)\cdot\cost(\MST(T_i))$ using $O(\log^2 n)$ changes in the matching, hence it uses recourse $O(\log^2 n)$.
\end{proof}

\section{Competitive Ratio --- Oblivious Adversary}
\label{sec:comp}

\subsection{Ultrametrics}

Consider a set $A\subseteq X$ of $2n$ points $x_1,x_2,\dots,x_{2n}$ in an ultrametric represented by a 1-HST $T$. A matching $M$ is called \emph{inward} if for every node $v\in T$, the number of internal edges in $X_v$ (i.e., edges of $M$ induced by $X_v$) is exactly $\left\lfloor \frac{|X_{v}\cap A|}{2}\right\rfloor$.
In particular, this implies that at most one point of $X_v\cap A$ is matched to a point outside $X_v$ (if and only if $|X_{v}\cap A|$ is odd).
We argue that every inward solution is optimal.
To simplify our argument, we will analyze also the case where the set $A$ has an odd size. Here, a perfect matching has to be of size $\left\lfloor \frac{|A|}{2}\right\rfloor$.

\begin{lemma} \label{lem:inward}
    Consider a 1-HST $(X,d_X)$ with an associated tree $T$ and a set
    $A=\{x_1,x_2,\dots,x_{l}\}$ of $l$ points in $X$. Then every inward matching is a minimum-weight perfect matching.
\end{lemma}

\begin{proof}
    The proof is by induction on the height of the HST $T$. In the base case where the height is $1$, the metric $d_X$ is uniform, and all matchings have the same weight. Consider an HST $T$ of height $h+1$, with root vertex $v$, and children $v_1,v_2,\dots,v_k$.
    Assume w.l.o.g.\ that $|X_{v_i}\cap A|$ is odd for $i\le s$, and even for $i\ge s+1$ ($s\in[0,k]$).
    Let $M^{\opt}$ be a perfect matching of minimum weight, and $M^{\inw}$ be an arbitrary inward matching.
    By definition, $M^{\inw}$ consist of $\left\lfloor \frac{s}{2}\right\rfloor$ edges $e_1,\dots,e_{\left\lfloor \frac{s}{2}\right\rfloor }$ with one endpoint in each of the sets $X_{v_i}$ for $i\le s$ (perhaps one set missing), and inward matchings $M^{\inw}_i$ restricted to $X_{v_i}$. Note that $M^{\inw}_i$ matches all the points (other than exactly $1$ for $i\le s$) in $X_{v_i}$ internally.
    From the other hand, $M^{\opt}$ consist of $k$ matchings $M^{\opt}_i=M^{\opt}\cap X_{v_i}$ restricted to the sets $X_{v_i}$, and in addition edges $e'_1,e'_2,\dots,e'_r$ between different sets, where $r\ge \left\lfloor \frac{s}{2}\right\rfloor $. For every set $X_{v_i}$, let $A_i\subseteq X_{v_i}$ be the subset of points matched to points outside of $X_{v_i}$ by $M^{\opt}$.

    We create a new matching $M$ from $M^{\opt}$ as follows: we delete all edges $e'_1,e'_2,\dots,e'_r$, and add new edges $\tilde{e}_1,\tilde{e}_2,\dots,\tilde{e}_r$ which are a perfect matching over $\bigcup_{i=1}^k A_i$, where each set $A_i$ has at most one point matched out of $A_i$. All the other edges in $M^{\opt}$ stay intact. Clearly $\cost(M)\le\cost(M^{\opt})$, as we replaced only edges of maximal possible weight.
    Denote by $M_i =M\cap X_{v_i}$ the matching restricted to the sets $X_{v_i}$. Note that $M_i$ is a perfect matching over $A\cap X_{v_i}$.
    By the induction hypothesis, $\cost(M^{\inw}_i)\le\cost(M_i)$ (as both are perfect matchings on $X_i\cap A$). Denote by $\Delta$ the label of $v$. We conclude:
    \[
    \cost(M{}^{\inw})=s\cdot\Delta+\sum_{i=1}^{k}\cost(M{}_{i}^{\inw})\le s\cdot\Delta+\sum_{i=1}^{k}\cost(M{}_{i})=\cost(M)\le\cost(M^{\opt})~,
    \]
    the lemma now follows.
\end{proof}

Next, we show that an inward solution can be maintained with recourse proportional to the height of the HST.

\begin{lemma} \label{lem:ultrametricOptimal}
Given an inward matching on a point set $A$ in a 1-HST $(X,d_X)$ of height $h$, for any new point inserted into $A$, one can maintain an inward matching with recourse at most $2h$.
\end{lemma}
\begin{proof}
    Consider an HST represented by a labeled tree $T$ of height $h$.
    A non-inward matching $M$ on $A$ is called \emph{$i$-problematic due to a node $v$ at height $i$} if the following holds:
    \begin{itemize}
        \item  For every internal node $u$, at most one point in $X_u\cap A$ is matched to a point outside of $X_u$.
        \item For every internal node $u$ which is not an ancestor of $v$, the number of internal edges in $X_u$ is $\left\lfloor \frac{|X_{u}\cap A|}{2}\right\rfloor$.
        In particular, at most one point $X_u\cap A$ is not matched to any point in $X_u$.
        \item The number of internal edges in $X_v$ is $\left\lfloor \frac{|X_{v}\cap A|}{2}\right\rfloor-1$. Moreover, $X_v\cap A$ contains exactly two points not matched to any point in $X_v\cap A$ (implying that $|X_{v}\cap A|$ is even).
        \item For every internal node $u$ which is an ancestor of $v$, the number of internal edges in $X_u$ is at least $\left\lfloor \frac{|X_{u}\cap A|}{2}\right\rfloor-1$. Moreover, $X_u\cap A$ contains at most two points not matched by any other point in $X_u$.
    \end{itemize}
	Note that an $i$-problematic matching is almost inward matching. In particular, consider an inward matching $M$ over a set $A$, and let $A'=A\cup\{x\}$. Let $v$ be the ancestor of $x$ of minimal height $i$ such that $|X_{v}\cap A|$ is odd. Then the matching $M$ will be $i$-problematic w.r.t.\ $A'$ due to $v$. If there is no such a node $v$, then the matching $M$ will remain inward.
	
    We show, by induction on $h-i$, that an $i$-\emph{problematic} matching can be transformed into an inward matching by deleting at most $h-i$ edges.
    The base case is when we are given an $h$-problematic matching $M$. Let $r$ be the root of $T$. Then the matching $M$  contains $\left\lfloor \frac{|X_{r}\cap A|}{2}\right\rfloor-1=\left\lfloor
    \frac{|A|}{2}\right\rfloor-1$ edges, while each child $v$ of $r$ contains $\left\lfloor \frac{|X_{v}\cap A|}{2}\right\rfloor$ edges. We can simply add an edge between a pair of unmatched points and obtain an inward matching, as required.

    In general, consider an $i$-problematic matching due to a node $v$ at height $i$. Then $X_v$ contains exactly two points $x,y\in X_v$ such that $M$ does not match $x$ and $y$ to points inside $X_v$. We continue by cases analysis:
    \begin{itemize}
        \item If both $x$ and $y$ are unmatched in $M$, we simply add the edge $\{x,y\}$ to $M$. As a result, $M$ is an inward matching, and we are done.
        \item If $x$ is matched to some point $\hat{x}$ while $y$ is unmatched. Let $\hat{v}$ be the minimal height internal node such that $x,\hat{x}\in X_{\hat{v}}$, and note that $\hat{v}$ is an ancestor of $v$.
        Delete the edge $\{x,\hat{x}\}$ from $M$ and add the edge $\{x,y\}$ to $M$. Denote the new matching by $\hat{M}$.
        The matching $\hat{M}$ is either an inward matching or it is $i'$-problematic due to a node $v'$ of height $i'>i$ which is an ancestor of $\hat{v}$.
        By the induction hypothesis, by deleting $h-i'$ edges from $\hat{M}$ (and adding others), we obtain an inward matching as required.
        \item If $y$ is matched to some point $\hat{y}$ while $x$ is unmatched, the solution is symmetric to the previous case.
    \end{itemize}
    Note that the fourth case where both $x$ and $y$ are matched to points outside of $X_v$ is impossible.

    Next, we turn to the actual algorithm for ultrametrics.
    We maintain an inward matching, which is optimal by \Cref{lem:inward}. When a new pair of points $x,y$ arrive, we first insert $x$. This might cause the matching to be $i$ problematic for some $1\leq i\leq h$, which we can fix with at most $h-i\le h$ deletions. Then we add $y$, which again can cause the matching to be $i'$ problematic for some $1\leq i'\leq h$, and we fix it again. Overall, the algorithm used recourse $2h$.
\end{proof}

Note that inward matchings remain inward wehn new points are added to an ultrametric.
\begin{lemma}\label{lem:stability}
    Let $(X,d_X)$ and $(Y,d_Y)$ be ultrametrics such that $X$ is a subspace of $Y$. If a matching $M$ on a set $A\subset X$ is inward w.r.t.\ $(X,d_X)$, then it is also inward w.r.t.\ $(Y,d_Y)$.
\end{lemma}

\begin{proof}
Let $T_X$ and $T_Y$ be the 1-HSTs associated with $(X,d_X)$ and $(Y,d_Y)$, where the leaves correspond to the elements of $X$ and $Y$, respectively. Since an ultrametric is associated with a unique HST~\cite{BartalLMN04}, $T_X$ is the tree induced by the leaves of $T_Y$ corresponding to $X$. In particular, for every node $v\in T_Y$, there is a node $u\in T_X$ such that $Y_v\cap X=X_u$. Consequently, for every node $v\in T_Y$, $M$ has exactly   $\left\lfloor \frac{|Y_{v}\cap A|}{2}\right\rfloor = \left\lfloor \frac{|X_u\cap A|}{2}\right\rfloor$ edges in $Y_v$, and so $M$ is inward w.r.t.\ $(Y,d_Y)$.
\end{proof}

\subsection{Heavy-Path Decomposition to Support Updates with Recourse $O(\log^3 n)$}
\label{ssec:HP}

\emph{Heavy-path decomposition} of rooted trees was introduced by Sleator and Tarjan~\cite{SleatorT83}. Let $T$ be a rooted tree, where the weight $w(v)$ of a node $v$ equals the number of leaves in the subtree rooted at $v$. An edge of $T$ between a parent $u$ and child $v$ is \emph{heavy} if $w(v)>\frac12 w(u)$, otherwise it is \emph{light}. The heavy edges form a set of descending paths in $T$, called \emph{heavy paths}; the collection of heavy paths forms the \emph{heavy-path decomposition} of $T$. If $T$ has $n$ nodes, every descending path intersects $O(\log n)$ heavy paths; in particular, if we contract all heavy edges, the height of the resulting tree $\tilde{T}$ is $O(\log n)$.

\paragraph{Heavy-Path Inward Matching.}
We define a ``relaxed'' inward matching on an HST, and then show that it is a $2$-approximate minimum-weight perfect matching, and it can be maintained with recourse $O(\log^3 n)$. Let $T$ be the HST, associated with the ultrametric $X$ on $n$ points. Let $\mathcal{H}$ be the collection of heavy paths in $T$, and for every heavy path $P\in \mathcal{H}$, let $t(P)$ be the vertex of $P$ closest to the root (i.e., the \emph{top} node of $P$).

Consider a set $A\subseteq X$ of $l$ points $A=\{x_1,x_2,\dots,x_{l}\}$. We may assume w.l.o.g.\ that every point in $A$ corresponds to a unique leaf in $T$. A matching $M$ on $A$ is \emph{heavy-path inward} (\emph{HP-inward}, for short) if the following conditions are satisfied:
\begin{enumerate}
\item for a heavy path $P=(u_1,u_2,\ldots , u_m)$ with $u_1=t(P)$, at most one point in $X_{u_1}\cap A$ is matched to a point outside $X_{u_1}$  (if and only if $|X_{u_1}\cap A|$ is odd); and
\item if a point $p\in X_{u_1}\cap A$ is matched to a point outside $X_{u_1}$, then $p\in X_{u_1}\setminus X_{u_{i^*}}$ for the minimum integer $i^*\geq 2$ such that $(X_{u_1}\setminus X_{u_{i^*}})\cap A$ is odd (where $X_{u_{m+1}}:=\emptyset$).
\end{enumerate}
Note that every inward matching is HP-inward, but an HP-inward matching is not necessarily inward. In particular, there is no restriction on the matching along a heavy path.

\begin{figure}[ht]
	\centering
	\includegraphics[width=0.8\textwidth]{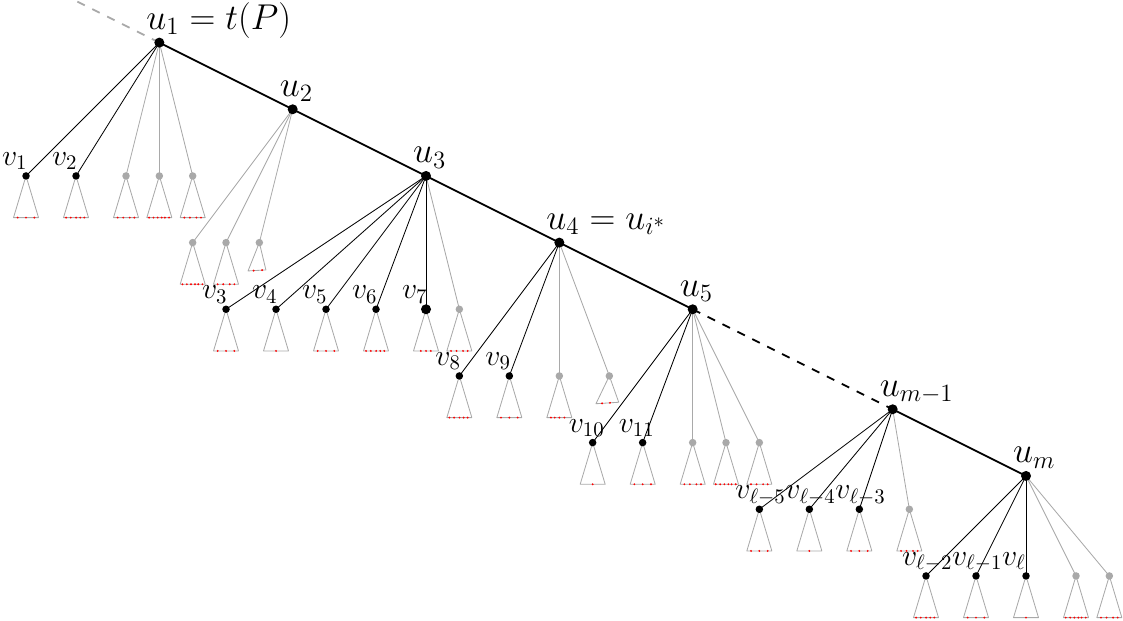}
	\caption{Illustration of the proof of \Cref{lem:heavy-inward}. The heavy path $P=(u_1,\dots ,u_m)$ is depicted using a bold line. The light edges from $P$ node $u_i$ towards children with an odd number of leaves are depicted using black lines (and the children denoted $v_j$), while edges towards children with an even number of leaves are depicted using gray lines (children are unnamed). The leaves are depicted in red. $u_{i^*}$ is the first node along $P$ such that $(X_{u_1}\setminus X_{u_{i^*}})\cap A$ is odd.
		\label{fig:InwardMatching}}
\end{figure}

\begin{lemma} \label{lem:heavy-inward}
    Consider an ultrametric $(X,d_X)$ with an associated tree $T$ and a set $A=\{x_1,x_2,\dots,x_{l}\}$ of $l$ points in $X$. For every HP-inward matching $M^{\HPI}$, we have $\cost(M^{\HPI}) \leq 2\,\cost(M^\opt)$, where $M^\opt$ is a minimum weight near-perfect matching on $A$.
\end{lemma}
\begin{proof}
Let $M^\opt$ be a minimum-weight near-perfect matching for $A$, and let $M^{\mathrm{HPI}}$ be an HP-inward matching. Recall that for a matched pair $\{x,y\}\subset A$, we have $\cost(\{x,y\})=d_X(x,y)=\Gamma_{\lca(x,y)}$, where $\lca(x,y)$ is the least common ancestor of $x$ and $y$ in $T$. We say that a matched pair $\{x,y\}$ is \emph{associated} with the node $\lca(x,y)$ of $T$.

For every node $u$ of $T$, let $M_u^{\opt} =\left\{ \{x,y\}\in M^\opt: \lca(x,y)=u \right\}$, that is, the pairs in $M^{\opt}$ associated with $u$; and similarly let $M_u^{\HPI} =\left\{ \{x,y\}\in M^{\HPI}: \lca(x,y)=u \right\}$.
For every heavy path $P\in \mathcal{H}$, let $M_P^{\opt} =\bigcup_{u\in P} M_u^{\opt}$ and $M_P^{\HPI} =\bigcup_{u\in P} M_u^{\HPI}$.  These are the pairs in the two matchings whose least common ancestors are in $P$. We claim that for every $P\in \mathcal{H}$, we have
\begin{equation}\label{eq:HPI2}
    \cost(M_P^{\HPI}) \leq 2\, \cost(M_P^{\opt})~.
\end{equation}
Summation of inequality \eqref{eq:HPI2} overall heavy paths $P\in \mathcal{H}$ will immediately imply the lemma.
To prove \eqref{eq:HPI2}, consider a heavy path $P=(u_1,\ldots , u_m)$ with $u_1=t(P)$; refer to \Cref{fig:InwardMatching}. Let $N=\{v_1, \ldots ,v_\ell\}$ be the set of nodes of $T$ such that for every $i\in \{1,\ldots ,\ell\}$, there exists a light edge $u_jv_i$ for some $j\in \{1,\ldots ,m\}$, where $v_i$ is a child of $u_j$ and $|X_{v_i}\cap A|$ is odd. Assume that the nodes in $N$ are labeled such that if $u_jv_i$ and $u_{j'}v_{i'}$ are light edges and $i<i'$, then $j<j'$. Every node $v_i\in N$ is the top node of some heavy path in $\mathcal{H}$, and both $M^{\opt}$ and $M^{\HPI}$ matches exactly one point in $X_{v_i}\cap A$ to some point outside of $X_{v_i}$.  Let $M_N^{\opt}$ be a matching on $N$ such that $\{v_i,v_j\}\in M_N^{\opt}$ if and only if $M^{\opt}$ matches a point in $X_{v_i}$ to a point in $X_{v_j}$; and define $M_N^{\HPI}$ analogously. Note that both $M_N^{\opt}$ and $M_N^{\HPI}$ are near-perfect matchings on $N$. We also define the cost function $\cost(\{v_i,v_j\}):=\Gamma(u)$, where $u=\lca(v_i,v_j)$.
Note that $u=\parent(v_{\max\{i,j\}})$. With these costs, we have
$\cost(M_P^{\HPI}) = \cost(M_N^{\HPI})$ and
$\cost(M_P^{\opt}) =\cost(M_N^{\opt})$.

First assume that $|X_{u_1}\cap A|$ is even. Then $|N|$ is even, and so both $M_N^{\opt}$ and $M_N^{\HPI}$ are perfect matchings on $N$.
We compare
$\cost(M_N^{\HPI})$ and $\cost(M_N^{\opt})$ with the following charging scheme:
We \emph{charge} each matched pair $\{v_i,v_j\}\in M_N^{\HPI}$ to the pair $\{v_i',v_j'\}\in M_N^{\opt}$ that includes $v_{\max\{i,j\}}$. Note that
$\cost(\{v_i,v_j\}) = \Gamma(\parent(\max\{i,j\})) \leq \Gamma(\parent(\max\{i',j'\}))$, that is, we charge each edge in $M_N^{\HPI}$ to an edge in $M_N^{\opt}$ of the same or largest cost. Clearly, each edge in $M_N^{\opt}$ receives charges from at most two edges of $M_N^{\HPI}$. Consequently, we obtain
\begin{align} \label{eq:HIP-OPT}
 \cost(M_N^{\HPI})
 & = \sum_{\{v_i,v_j\}\in M_N^{\HPI}} \cost(\{v_i,v_j\})\\
 &=\sum_{\{v_i,v_j\}\in M_N^{\HPI}}
 \Gamma(\parent(v_{\max\{i,j\}}))\nonumber\\
 & \leq 2\, \sum_{\{v_i,v_j\}\in M_N^{\opt}}
 \Gamma(\parent(v_{\max\{i,j\}}))\nonumber\\
 &= 2\, \sum_{\{v_i,v_j\}\in M_N^{\opt}}
 \cost(\{v_i,v_j\})
 =2\, \cost(M_N^{\opt}) \nonumber
\end{align}

Next assume that  $|X_{u_1}\cap A|$ is odd, and so $|N|$ is odd, as well. By condition~1 in the definition of HP-inward matching, both $M^{\opt}$ and $M^{\HPI}$ matches exactly one point in $X_{u_1}\cap A$ to a point outside of $X_{u_1}$. That is, both $M_N^{\opt}$ and $M_N^{\HPI}$ keep one node in $N$ unmatched. If the same node is unmatched in both $M_N^{\opt}$ and $M_N^{\HPI}$,
then the charging scheme and inequality \eqref{eq:HIP-OPT} carries over.

Let $i^*\geq 1$  be the smallest integer such that $(X_{u_1}\setminus X_{u_{i^*}})\cap A$ is odd.
By condition~2 in the definition of HP-inward matching, $M^{\opt}$ and $M^{\HPI}$ each match exactly one point in $(X_{u_1}\setminus X_{u_{i^*}})\cap A$ to a point outside of $X_{u_1}$. In particular, a child of some node $u_i$, $1\leq i< i^*$, is unmatched in $M_N^{\HPI}$. We show that for all $i< i^*-1$, node $u_i$ has an even number of children in $N$. Indeed, suppose otherwise and consider the smallest $i<i^*$ where $u_i$ has an odd number of children; then $(X_{u_1}\setminus X_{u_{i+1}})\cap A$ is odd, contradicting the choice of $i^*$). Similarly, node $u_{i^*-1}$ has an odd number of children (or else $i^*$ would not be minimal).
Note that a minimum-weight near-perfect matching $M_N^\opt$ is not necessarily unique: There may be several choices for the unmatched node. Let us construct $M_N^{\opt}$ such that a child of $u_{i^*-1}\in P$ is unmatched, and $M_N^\opt$ is perfect matching among the (evenly many) children of $u_i$ for all $i<i^*-1$.
Now inequality \eqref{eq:HIP-OPT} carries over and completes the proof of \eqref{eq:HPI2}.

Summation of \eqref{eq:HPI2} over all heavy paths in $\mathcal{H}$ yields
\[
\cost(M^{\HPI})
  = \sum_{P\in \mathcal{H}} \cost(M_P^{\HPI})
  \leq \sum_{P\in \mathcal{H}} 2\, \cost(M_P^{\opt})
  =2\, \cost(M^{\opt}).
    \qedhere
\]
\end{proof}

It remains to show that we can maintain an HP-inward matching with recourse $O(\log^3 n)$ per point in the evolving ultrametric (cf.~\Cref{rem:HST}). Let $H$ denote the set of heavy edges in a tree $T$ with $n$ leaves. As noted above, the tree obtained by contracting all heavy edges has height $O(\log n)$. We show that the arrival of a new node in $T$ incurs only $O(\log n)$ changes in $H$.

\begin{lemma} \label{lem:heavyedges}
The addition of a new point to an ultrametric on $n$ points incurs $O(\log n)$ insertions and deletions in the set $H$ of heavy edges of the associated HST.
\end{lemma}
\begin{proof}
Let $T$ be the HST associated with an ultrametric $(X,d_X)$ with $|X|=n$. Then $T$ has $n$ leaves, and at most $2n-1$ nodes. When a new point is added to the ultrqametric, the tree is augmented by one new leaf and at most one additional node, as outlined in \Cref{rem:HST}> Specifically, the new tree is created by the following operations: Attaching a new leaf to a node, adding a new root above the current root, or subdividing an edge.
We show that each of these operations, a heavy-path decomposition can be updated using $O(\log n)$ insertions and deletions of heavy edges.

If a new root is added to $T$, it is attached to the heavy path that ends at the current root. Assume now that a new node $\ell$ is added to $T$ as either a new leaf or a subdivision node. The weight of a subtreee of $T$ can only increase; and the weight of a subtree $T_u$ rooted at a node $u$ increases if and only if $u$ is an ancestor of $\ell$. If both $u$ and $v$ are ancestors of $\ell$, and $uv$ is a heavy edge (i.e., in $uv\in H$), then $uv$ remains a heavy. The path from $\ell$ to the root contains $O(\log n)$ edges that are \emph{not} in $H$; any of these edges may become heavy. If an edge $uv$ becomes heavy, where $u$ is a parent of $v$, then at most one edge incident to $u$ may become light. Overall, the set $H$ of heavy edges can be updated with $O(\log n)$ insertions and deletions.
\end{proof}

\begin{lemma} \label{lem:2HSTAppprox}
In a $k$-HST with $n$ nodes, where new nodes are added in an online fashion, one can maintain a
HP-inward matching with recourse $O(\log^3 n)$.
\end{lemma}
\begin{proof}
    We maintain a heavy path decomposition of $T$ (i.e., the collection $\mathcal{H}$ of heavy paths in $T$). By~\Cref{lem:heavyedges}, $\mathcal{H}$ can be maintained by $O(\log n)$ \emph{operations} per new node, where each operation either splits a heavy path into two, or merges two heavy paths into one.

    For each $P\in \mathcal{H}$, we maintain a near-perfect matching using the 1D data structure in~\Cref{sec:line}. Specifically, for every heavy path $P\in \mathcal{H}$, we maintain the sequence $N(P)=(v_1,\ldots ,v_\ell)$ of children of the nodes in $P$ that $|X_{v_j}\cap A|$ is odd, sorted in increasing order along $P$.

    By \Cref{lem:laminar}, we can maintain near-perfect matchings of depth $O(\log n)$ on the lists $\mathcal{N}=\{N(P) : P\in \mathcal{H}\}$, with $O(\log n)$ changes in the matching per \texttt{merge} operation, and $O(\log^2 n)$ changes per \texttt{split} operation. Specifically, when $|N(P)|$ is even, we maintain a perfect matching on $N(P)$, which satisfies both conditions of HP-inward matchings for $P$. When $|N(P)|$ is odd, then condition~2 of HP-inward matchings specifies a node $v_j\in N$ that must be matched to some node outside of $N(P)$.  We can easily modify the \texttt{merge} and \texttt{split} operations described in \Cref{sec:line} to specify the unmatched node in $N(P)$ when $|N(P)|$ is odd, using the following post-processing step: Split $N=(v_1,\ldots ,v_\ell)$ into three parts $N^{-1}=(v_1,\ldots , v_{j-1})$, $N^0=(v_j)$, and $N^+=(v_{j+1},\ldots , v_{\ell})$, and then combine the near-perfect matchings on $N^-$ and $N^+$ into a perfect matching on $N^-\cup N^+=(v_1,\ldots , v_{j-1}, v_{j+1},\ldots , v_{\ell})$ by pairing up any unmatched nodes in $N^-$ and $N^+$.
    With this modification, the data structure in \Cref{sec:line} maintains an HP-inward matching on $T$ with recourse $O(\log^2 n)$ per operation.

    A node insertion in $T$ incurs $O(\log n)$ merge and split operations on the heavy paths in $\mathcal{H}$, by \Cref{lem:heavyedges}, hence on the lists $\mathcal{N}=\{N(P): P\in \mathcal{H}\}$. Consequently, we can maintain an HP-inward matching with recourse $O(\log^3 n)$.
\end{proof}

\subsection{Oblivious Solution for General Metric}

Our main tool in this section will be our dominating stochastic metric embeddings
into ultrametrics, specifically,
\Cref{thm:DoublingOnlineEmbedding,thm:EuclidineOnlineEmbeddingIntoHST}.
We will use the following observation in our algorithm:

\begin{observation}\label{obs:DiameterHSTbound}
	Consider a metric space $(X,d_X)$, and let $U$ be an ultrametric produced using \Cref{thm:DoublingOnlineEmbedding}.
	Then the diameter of the resulting HST is equal to the diameter of $\{x_1,\dots,x_n\}$, up to an $O(1)$ factor.\\
	The same observation holds for \Cref{thm:EuclidineOnlineEmbeddingIntoHST}, as well.
\end{observation}
\begin{proof}
	The ultametric in \Cref{thm:DoublingOnlineEmbedding} was created by constructing a padded decomposition for every scale $\Delta_i=2^i$ using \Cref{thm:LDDonline}.
	By the properties of \Cref{thm:LDDonline}, for every scale $i$ such that $\Delta_i\ge4\cdot\diam(X)$, it holds that the partition at scale $i$ contains only a single cluster containing all the points (centered in $x_1$). In particular the resulting diameter of the ultrametric is bounded by $4\cdot\diam(X)$.
	
	The same property for \Cref{thm:LDDonlineEuclidean}, and thus for \Cref{thm:EuclidineOnlineEmbeddingIntoHST}, as required.
\end{proof}

Note that \Cref{obs:DiameterHSTbound} implies that the height of the resulting ultrametric is bounded by the logarithm of the aspect ratio.
We are now ready to prove~\Cref{thm:ObliviousGeneralMetric} (restated for convenience):

\ObliviousGeneralMetricUB*

\begin{proof}
As we cope with an oblivious adversary, the metric space $(X,d_{X})$,
as well as the order of arriving points $x_{1},x_{2},\dots,x_{2n-1},x_{2n},\dots$
is fixed in advance (even though unknown to the algorithm). Denote by $\Phi_{2n}$
the aspect ratio of the metric induced by $x_{1},x_{2},\dots,x_{2n-1},x_{2n}$.
Using \Cref{thm:DoublingOnlineEmbedding} we will maintain a probabilistic embedding $f_{2n}$
of $x_{1},\dots,x_{2n}$ into an ultrametric $U_{2n}$. Note that
$U_{2n}$ is dominating: for all $x$ and $y$, $d_{X}(x,y)\le d_{U_{2n}}(f_{2n}(x),f_{2n}(y))$,
has small expected distortion: for all $x$ and $y$, $\mathbb{E}_{(f_{2n},U_{2n})}\left[d_{U_{2n}}(f_{2n}(x),f_{2n}(y))\right]\le O(\ddim\cdot\log\Phi_{2n})\cdot d_{X}(x,y)$,
and that $(f_{2n},U_{2n})$ is an extension of $(f_{2n-2},U_{2n-2})$.
We will run the algorithm from \Cref{lem:ultrametricOptimal} or \Cref{lem:2HSTAppprox} on the evolving ultrametric. That is, in the $n$-th step we will feed
it with $f_{2n}(x_{2n-1}),f_{2n}(x_{2n})$.
Let ${\cal M}_{U_{2n}}=\left\{ \left\{ y_{i},z_{i}\right\} \right\} _{i=1}^{n}$
be the matching produced by the algorithm of \Cref{lem:ultrametricOptimal} or \Cref{lem:2HSTAppprox} for the ultrametric $U_{2n}$. This matching is inward or HP-inward w.r.t.\ $U_{2n}$ by \Cref{lem:stability} and \Cref{lem:2HSTAppprox}.
We will maintain a matching for the first $2n$ points in the metric space $(X,d_{X})$: ${\cal M}_{2n}=\left\{ \left\{ f_{2n}^{-1}(y_{i}),f_{2n}^{-1}(z_{i})\right\} \right\} _{i=1}^{n}$.
Clearly, it is a perfect matching, and in addition, uses at most the same
amount of recourse as the algorithm on the ultrametric.
In \Cref{lem:ultrametricOptimal} the recourse is bounded by the height of the ultrametric, which is $O(\log\Phi_{i})$ by \Cref{obs:DiameterHSTbound}; and in \Cref{lem:2HSTAppprox} the recourse is bounded by $O(\log^3 n)$ using a heavy-path decomposition of the associated 2-HST.

Let $\cost_{X}$ and $\cost_{U_{2n}}$ be weight functions
of the matchings in the metric space and ultrametric, respectively. It is simply
the sum of pairwise distances between the matched points. Next we bound
the competitive ratio.
Let ${\cal M}_{2n}^{{\opt}}=\left\{ \left\{ a_{i},b_{i}\right\} \right\} _{i=1}^{n}$
be a minimum-weight perfect matching for $x_{1},x_{2},\dots,x_{2n-1},x_{2n}$.
That is, ${\cal M}_{2n}^{{\opt}}$ minimizes $\cost_{X}\left({\cal M}_{2n}\right)$ over all perfect matchings $\mathcal{M}_{2n}$.
Consider the matching ${\cal M}_{2n}^{{\opt},f}=\left\{ \left\{ f_{2n}(a_{i}),f_{2n}(b_{i})\right\} \right\} _{i=1}^{n}$ in the ultrametric.
Its expected weight is bounded by
\begin{align*}
\mathbb{E}_{(f_{2n},U_{2n})}\left[\cost_{U_{2n}}({\cal M}_{2n}^{{\opt},f})\right] & =\sum_{i=1}^{n}\mathbb{E}_{(f_{2n},U_{2n})}\left[d_{U_{2n}}(f_{2n}(a_{i}),f_{2n}(b_{i}))\right]\\
 & \le O(\ddim\cdot\log\Phi_{2n})\cdot\sum_{i=1}^{n}d_{X}(a_{i},b_{i})\\
 & =O(\ddim\cdot\log\Phi_{2n})\cdot\cost_{X}\left({\cal M}_{2n}^{{\opt}}\right)\,.
\end{align*}
The matching produced by \Cref{lem:ultrametricOptimal} has minimum
weight in the ultrametric, and the matching produced by \Cref{lem:2HSTAppprox} is a 2-approximation of the minimum-weight matching. In both cases,
\[
\cost_{U_{2n}}\left({\cal M}_{U_{2n}}\right)
\le 2\cdot \cost_{U_{2n}}\left({\cal M}_{2n}^{{\opt},f}\right) .
\]
We conclude that
\begin{align*}
\mathbb{E}_{(f_{2n},U_{2n})}\left[\cost_{X}\left({\cal M}_{U_{2n}}\right)\right] & =\mathbb{E}_{(f_{2n},U_{2n})}\left[\sum_{i=1}^{n}d_{X}(f_{2n}^{-1}(y_{i}),f_{2n}^{-1}(z_{i}))\right]\\
 & \le\mathbb{E}_{(f_{2n},U_{2n})}\left[\sum_{i=1}^{n}d_{U_{2n}}(y_{i},z_{i})\right]\\
 &=\mathbb{E}_{(f_{2n},U_{2n})}\left[\cost_{U_{2n}}\left({\cal M}_{U_{2n}}\right)\right]\\
 & \le\mathbb{E}_{(f_{2n},U_{2n})}\left[ 2\cdot \cost_{U_{2n}}\left({\cal M}_{2n}^{{\opt},f}\right)\right]\\
 &\le O(\ddim\cdot\log\Phi_{2n})\cdot \cost_{X}\left({\cal M}_{2n}^{{\opt}}\right)\,. \qedhere
\end{align*}
\end{proof}

Following the exact same lines (while replacing \Cref{thm:DoublingOnlineEmbedding} with \Cref{thm:EuclidineOnlineEmbeddingIntoHST}), we conclude:

\begin{theorem}\label{thm:ObliviousEuclideanMetric}
	There is a randomized algorithm such that, given points $x_1,\dots,x_{2n}\in\R^d$ revealed by an oblivious adversary in an online fashion with aspect ratio $\Phi=\frac{\max_{i,j}\|x_i-x_j\|_2}{\min_{i,j}\|x_i-x_j\|_2}$ (unknown in advance), maintains a perfect matching of expected competitive ratio $O(\sqrt{d}\cdot\log\Phi)$ with recourse $O(\log\Phi)$.
	Alternatively, the recourse can be bounded by $O(\log^3 n)$.
\end{theorem}

\section{Lower Bounds for Competitive Ratio and Recourse}
\label{sec:LB}

\subsection{One Recourse per Point Pair is Not Enough}

It is clear that without any recourse, the competitive ratio for the online minimum-weight perfect matching is unbounded (see an example in \Cref{sec:intro}).
Can we bound the competitive ratio if we allow one recourse per point pair? That is, for each new point pair, we allow the algorithm to delete one edge when it updates a perfect matching.
For online MST, for example, Gu et al.~\cite{GuG016} achieve a competitive ratio of $\Omega(\frac{1}{k})$ with one recourse for every $k$ new points.

\begin{proposition}\label{pp:one}
Given a single recourse per vertex pair, there is no competitive online algorithm.
This already holds for a sequence of $8$ points in $\R$, even if the sequence is known in advance.
\end{proposition}

\begin{figure}[htbp]
\begin{center}
		\includegraphics[width=0.75\textwidth]{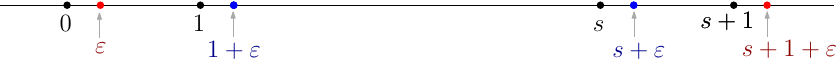}
\end{center}
\caption{A sequence of eight points: $0$, $1$, $s$, $s+1$, followed by $1+\epsilon$, $s+\epsilon$, and then $\epsilon$ $s+1+\epsilon$.\label{fig:OnlioneMatchingLB}}
\end{figure}

\begin{proof}
Suppose for the sake of contradiction that there is an online algorithm with one recourse pair vertex pair, and competitive ratio $k\ge 1$.  Fix $s=k+1$ and $\epsilon<\frac{1}{4k}$.
The first four points presented by the adversary presents are $0,1,s,s+1$; see \Cref{fig:OnlioneMatchingLB}. The algorithm must then have the perfect matching $M_2=\{\{0,1\},\{s,s+1\}\}$ of weight $2$, or else the weight of the matching would be at least $2s>2k$, a contradiction.
Next, the adversary presents the points pair $1+\epsilon,s+\epsilon$.
The algorithm updates its matching. However, as the algorithm is allowed to delete only one edge, the new matching can contain at most one edge from $\{1,1+\epsilon\},\{s,s+\epsilon\}$.
Finally, the adversary introduces the final pair of points: $\epsilon,s+1+\epsilon$. As the algorithm is allowed to add only two new edges, the resulting matching $M_4$ can contain at most three edges out of the optimum matching $\{\{0,\eps\},\{1,1+\eps\},\{s,s+\eps\},\{s+1,s+1+\eps\}\}$. In particular, $M_4$ will contain at least one other edge of weight at least $1-\eps$. It follows that the weight of $M_4$ is at least $1-\eps+3\cdot\eps>1$, while the optimum matching has weight $4\eps$. We conclude that the competitive ratio is greater than $\frac{1}{4\eps}>k$, a contradiction.
\end{proof}

\subsection{Lower Bounds for Competitive Ratio and Lightness: Proof of \Cref{thm:LowerBound}}

In this section, we prove \Cref{thm:LowerBound}.
we first present a strategy for an adaptive adversary that runs in $O(\log_r n)$ stages (\Cref{thm:LB}). Afterward, we will strengthen it so that it can be used by an oblivious adversary as well.

\begin{restatable}[]{lemma}{LowerBoundMatching}
\label{thm:LB}
Let $r\geq 2$ be an integer, and let $\alg$ be a deterministic online perfect matching algorithm with recourse $r$ for each arriving point pair. Then for every integer $n\geq 10\, r$, an adaptive adversary can construct a sequence $S_n$ of $2n$ points in $\mathbb{R}$ such that $\alg$ returns a matching of weight
$\Omega\left( \frac{\diam(S_n)\,\log n}{r\, \log r}\right)$.
\end{restatable}

\begin{proof}
Let $q=10\, r$. Fix an integer $k$ such that $q^k\leq n<q^{k+1}$. Note that $k=\Theta(\log_q n)=\Theta(\log n/\log q)=\Theta(\log n/\log r)$.
An adaptive adversary will present at most $n$ points to $\alg$.
For $i=0,1,\ldots , k$, let $Q_i=\{j\cdot q^i: j=1,\ldots ,q^{k-i}\}$, that is, $Q_i$ contains positive multiples of $q^i$ up to $q^k=q^{k-i}\cdot q^i$. Note that $Q_0\supset Q_1\supset \ldots \supset Q_k$
and  $|Q_i|=q^{k-i}$ for all $i$.

\paragraph{General strategy.}
The adversary proceeds in $k$ \emph{rounds}: In round 0, it presents the points $Q_0\setminus Q_1$ in an arbitrary order. In round $i\in \{1,\ldots , k-1\}$, in general, it will present either no new points or the points in $Q_i\setminus Q_{i+1}$, as described below.
The diameter of the point set is at most $\Delta=q^k-1=\Theta(q^k)$ at all times.
Note also that all new points that arrive \emph{after} round $i$ are in $Q_{i+1}$. That is, at most $|Q_{i+1}|=q^{k-(i+1)} = \frac{q^{k-i}}{10r}$ points arrive after round $i$, and so $\alg$ can delete at most $\frac{r}{2}\cdot |Q_{i+1}|=\frac{q^{k-i}}{20}$ edges after round $i$.

For the analysis, let $M_i$ denote the perfect matching maintained by $\alg$ at the end of round $i$; and let $M=M_{k-1}$ be the matching at the end of the process.
We will specify edge sets $E(0),\ldots , E(k-1)$ with the following properties:
\begin{itemize}
    \item $E(i)\subseteq M_i$,
    \item $|E(i)|\geq \frac{q^{k-i}}{10}$,
    \item every edge $e\in E(i)$ has weight $\cost(e)\geq \frac{q^i}{4r}$.
    \end{itemize}
As noted above, $\alg$ can delete at most $\frac{q^{k-i}}{20}$ edges \emph{after} round $i$. Consequently, for all $i\in \{0,\ldots , k-1\}$, we can find subsets $E'(i)\subset E(i)$ of size $|E'(i)|=\frac{q^{k-i}}{20}$ such that the edges in $E'(i)$ survive until the end of the process and are in $M$. Note that $E'(0),\ldots, E'(k-1)$ need not be disjoint. We can choose disjoint subsets as follows: For $i=0,\ldots , k-2$, let $E''(i)=E'(i)\setminus \left(\bigcup_{i<j} E'(j)\right)$, and $E''(k-1)=E'(k-1)$.
Since the cardinalities of the sets $E'(i)$ decay exponentially,
then $|E''(i)|> \frac{4}{5}\, |E'(i)|> \frac{1}{25}\, q^{k-i}$.
Since $E''(0),\ldots, E''(k-1)$ are disjoint subsets of $M$, then
\begin{equation}\label{eq:LB}
    \cost(M)\geq \sum_{i=0}^{k-1}\cost(E''(i))
    \geq k\cdot \frac{q^{k-i}}{25} \cdot \frac{q^i}{4r}
    = \Omega(k\cdot q^{k-1})
    = \Omega\left(\frac{\log n}{\log r}\cdot \frac{\Delta}{r}\right)
    = \Omega\left( \frac{\diam(S_n)\,\log n}{r\, \log r}\right).
\end{equation}

It remains to present the adversary's strategy in rounds $i=1,\ldots , k-1$, and choose sets $E(0), \ldots , E(k-1)$ that satisfy the three properties above.

\paragraph{Round $i=0$.}
At the end of round 0, algorithm $\alg$ has a perfect matching $M_0$ on the current point set $Q_0\setminus Q_1$. The size of this matching is $|M_0|=\frac12 (q^k-q^{k-1})=\frac{1}{2}\left(1-\frac{1}{q}\right)\cdot q^{k}> \frac{q^{k}}{32}$. Each edge in $M_0$ has a weight of at least 1. Let $E(0)=M_0$. Clearly, $E(0)$ satisfies all three required properties.

\paragraph{Strategy in round $i=1,\ldots ,k-2$.}
As noted above, the adversary will present either no new points or the points in $Q_i\setminus Q_{i+1}$, based on $M_{i-1}$. We distinguish between two cases:

\smallskip\noindent \emph{Case~1: $M_{i-1}$ contains at least  $\frac{q^{k-i}}{10}$ edges of length at least $\frac{q^i}{4r}$.}
Then the adversary does not present any new points in round $i$. Therefore, there are no recourse, and $M_i=M_{i-1}$. Let $E(i)$ be the set of all edges in $M_i$ of length at least $\frac{q^i}{4r}$. It is clear that $E(i)$ satisfies all three required properties.

\smallskip\noindent \emph{Case~2: $M_{i-1}$ contains fewer than $\frac{q^{k-i}}{10}$ edges of length at least $\frac{q^i}{4r}$.} Then the adversary presents the points in $Q_i\setminus Q_{i+1}$ in an arbitrary order, and $\alg$ computes $M_i$. It remains to specify $E(i)$ and show that it satisfies the three required properties.

\begin{figure}[h]
\begin{center}
		\includegraphics[width=0.95\textwidth]{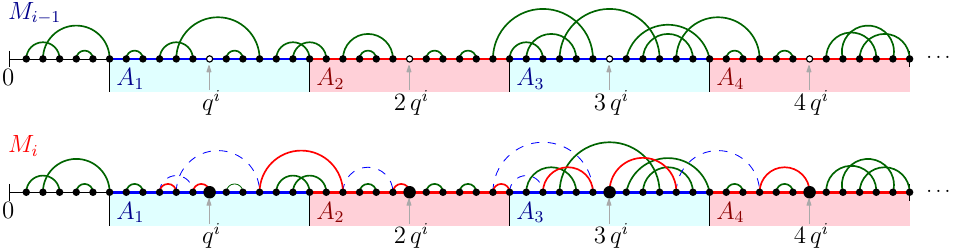}
\end{center}
\caption{A schematic figure of matchings $M_{i-1}$ and $M_i$ with $r=2$.  In round $i$, dashed edges are deleted ($M_{i-1}\setminus M_i$), and red edges are inserted ($M_i\setminus M_{i-1}$). \label{fig:IntervalsLB}}
\end{figure}

In round $i$, $\alg$ may delete up to $\frac{r}{2}\, (q^{k-i}-q^{k-(i+1)})$ edges, but it has to cover the endpoints of all deleted edges by the end of round $i$---we say that these points are \emph{re-matched} in round $i$.
For each points $p\in Q_i\setminus Q_{i+1}$, we define the open intervals $A_p=\left(p - \frac{q^i}{2}, p+\frac{q^i}{2}\right)$, centered at the points $p$; see \Cref{fig:IntervalsLB}. Note that these intervals are pairwise disjoint. We call an interval $A_p$ \emph{good} if it contains less than $2r$  re-matched points \emph{and} it does not contain the endpoint of any edge in $M_{i-1}$ that has length $q^i/(4r)$ or more; otherwise the interval $A_p$ is \emph{bad}. By the assumption in Case~2, fewer than $2\cdot \frac{1}{10}\,q^{k-i} = \frac{1}{5}\,q^{k-i}$ intervals $A_p$ contain an endpoint of some edge of $M_{i-1}$ of length at least $q^i/(4r)$.
Since $q^{k-i}-q^{k-(i+1)}$ new points arrive in round $i$, then this is the number of  intervals $A_p$, and there are at most $r\cdot (q^{k-i}-q^{k-(i+1)})$ re-matched points. On average, an interval $A_p$ contains at most $r$ re-matched points. By Markov's inequality, at most $\frac12(q^{k-i}-q^{k-(i+1)})$ intervals $A_p$ contain $2r$ or more re-matched points. By the union bound, fewer than  $\frac{1}{5}\,q^{k-i}+\frac{1}{2}(q^{k-i}-q^{k-(i+1)}) =\frac{7}{10} q^{k-i}-\frac12 q^{k-(i+1)}$ intervals are bad; hence more than $\frac{3}{10}\,q^{k-i} -\frac12 q^{k-(i+1)} = \left(\frac{3}{10}-\frac{1}{20r}\right)q^{k-i}>\frac{q^{k-i}}{5}$ intervals are good.

The symmetric difference $M_{i-1}\triangle M_i$
is the union of alternating cycles and paths: Each path connects two points inserted in round $i$ (i.e., centers of intervals). All interior vertices of a path are points that are re-matched in round $i$; and the edges along each path alternate between $M_i\setminus M_{i-1}$ (i.e., edges inserted) and $M_{i-1}\setminus M_i$ (i.e., edges deleted).

Let $A_p$ be a good interval centered at $p$, and consider the path in $M_{i-1}\triangle M_i$ that starts from $p$. This path ends at the center $p'$ of another interval $A_{p'}$, $p\neq p'$. Consider the initial part of this path, from $p$ to the first vertex outside of $A_p$: It has at most $2r$ edges, since $A_p$ is good, and its length is at least $\frac12\, \diam(A_p)=\frac{q^i}{2}$; and so it must have an edge of length at least $\frac{q^i}{4r}$. This edge cannot be in $M_{i-1}\setminus M_i$ since $A_p$ is good.
In summary, for each good interval $A_p$, the initial part of the path starting at $p$ contains an edge in $M_i\setminus M_{i-1}$ of length at least $\frac{q^i}{4r}$. Each such edge belongs to at most two good intervals. Summation over all good intervals yields at least $\frac12\cdot \frac{q^{k-i}}{5}= \frac{q^{k-i}}{10}$ such edges. Let $E(0)$ be the set of all such edges. It is now clear that $E(i)$ satisfies all three required properties. This completes the analysis in Case~2 and, hence, the proof of the theorem.
\end{proof}

Next, we extend the strategy of \Cref{thm:LB}, so that it could be used by an oblivious adversary. The idea is to perform each stage according to a random bit string. \Cref{thm:LB2} implies, already for $n$ points on a line, that any $O(1)$-competitive algorithm (against an oblivious adversary) requires recourse $\Omega(\log n/\log \log n)$; and any algorithm with recourse $O(1)$ must have a competitive ratio (resp., lightness) $\Omega(\log n)$.

\begin{restatable}[]{lemma}{LowerBoundTheoremTwo}
\label{thm:LB2}
Let $r\geq 2$ be an integer; and let $\alg$ be an online perfect matching algorithm with recourse $r$ for each arriving point pair. Then for every $n\geq 10\, r$, an oblivious adversary can construct a collection $\mathcal{S}$ of $n^{O(1/\log r)}$ sequences, each of $2n$ points in $\mathbb{R}$ such that
for a sequence $S$ selected uniformly at random from $\mathcal{S}$,
$\alg$ returns a matching of the expected weight
$\Omega\left( \frac{\diam(S)\,\log n}{r\, \log r}\right)$.
\end{restatable}

\begin{proof}
As in the proof of \Cref{thm:LB}, let $q=10\, r$ and fix $k$ such that $q^k\leq n<q^{k+1}$. Note that $k=\Theta(\log n/\log r)$. Let $x=(x_1,\ldots , x_{k-1})\in \{0,1\}^{k-1}$ be vector of $k$ bits. Each vector $x$ encodes a sequence $S(x)$ of points in $\mathbb{R}$ in $k-1$ rounds as follows: $S(x)$ starts with the elements of $Q_0\setminus Q_1$ in arbitrary order; and in round $i=1,\ldots , k-1$, we append the elements of $A_i\setminus A_{i+1}$ iff $x_i=1$. Overall, we construct a collection $\mathcal{S}$ of  $2^{k-1}=n^{\Theta(1/\log r)}$ sequences.

Let $x\in \{0,1\}^{k-1}$ be a uniformly random vector. Denote by $M_i$ the perfect matching computed by $\alg$ after round $i$, for $i=0,\ldots , k-1$; with $M=M_{k-1}$ the matching at the end.

We say that a round $i\in \{0,\ldots , k-1\}$ is \emph{heavy} if we can find an edge set $E(i)\subset M_i$ such that $|E(i)|\geq \frac{q^{k-i}}{10}$ and every edge $e\in E(i)$ has weight $\cost(e) \geq \frac{q_i}{4r}$. We show below that the expected number of heavy rounds is $\Omega(k)$. Similarly to the proof of \Cref{thm:LB}, from all heavy rounds, we can choose disjoint subsets $E''(i)\subset E(i)\cap M$ of weight $\cost(E''(i))\geq \Omega(\diam(S(x))/r)$. Summation over the heavy rounds yields
$\mathbb{E} [\cost(M)]
\geq \Omega\left(k\cdot \frac{\diam(S(x))}{r}\right)
=\Omega\left( \frac{\diam(S_n)\,\log n}{r\, \log r}\right)$.

It remains to bound the expected number of heavy rounds. In round~0, we can choose $E(0)=M_0$, so this round is always heavy. Consider rounds $i=1,\ldots , k-1$. To find edge sets $E(i)\subset M_i$, we consider the two cases from the proof of \Cref{thm:LB}:

Case 1: $M_{i-1}$ contains at least $\frac{q^{k-i}}{10}$ edges of length at least $\frac{q^i}{4r}$. If $x_i=0$, then the sequence $S(x)$ does not contain the points in $Q_i\setminus Q_{i+1}$, and we find a subset $E(i)\subset M_{i-1}=M_i$.

Case 2: $M_{i-1}$ contains fewer than $\frac{q^{k-i}}{10}$ edges of length at least $\frac{q^i}{4r}$. If $x_i=1$, then the sequence $S(x)$ contains all the points in $Q_i\setminus Q_{i+1}$, and we find a subset $E(i)\subset M_i\setminus M_{i-1}$.

In both cases, round $i$ is heavy with probability at least $\frac12$, as $\mathrm{Pr}(x_i=0)=\frac12$ and $\mathrm{Pr}(x_i=1)=\frac12$. By linearity of expectation, at least half of the rounds $i=1,\ldots ,k-1$ are heavy, so the expected number of heavy rounds is $\Omega(k)$, as required.
\end{proof}

\Cref{thm:LowerBound} is now an immediate corollary of \Cref{thm:LB2} (restated for convenience).
\MainNewLowerBoundTheorem*

\begin{proof}
For a multiset $S_n$ of $2n$ points in the plane, the MST is a path of weight $\diam(S_n)$,
and the minimum weight of a perfect matching is trivially bounded by $\opt\leq \diam(S_n)$.
\end{proof}

\Cref{thm:LowerBound} shows that with recourse $r=O(1)$, the lightness and competitive ratio of any online algorithm against oblivious adversary is $\Omega(\log n)$; and an $O(1)$-competitive algorithm would require recourse at least  $r=\Omega(\log n/\log \log n)$.

\section{Conclusions}
\label{sec:con}

We introduced the problem of online minimum-weight perfect matchings for a sequence of $2n$ points in a metric space. In contrast to the online MST and TSP, where $O(1)$-competitive algorithms with recourse $O(1)$ are known, we showed that the competitive ratio \emph{or} the recourse must be at least polylogarithmic. We also devised polylogarithmic upper bounds for the competitive ratio and lightness, resp., against oblivious and adaptive adversaries, using polylogarithmic recourse. Closing the gaps between the upper and lower bounds are obvious open problems, both in general metrics and in special cases such as Euclidean spaces or in ultrametrics. We  highlight a few specific open problems.

\begin{enumerate}
    \item We have shown (\Cref{pp:one}) that recourse $r=1$ per point pair is not enough for a competitive algorithm. Is there a competitive algorithm with recourse $r=O(1)$?\label{ques:ConstatRecourse}
    \item What is the minimum recourse for an $O(1)$-competitive algorithm against an adaptive (resp. oblivious) adversary?
    Our \Cref{thm:LB2} gives a lower bound of $\Omega(\log n/\log \log n)$, but we are unaware of any nontrivial upper bound.
    \item Are the optimal trade-offs different for competitive ratio and for lightness? Does it take more recourse to maintain a perfect matching of weight $O(\varrho\cdot \opt)$ than one of weight $O(\varrho \cdot \mst)$ for any ratio $\varrho\geq 1$? Our lower bound (\Cref{thm:LB2}) do not distinguish between competitive ratio and lightness, but in general the ratio $\frac{\mst}{\opt}$ is unbounded.
    \item For maintaining a minimum-weight near-perfect matching on a fully dynamic point set (with insertions \emph{and} deletions), what are the best possible trade-offs between the approximation ratio (or lightness) and the number of changes in the matching? Our 1D data structure  (\Cref{thm:1D}) can handle both point insertions and deletions, and maintains a matching of lightness $O(\log n)$, but the problem remains open in other metric spaces.

    \item 
    Consider the adversarial model, sometimes called \emph{prefix-model}, which is weaker than the oblivious model. Here, the metric space $(X,\delta)$ and the entire sequence of arriving points $x_1,x_2,\dots,x_{2n}$ are known in advance. The goal is to construct a sequence of $n$ matchings $M_1,M_2,\dots,M_n$, where $M_i$ is a perfect matching for $S_i=\{x_1,\dots,x_{2i}\}$, while minimizing the competitive ratio  $\max_i\frac{\cost(M_i)}{\opt(S_i)}$, and the maximum recourse $\max_i|M_i\setminus M_{i+1}|$.
    Note that our lower bound from \Cref{pp:one} holds in this model, as well, and thus a single recourse is not enough for a competitive algorithm.
    Clearly, we can use the same algorithm we used for oblivious routing (\Cref{thm:ObliviousGeneralMetric}), while using the embedding of \cite{FRT04} instead of \cite{BFU20}, thus improving the competitive ratio to $O(\log n)$.
    Is it possible to further improve on either the competitive ratio or the recourse?

    \item Online Decomposition of Minor-Free Graphs. Consider the shortest path metric of a fixed minor-free graph (e.g., planar graphs). Such metrics enjoy a good padded decomposition scheme \cite{KPR93,AGGNT19,Fil19Approx}, and therefore also better embeddings into Euclidean space, compared to general metric spaces \cite{Rao99}.
    However, no online version of such a decomposition is known, even for subfamilies such as planar graphs, and bounded treewidth/pathwidth graphs.
    Note that, given such a decomposition, the framework in this paper will imply good online metric embeddings into both HST and Euclidean spaces. Thus we find the construction of such decompositions to be a fascinating open problem.
    One might be tempted to think that the KPR \cite{KPR93} decomposition can be easily implemented in an online fashion.  This would make sense as the choice of a center in every ring in the decomposition is arbitrary. However, unfortunately, only the first rings/chops are constructed w.r.t.\ the original metric. All the other steps are performed w.r.t.\ induced subgraphs. It is unclear how to obtain distances or even an approximation in an induced subgraph defined by a chop in an online fashion.

    \item Krauthgamer \etal \cite{KLMN04} showed that every metric space with doubling dimension $\ddim$ can be embedded into Euclidean space with distortion $O(\sqrt{\ddim\cdot\log n})$; and this bound is known to be tight \cite{JLM11}.
    In contrast, our online embedding has distortion $O(\ddim\cdot\sqrt{\log\Phi})$. In \Cref{thm:LBEuclideanDoubling} we showed that the dependence on the aspect ratio $\Phi$ is tight. It would be interesting to see whether the dependence on the doubling dimension can be improved to $\sqrt{\ddim}$.
   \end{enumerate}

\section*{Acknowledgments}
In a previous version of this paper, \Cref{thm:LBEuclideanDoubling} was stated as a lower bound w.r.t.\ dominating stochastic embeddings. We are grateful to Christian Coester for pointing out an error in that proof. In this version of the paper this theorem is stated w.r.t.\ non-expansive stochastic embeddings (as well as deterministic embeddings).
It is currently open whether a dominating stochastic embedding into $\ell_2$ with expected distortion better than that in \Cref{thm:DoublingOnlineEmbeddingtoEuclidean} exist.

\bibliographystyle{alphaurl}
\bibliography{matching,OnlineDoubling}

\appendix

\section{Network Design Problems}\label{sec:Network-Design-Problems}
Bartal \etal \cite{BFU20} used their dominating stochastic online embedding to design competitive online algorithms for network design problems. Surprisingly, they showed that in many cases, the dependence on the aspect ratio can be avoided.
Specifically, Bartal \etal \cite{BFU20} showed that in cases where the problem admits a min-operator \cite{ABM93}, and has $\alpha$-competitive solutions on ultrametrics, one can obtain an algorithm with competitive ratio $O\left(\alpha\cdot\log k\cdot\min\left\{ \log(k\alpha\lambda_{\rho}),\log k\cdot\log(k\alpha)\right\} \right)$ against an oblivious adversary, where $\lambda_\rho$ is the level of subadditivity of the target function, and $k$ is the number of points. We refer to \cite{BFU20} for the definition of the network design problem (and the parameter $\lambda_\rho$).
In particular, it follows \cite{Umboh23} from the proof in \cite{BFU20}, that a similar embedding into ultrametric with distortion $\beta\cdot\log\Phi$ (coming from paying $\beta$ in $\log\Phi$ different scales) will yield an algorithm with competitive ratio $O\left(\alpha\cdot\beta\cdot\min\left\{ \log(k\alpha\lambda_{\rho}),\log k\cdot\log(k\alpha)\right\} \right)$ against an oblivious adversary.
\begin{corollary}\label{cor:NetworkDesign}
	Consider an abstract network design problem. If it admits a min operator and
	if there exists an algorithm that is $\alpha$-competitive on instances where the input graph is an ultrametric,
	then there exists a randomized algorithm that, on every instance induces a metric space of doubling dimension $\ddim$, has competitive ratio $O\left(\alpha\cdot\ddim\cdot\min\left\{ \log(k\alpha\lambda_{\rho}),\log k\cdot\log(k\alpha)\right\} \right)$ against an oblivious adversary.
\end{corollary}
\sloppy Note that if the points in the network design problem are coming from a $d$-dimensional Euclidean space, the competitive ratio can be further improved to    $O\left(\alpha\cdot\sqrt{d}\cdot\min\left\{ \log(k\alpha\lambda_{\rho}),\log k\cdot\log(k\alpha)\right\} \right)$.
Bartal \etal \cite{BFU20} showed several applications for their meta-theorem.  \Cref{cor:NetworkDesign} implies improvements in all of them for the case where the input metric has a doubling dimension $o(\log n)$.
One example is the Subadditive Constrained Forest problem \cite{GW95}, where  \Cref{cor:NetworkDesign} improves the competitive ratio from $O(\log^2 k)$ to $O(\ddim\cdot\log k)$, (or $O(\sqrt{d}\cdot\log k)$ for points in Euclidean $d$-space).

\section{Parallelogram inequality - proof of inequality (\ref{eq:parallelogramExtended})}\label{sec:ParallelogramProof}
In this appendix we will prove inequality (\ref{eq:parallelogramExtended}).
\begin{claim}\label{clm:parallelogramExtended}
	For every six vectors $\vec{s},\vec{t},\vec{a},\vec{b},\vec{c},\vec{d}\in\ell_{2}$, it holds that
	\begin{equation*}
		\|\vec{s}-\vec{t}\|_{2}^{2}+\|\vec{b}-\vec{d}\|_{2}^{2}\le4\cdot\left(\|\vec{a}-\vec{s}\|_{2}^{2}+\|\vec{c}-\vec{t}\|_{2}^{2}\right)+2\cdot\left(\|\vec{a}-\vec{b}\|_{2}^{2}+\|\vec{b}-\vec{c}\|_{2}^{2}+\|\vec{c}-\vec{d}\|_{2}^{2}+\|\vec{d}-\vec{a}\|_{2}^{2}\right)~.
	\end{equation*}	
\end{claim}
\begin{center}
	\includegraphics[width=.45\textwidth]{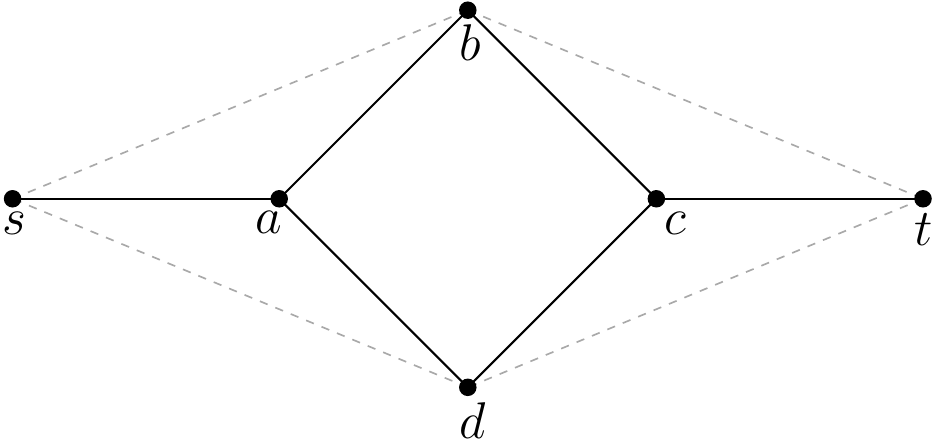}
\end{center}

\begin{proof}
First, recall the Parallelogram inequality (or Euler's quadrilateral theorem), that states that for every four vectors $\vec{a},\vec{b},\vec{c},\vec{d}\in\ell_{2}$, it holds that
\begin{equation}
	\|\vec{a}-\vec{c}\|_{2}^{2}+\|\vec{b}-\vec{d}\|_{2}^{2}\le\|\vec{a}-\vec{b}\|_{2}^{2}+\|\vec{b}-\vec{c}\|_{2}^{2}+\|\vec{c}-\vec{d}\|_{2}^{2}+\|\vec{d}-\vec{a}\|_{2}^{2}~\label{eq:parallelogram}
\end{equation}
(The inequality holds with equality holds if and only if $\vec{a},\vec{b},\vec{c},\vec{d}$ form a parallelogram.)
Next, note that for every $3$ vectors $\vec{x},\vec{y},\vec{z}\in\ell_{2}$ it holds that
\begin{align}
	\|\vec{x}-\vec{y}\|_{2}^{2} & \le\left(\|\vec{x}-\vec{z}\|_{2}+\|\vec{z}-\vec{y}\|_{2}\right)^{2}\nonumber \\
	& =\|\vec{x}-\vec{z}\|_{2}^{2}+\|\vec{z}-\vec{y}\|_{2}^{2}+2\cdot\|\vec{x}-\vec{z}\|_{2}\cdot\|\vec{z}-\vec{y}\|_{2}\nonumber \\
	& \le2\|\vec{x}-\vec{z}\|_{2}^{2}+2\|\vec{z}-\vec{y}\|_{2}^{2}~,\label{eq:QuadraticPythagoras}
\end{align}
where we used the triangle inequality followed by the Cauchy-Schwarz inequality
(alternatively, we use $0\le\left(\|\vec{x}-\vec{z}\|_{2}-\|\vec{z}-\vec{y}\|_{2}\right)^{2}=\|\vec{x}-\vec{z}\|_{2}^{2}+\|\vec{z}-\vec{y}\|_{2}^{2}-2\cdot\|\vec{x}-\vec{z}\|_{2}\cdot\|\vec{z}-\vec{y}\|_{2}$).
Inequality (\ref{eq:parallelogramExtended}), now follows:
\begin{align*}
	\|\vec{s}-\vec{t}\|_{2}^{2}+\|\vec{b}-\vec{d}\|_{2}^{2} & \le\|\vec{s}-\vec{b}\|_{2}^{2}+\|\vec{b}-\vec{t}\|_{2}^{2}+\|\vec{t}-\vec{d}\|_{2}^{2}+\|\vec{d}-\vec{s}\|_{2}^{2}\\
	& \le4\cdot\left(\|\vec{a}-\vec{s}\|_{2}^{2}+\|\vec{c}-\vec{t}\|_{2}^{2}\right)+2\cdot\left(\|\vec{a}-\vec{b}\|_{2}^{2}+\|\vec{b}-\vec{c}\|_{2}^{2}+\|\vec{c}-\vec{d}\|_{2}^{2}+\|\vec{d}-\vec{a}\|_{2}^{2}\right)~,
\end{align*}
where we first applied the Parallelogram inequality (\ref{eq:parallelogram}) w.r.t.\ the vectors $\vec{s},\vec{b},\vec{t},\vec{d}\in\ell_{2}$, and then applied inequality (\ref{eq:QuadraticPythagoras}) w.r.t.\ the four triplets: $(\vec{s},\vec{a},\vec{b})$, $(\vec{b},\vec{c},\vec{t})$,  $(\vec{t},\vec{c},\vec{d})$, and $(\vec{d},\vec{a},\vec{s})$.
\end{proof}	

\end{document}